\documentclass[12pt]{article}
\usepackage{amsmath}
\def\vpo{\vp\so}
\def\vpw{\vp\sw}
\def\uxx{\uu{xx}}
\def\uxo{\uu x}
\def\ruh{\rho\uu h}
\def\ioh{\h{(0,\f12]}}

\newcommand\izx[1]{\int_0^x \, #1 \, dy}
\newcommand\izy[1]{\int_0^y \, #1 \, du}
\newcommand\intvl[1]{\h{(-#1,#1)}}
\newcommand\cintvl[1]{\h{[-#1,#1]}}
\newcommand\ea[1]{\begin{eqnarray} #1 \end{eqnarray}}
\newtheorem{claim}{Claim}
\newtheorem{theorem}{Theorem}
\newtheorem{conjecture}{Conjecture}
\newcommand\erf[1]{Eq.~(\ref{#1})}
\newcommand\eqnref[1]{\erf{#1}}

\newcommand\lsc[3]{\e\tu{#1 #2 #3}}

\def\ozch{\left(\ltd0,\f12\right]\rtd}
\def\slt{\, < \, }

\def\vn{\vp\sps n}
\def\vnm{\vp\sps{n-1}}

\def\mv{m\uu V}

\usepackage{xcolor}
\usepackage{ams}
\usepackage{amsmath,amssymb,latexsym}
\usepackage{mathtools}
\usepackage[normalem]{ulem}
\usepackage{float}
\usepackage{graphicx}
\usepackage[utf8]{inputenc}
\usepackage{soul}
\usepackage{tikz}
\newtheorem{lemma}{Lemma}
\newtheorem{corollary}{Corollary}

\def\ds{\displaystyle}
\def\bls{\baselineskip} 
\def\beq{\begin{equation}}
\def\eeq{\end{equation}}
\def\bi{\begin{itemize}}
\def\be{\begin{enumerate}}
\def\bc{\begin{center}}
\def\bt{\begin{tabular}}
\def\bd{\begin{description}}
\def\bq{\begin{quotation}}
\def\ba{\begin{array}}
\def\beas{\begin{eqnarray*}}
\def\bea{\begin{eqnarray}}
\def\bv{\begin{verbatim}}
\def\ev{\end{verbatim}}
\def\eea{\end{eqnarray}}
\def\eeas{\end{eqnarray*}}
\def\ea{\end{array}}
\def\eq{\end{quotation}}
\def\ed{\end{description}}
\def\ec{\end{center}}
\def\et{\end{tabular}}
\def\ee{\end{enumerate}}
\def\ei{\end{itemize}}
\def\ada{  \dfn }
\def\aea{ = }

\def\alea{  \leq  }

\def\ala{  <  }

\def\a{\h{\alpha}}

\def\g{\h{\gamma}}

\def\de{\h{\delta}}
\def\De{\h{\Delta}}

\def\m{\h{\mu}}
\def\n{\h{\nu}}

\def\rh{\h{\rho}}
\def\e{\h{\epsilon}}

\def\vp{\h{\varphi}}

\def\qed{\hfill $\rule{8pt}{8pt}$}
\def\sz{\h{\uu 0}}
\def\so{\h{\uu 1}}
\def\sw{\h{\uu 2}}
\def\sr{\h{\uu 3}}

\def\sn{\h{\dn n}}
\def\si{\h{\dn i}}
\def\sj{\h{\dn j}}

\def\sb{\h{\uu b}}

\def\uo{\h{^ 1}}
\def\sq{\h{^ 2}}
\def\cu{\h{^ 3}}

\def\upn{\h{^ n}}
\def\ui{\h{^ i}}
\def\uj{\h{^ j}}
\def\ua{\h{^ a}}

\def\str{\h{^*}}
\def\upm{\h{^\mu}}
\def\dnm{\h{_\mu}}

\def\upn{\h{^\nu}}
\def\dnn{\h{_\nu}}

\def\tuh{\h{\tu{\f12}}}

\def\inv{\h{^{-1}}}
\def\dmn{\h{_{\m\n}}}
\def\umn{\h{^{\m\n}}}

\def\sk{\h{\dn k}}

\def\pmd{\ensuremath{\pl_\mu}}
\def\pmu{\ensuremath{\pl^\mu}}

\def\pnd{\ensuremath{\pl_\nu}}

\def\ra{\ensuremath{\rightarrow}}

\def\pl{\partial}

\def\smg{\sqrt{\rule{0in}{8pt}-g}}

 \def\sep{\vspace{12pt} \hrule width \hsize \kern 1mm \hrule width \hsize height 2pt \vspace{12pt} }

\def\clh{\h{{\cal H}}}

\def\cll{\h{{\cal L}}}

\def\clh{\h{{\cal H}}}

\def\deq{\h{\ = \ }}

\def\beass{\begin{small}\beas}
\def\eeass{\eeas\end{small}}
 \def\beasm{\begin{small}\bea}
\def\eeasm{\eea\end{small}}

\def\rtd{\right.}
\def\ltd{\left.}

\def\clh{\h{{\cal H}}}
\def\cll{\h{{\cal L}}}

\def\cle{\h{{\cal E}}}

\def\dfn{\h{:=}}

\def\tuh{^{^{(H)}}}

\def\mmi{^-}
\newcommand\h[1]{\ensuremath{#1}}

\newcommand\tu[1]{\h{^{#1}}}
\newcommand\dn[1]{\h{ _{#1}}}
\newcommand\f[2]{\h{\frac{#1}{#2}}}
\newcommand\brf[2]{\h{\lrr{\f{#1}{#2}}}}
\newcommand\bsf[2]{\h{\lrs{\f{#1}{#2}}}}

\newcommand\dsf[2]{\h{\ds \f{#1}{#2}}}
\definecolor{indigo}{RGB}{34,27,59}

\newcommand\eto[1]{\h{e^{\ds #1}}}

\newcommand\gd[2]{\h{g\dn{#1 #2}}}

\newcommand\sps[1]{\h{^{(#1)}}}

\newcommand\dbyd[2]{\h{\dsf{d #1}{d #2}}}

\newcommand\lrr[1]{\h{\left( #1 \right)}} %left-right-round
\newcommand\lrs[1]{\h{\left[ \, #1 \, \right]}}
\newcommand\lrc[1]{\h{\left\{ #1 \right\}}}
\newcommand\lrb[1]{\h{\left| #1 \right|}}
\newcommand\eas[1]{\beas #1 \eeas}
\newcommand\eae[1]{\bea #1 \eea}

\newcommand\lt[2]{\lim_{#1 \ra #2}}
\newcommand\ct[1]{(\ref{eqn#1})}
\newcommand\hs[1]{\h{\hspace{#1in}}}

\newcommand\ti[1]{\h{\tilde{#1}}}

\newcommand\uu[1]{_{_{#1}}}

\newcommand\leqn[1]{\label{eqn#1}}
\newcommand\crop[4]{}

\newcounter{bean}

\newcommand\scn[2]{\h{#1 \times 10^{#2}}}

\def\tuh{\tu{\f12}}

\newcommand\qn[2]{\index{#1}{\the\numexpr\value{bean} ) {\it #2}}}

\newcount\myloopcounter

\newcommand{\clone}[2][10]{%
  \myloopcounter0% initialize the loop counter
  \loop\ifnum\myloopcounter < #1 % Test if the loop counter is < #1
  #2%
  \advance\myloopcounter by 1 % 
  \repeat % start again
}
\newcommand\norm[1]{\left| #1 \right|}
\def\mx{m\uu X}

\newcommand\igwe[2]{\includegraphics[width=#1in]{#2}}
\newcommand\ighe[2]{\includegraphics[height=#1in]{#2}}
\newcommand\pot[1]{10^{#1}}
\newcommand\potm[1]{10^{-#1}}

\newcommand\eqr[1]{Eq.~(\ref{#1})}
\newcounter{fgrs}

\title{Static Solitons in an Expanding Universe}
\author{Nagabhushana Prabhu\\
Purdue University, West Lafayette, IN 47907}

\begin{document}
\maketitle
\begin{abstract}
We show, analytically, that a static sine-Gordon soliton cannot exist
 in 1 + 1 non-dynamical  de Sitter
spacetime  if    $\a\dfn (m/H)\sq < 2$, 
%MDPI: Please ensure all variables/values in the equation appear in the same format in the text (normal/italic/bold/subscript/superscript).  %Author: I believe I have done that.
where $m$ is the mass
parameter of the sine-Gordon theory and $H$ is the Hubble constant.  
Conversely, we also show that static 
sine-Gordon solitons exist  in 1 + 1 non-dynamical de Sitter spacetime
if $\a > 2$.
   The 
 above threshold is explained{---qualitatively and to within an O(1) 
 factor---using a heuristic argument
involving the interplay of tensile force  in the  {Lorentzian} sine-Gordon soliton  
and the tidal force  in   de Sitter spacetime. A similar heuristic 
argument, which remains to be confirmed analytically, also suggests the existence of 
 a threshold, $ (\mv/H)\sq \sim O(1)$,  
below which the tidal forces are too strong to permit the existence of
a static `t Hooft--Polyakov monopole in non-dynamical 3 + 1 de Sitter spacetime; 
$\mv$ is the  mass of the vector boson.} 
Linde   has {suggested} that new inflation 
could have triggered secondary inflation {at} the core of a   {GUT} 
{(Grand Unified Theory)} monopole 
even {if} the Hubble 
constant 
{at or after the GUT phase transition} {was} significantly smaller than the mass of the $X$ boson.  {We present a 
heuristic argument, which suggests that the $SO(3)$ `t Hooft--Polyakov monopole does not
allow  
secondary inflation at its core  when the inflationary background is weak. Based
on the above,  as yet analytically unconfirmed,  heuristic argument for the $SO(3)$ \mbox{`t Hooft--Polyakov} monopole, we conjecture
that secondary inflation at the core of a GUT monopole  is infeasible.}
\end{abstract}

%%%%%%%%%%%%%%%%%%%%%%%%%%%%%%%%%%
\section{Introduction}
Following `t Hooft~\cite{thoo} and Polyakov's~\cite{poly} construction of a monopole solution in a
gauge theory with spontaneously broken symmetry in Lorentzian spacetime, monopole solutions in non-Lorentzian backgrounds have been studied by 
several researchers.  For~example, Niewenhuizen, Wilkinson and Perry~\cite{nwpe} describe a `t Hooft--Polyakov-like monopole solution that couples to the background metric
generated by the monopole itself.   Using the dependence of a monopole's mass on the Higgs vacuum
expectation value, it follows that, for a sufficiently high Higgs vacuum expectation value, a~monopole would collapse into a black hole~\cite{lnwe,orti}. For a sufficiently high Higgs vacuum expectation value, it was shown  that non-singular monopole solutions do not exist~\cite{lnwe,bfma}.    The~stability of various types of magnetically charged solutions coupled to gravity has been studied in~\cite{bfma2,holl,lnwe2}.  Monopoles in   non-Lorentzian background{s} are of interest since monopoles, if~they exist, are believed to have been 
created during the GUT {(Grand Unified Theory)} phase transition in the early
universe  when the background was non-Lorentzian
~\cite{kibbo,kibbw,pres}.   

While~\cite{nwpe,lnwe,orti} considered self-gravitating monopoles---that is, monopoles coupled to {the} background metric generated by the monopole itself---Linde~\cite{li94} considered 
the interaction of   {a GUT} monopole's core with an inflationary background.  
Linde  argued that, although~the Hubble constant
of $\sim$$10^{10}$ GeV  in the inflationary phase{, circa the GUT phase transition,}
is several 
orders of magnitude smaller than the mass of the $X$ boson, which is believed to {have a mass of} 
$\sim$$10^{15}$  GeV, the~inflationary background could trigger
secondary inflation at the core of {a GUT} monopole.  
A similar suggestion
was also made by Vilenkin~\cite{vile}, who called it topological inflation.  
We revisit Linde's argument in Section~\ref{sec:s4}.

In this paper we {analytically solve} the problem of the existence of 
static soliton solutions in 1 + 1 sine-Gordon theory  {that is } coupled minimally to a non-dynamical de Sitter background.   The~associated differential equation is singular, making the construction of well-behaved soliton solutions a rather intricate problem.   We prove that  in the de Sitter background  static sine-Gordon solitons exist  if {$\a:=(m/H)\sq > 2$}  and~do not exist  if {$\a < 2$}, where $m$ is the mass parameter of the sine-Gordon theory and $H$ is the Hubble constant of the de Sitter~background.

A similar problem for~the Mexican hat potential was studied numerically by Basu and Vilenkin~\cite{bavi}.
{They 
report that they could not numerically find 
any nontrivial solution for $\a\tu{BV} \leq 2$, where  
 $\a\tu{BV}$ is the square of the ratio of the size of the de Sitter
horizon to the flat-space thickness of the domain wall.  For~
$\a\tu{BV} > 2$, they display solutions, which are obtained
numerically.  The~threshold $\a\tu{BV}=2$
that they report, based on numerical studies, is consistent with the
analytic threshold {$\a   = 2$} that we derive.
}  

{Chen, Cheng, Li and Zhai~\cite{cclz} have attempted to analytically justify Basu and Vilenkin's numerical work. However, their argument  has an error.  They attempt a series expansion of the solution around their equation's singular point, implicitly assuming that the solution is analytic at the singular point.  As~\eqr{tr} in our discussion shows, the~third  and   higher-order  derivatives of a solution of a second-order singular
differential equation diverge at the equation's singular point, making a  series expansion around a singular point
meaningless.  Instead, an~effective approach to studying the existence of 
bounded solutions in the vicinity of a singular point is to use 
integral equations, as {we show} in \mbox{Lemmas \ref{lem:l2} and \ref{lem:l3}} (Section \ref{sec:poe}). Also, their argument hinges on the incorrect claim that  their boundary condition  (8) requires $a\so > 0$.}

{The thresholds, such as {$\a =2$} 
that we derive}{,} have been  explained in the past by comparing the 
`size' of the soliton with {the  length scale
of the de Sitter horizon.}  \mbox{A soliton} is a non-dissipative solution  that is
held together as a compact configuration by the internal tensile forces.
It is instructive to try to understand the above threshold in terms of the interplay
between {de Sitter} background's tidal forces, which tear apart extended 
objects  such as soliton{s}, and~the soliton's internal tensile forces,
which resist the tidal forces to 
keep the soliton together as a compact~object.   

{Following the analytic derivation of
the threshold {$\a=2$},  we 
present  a  heuristic qualitative explanation of the 
threshold in terms of 
the interplay of tensile and tidal forces. The~ 
heuristic argument suggests the existence of a threshold, which  
agrees with the exact threshold to 
within {an O(1) factor}.}   

{We present a similar heuristic argument involving the interplay
of the tensile force in {a} {Lorentzian} `t Hooft--Polyakov monopole
and the tidal force in {a} 3 + 1  non-dynamical 
de Sitter spacetime.  Our heuristic argument, which is yet to be
confirmed analytically, suggests that a  static
\mbox{`t Hooft--Polyakov} monopole cannot   exist in a de Sitter background
when the tidal force is too strong.  Based on 
our heuristic analysis, we propose 
 a conjecture  
pertaining to the existence of static `t Hooft--Polyakov monopoles
in a non-dynamical de Sitter background.}

{
Linde~\cite{li94,li942}  has  %MDPI: References should be numbered in order of appearance. We noticed that “Ref li942 (Bibliography 24)” appears after “Ref cclz (Bibliography 15)”; please rearrange all the references so that they appear in numerical order. %Author: rearranged the references to ensure they are in the order of appearance.
 suggested that  even a `weak'
inflationary background could  have inflated away the 
field gradients at the core of a GUT monopole and could have 
induced secondary inflation at the monopole's core.  
We present a heuristic argument, which suggests
that when the inflationary background is weak 
the tidal forces cannot overcome
the tensile forces at the core of {a Lorentzian} $SO(3)$ `t Hooft--Polyakov monopole and, hence, the~background cannot 
induce   secondary inflation  
at the monopole's core.   Based on the above
heuristic argument, which remains to be confirmed
analytically, we propose a conjecture that
{\it  the relative weakness of the inflationary background
at and after the GUT phase transition makes secondary
inflation at the core of a GUT monopole infeasible}. %MDPI: Please confirm if the italics are necessary; if not, please remove them. The following highlights are the same. %Author: the italics are necessary since it is the summary of a conjecture.
Our conjecture can be settled
 through an exact analysis of the
GUT monopole in the  inflationary  background
{that prevailed} during the GUT phase transition. 
}

The paper is organized as follows.  In~Section~\ref{sec:s1} we formulate
and analytically solve the problem of the existence of static solitons in a non-dynamical
1 + 1 de Sitter background.  The~main results are {stated} in Section~\ref{sec:s1},
{and } the {proofs are} deferred to Appendices \ref{app:c}--\ref{app:e}.  {At the 
end of Section~\ref{sec:s1}, we also discuss the back reaction of the soliton
on the de Sitter background.}  {In Section~}\ref{sec:ttf} {we present heuristic
arguments involving the interplay of tensile and tidal forces.}
% for~the sine-Gordon
%soliton and `t Hooft--Polyakov monopole in a de Sitter background.}
In Section~\ref{sec:s2}  we {present a heuristic} 
estimate {of} the tensile force in the sine-Gordon soliton
{in  Lorentzian} background. Comparing the tensile force to the tidal force of 
the 1 + 1 de Sitter background we {heuristically derive} a threshold that agrees with the {exact}
threshold, obtained in Section~\ref{sec:s1}, to~within {an O(1) factor}.  In~Section~\ref{sec:s3} we compare { a heuristic}  estimate of 
the tensile force in a {static} `t Hooft--Polyakov monopole
{in   Lorentzian} background with the tidal force of a 3 + 1 de Sitter background
to derive a {heuristic} threshold for the non-existence of a `t Hooft--Polyakov 
monopole{; based on the heuristic threshold, we propose a conjecture
regarding the existence of a static $SO(3)$ `t Hooft--Polyakov monopole in a
de Sitter background.}
{Addressing Linde's suggestion, in~Section~\ref{sec:s4}, 
 we present a heuristic analysis of 
the extent to which a weak inflationary background can stretch} 
a pre-existing $SO(3)$ `t Hooft--Polyakov
monopole. {Based on our heuristic analysis, we propose a conjecture
that the  inflationary background {that prevailed at the GUT phase transition} 
could not have induced secondary inflation at the core of the GUT monopole.  Finally,
in Section~\ref{sec:disc}, we present some remarks about the stability of Sine-Gordon solitons, the~behavior of solitons at the end of  inflation, the~impact
of quantum corrections and the precision of the numerical results
presented in the paper.}
%%%%%%%%%%%%%%%%%%%%%%%%%%%%%%%%%%%%%%%%%%%%%%%%%%%%%
\section{Static Sine-Gordon Solitons in de Sitter~Spacetime}\label{sec:s1}
{The exact result, presented in this section and in Appendices \ref{app:a}--\ref{app:e},  pertains to a  singular differential equation and,
hence, is rather intricate.  In~Section~\ref{sec:overview}, we present  an overview of the result to 
ensure that the mathematical arguments that follow do not 
obscure the essence of the result.  We begin by formulating the 
key differential equation.} 

\subsection{Lagrangian and Its~Symmetries}\label{sec:lag}
Consider the 1+1 sine-Gordon Lagrangian minimally coupled to  background 
de Sitter spacetime.  The~action  is  %MDPI: Please ensure all variables/values in the equation appear in the same format in the text (normal/italic/bold/subscript/superscript). %Author: Thanks for the guidance. I believe I have done that.
\bea
S[\varphi] 
= \int \sqrt{-g} \ \left\{ 
\frac 1 2 g^{\mu\nu} \partial_\mu \phi \partial_\nu \phi - 
{m^2\over \beta^2}(1-\cos(\beta\phi))
\right\} d\chi \  dt; 
\quad g\umn = diag\lrr{1, -\eto{-2Ht}}\quad \leqn{sga}
\hspace{-9pt}
\eea      
where $m$ is the mass parameter of the sine-Gordon theory, 
 $\beta$ the~sine-Gordon coupling constant,  
$g$  the determinant of the metric and $H$ the~Hubble
constant.  %\textls[-15]
{We ignore the back reaction of the field on the 
background metric. 
The corresponding
 Euler--Lagrange equation is}
\begin{eqnarray}
\phi_{tt} + H \phi_t - e^{-2Ht} \phi_{\chi\chi} = - {m^2\over \beta} \sin(\beta\phi).
\label{ele}
\end{eqnarray}
The %MDPI: Please confirm if the no-indent paragraph should be retained entire manuscript. %Author: No, this is not the beginning of a new paragraph.
  physical  spatial coordinate $u$ is related to the comoving spatial coordinate $\chi$ as 
$
u = e^{Ht} \chi.
$
We look for solutions of (\ref{ele}) of the form 
{$\phi(t,\chi) = \psi\left(u \right)$} that 
are {`static'} in the sense that they depend only on the physical spatial~coordinate.  

{
A `static'  
solution is not stretched by the expansion of the
de Sitter universe; it bucks the tidal forces of  inflationary expansion, its physical `dimensions'---as measured by a physical ruler---remaining 
unaltered as the universe expands.  As~an example, consider a 
`static' solution $\psi(u)$ 
 that has a zero at $u\sz>0$.  As~the universe 
 expands, the~location of the zero---as measured
 by a physical ruler---remains at a distance $u\sz$
 from $u=0$, 
 although the comoving 
 coordinate of the above zero, denoted 
 $\chi\sz(t)$,    
decreases exponentially as 
 $\chi\sz(t) \deq e\tu{-Ht} u\sz$.   A~distant 
 observer  at a fixed comoving coordinate, who is to the
 right of $u\sz$ at $t=0$, 
 would see the comoving coordinate of the
 zero decreasing exponentially  as the universe expands.
}

{Setting $z := Hu, \ \vp(z) := \beta\,  \psi(z/H) $, and~$\a:= {m^2 \over H^2}$}, 
%
%From (\ref{ele}) 
we see that a 
static solution of the Euler--Lagrange equation satisfies
the following singular differential equation:
\begin{eqnarray}
(z^2-1) \varphi_{zz}(z) + 2\, z\, \varphi_z(z) + \alpha \sin(\varphi(z)) = 0; 
\label{sfe-in-z}
\end{eqnarray}
 \erf{sfe-in-z} is singular { because the  coefficient of the highest 
 derivative vanishes at certain values of the independent variable 
 $z$.  In Eq. } (\ref{sfe-in-z}), {the coefficient of $\vp_{zz}$ 
 vanishes}  %MDPI: \hl{Footnote}s are not supported in our journal. We have therefore included this paragraph in the main text. Please confirm. %Author: confirmed.
 at $z=\pm 1$, which corresponds to the horizon at 
$u=\pm {1\over H}$.  We will call $z$ the {\it static coordinate}.

The vacua of the theory are
at $\vp \deq 2 n \pi$ for integer values of $n$.  A~topologically nontrivial static soliton---hereafter abbreviated to just soliton---is a  solution of 
(\ref{sfe-in-z}) that {approaches different vacua}  as $z \ra \pm \infty$.   The~topological charge $C[\vp]$
of {a  soliton} $\vp$ is 
\bea
C[\vp] \aea \dsf 1 {2\pi} \lrc{\lt z \infty \vp(z) - \lt z {-\infty}\vp(z)}
\label{tc}
\eea

We note that $\vp \ra -\vp$ and $\vp \ra \vp + 2 m\pi$  
are symmetries of the Lagrangian for~every integer $m$.
Further, {if } $\vp(z)$ is a solution of \erf{sfe-in-z}, then so is $\vp(-z)$, implying
that, if a  soliton with charge $Q$ exists, then so does 
a  soliton with charge $-Q$.

In order 
to {prove the existence of } a   soliton solution of \erf{sfe-in-z} with charge $Q=\pm 1$ it is 
sufficient to {prove the existence of a bounded} solution for $z \geq 0$ with the following boundary conditions:
\bea
 \vp(0) \deq \pi, \qquad \lt z {\infty} \vp(z) \aea 0,  \leqn{boundcon}
\eea

As noted earlier, \erf{sfe-in-z} is singular at
the horizon $z = \pm 1$.  {We 
present results about the behavior of bounded solutions
near the horizon ($z\approx 1$),
outside the horizon ($z>1$) 
 and  inside the horizon ($0\leq z< 1$) in~Sections~\ref{sec:aah}--\ref{sec:is},
 respectively.  We start by presenting an overview
 of the discussion that follows.}
%%%%%%%%%%%%%%%%%%%%%%%%%%%%%%%%%%%%%%%%%

\subsection{Overview}\label{sec:overview}
{In this section we present an overview of our analytical results pertaining to the sine-Gordon solitons in a de Sitter background.}

\subsubsection{Solutions That Are Bounded at~Singularities}

{Eq. \eqref{sfe-in-z} has singularities at $z = \pm 1$, and~its
solutions can become unbounded in the vicinity of the singularities. In~
Corollary \ref{cor:c2} (Section \ref{sec:nc}), 
we show that a necessary condition for a solution $\vp\sb$ to remain bounded 
at  $z=1$   is that 
\beas
\vp\sb'(1) \aea - \dsf\a 2 \, \sin(\vp\sb(1)).
\eeas
Thus, for~bounded solutions of \eqr{sfe-in-z},
the initial conditions at $z=1$, namely $\vp\sb(1)$ and $\vp\sb'(1)$, cannot
be specified independently.  Solutions of \eqr{sfe-in-z}, bounded at  
$z=1$, form a 1-parameter family, parametrized by $\vp\sb(1)$. 
}

{We need to show that  solutions that are bounded in the 
vicinity of  $z=1$ actually exist.  Rewriting \eqr{sfe-in-z} in terms
of the {\it hypergeometric coordinate,} $x = (1-z)/2$,  
we show, in~Lemma \ref{lem:l2} (Section \ref{sec:poe}), that the solution of a
certain integral equation gives a solution
of \eqr{sfe-in-z} in a neighborhood of $z=1$ ($x=0$). We then show, in~Lemma 
\ref{lem:l3} \mbox{(Section \ref{sec:poe}),}  that the aforementioned integral equation 
actually has a bounded solution
in the vicinity of $z=1$ for~every $0<\vp(1)<\pi$,
thereby proving that the  1-parameter family of bounded solutions of 
\eqr{sfe-in-z}, mentioned above, does exist.}

\subsubsection{Form of the~Solitons}

In Lemma \ref{lem:le10} (Appendix \ref{app:d}), we show  that, if $\vp(1) = n\pi$, 
then $\vp(z) \equiv n\pi$ for all $z$.  
Therefore, based on the discussion preceding Eq. \ct{boundcon},
without loss of generality, we can assume that $0<\vp(1)<\pi$  for~topologically nontrivial solutions.

 Figure  \ref{theasymptopia}  %MDPI: Figure should be numbered in order of appearance. We revised all figure to be in numerical order; please check all the citations accordingly entire manuscript. %Author: checked.
 shows {illustrative examples of} the behavior of bounded solutions
for $z> 1$.  If~$0<\vp(1)<\pi$, then, regardless of the value of 
$\a$, the~solutions vanish   as $z\ra \infty$.   The~decay is quite
slow and, hence,  the~plots are shown with a logarithmic scale.   

{In}~Lemma
\ref{lem:l1} (Section \ref{sec:asympt}), we show, rigorously, that 
the bounded solutions with $0 < \vp(1) < \pi$  vanish as $z\ra \infty$.  
Therefore, in~order to prove the existence of a topologically nontrivial solution, it suffices
to find a solution that starts at $0 < \vp(1) <\pi$ and reaches $\vp(0) = \pi$ when
we integrate {Eq.} (\ref{sfe-in-z})  from $z=1$ to $z=0$.

\begin{figure}[H]
\includegraphics[height=2.25in]{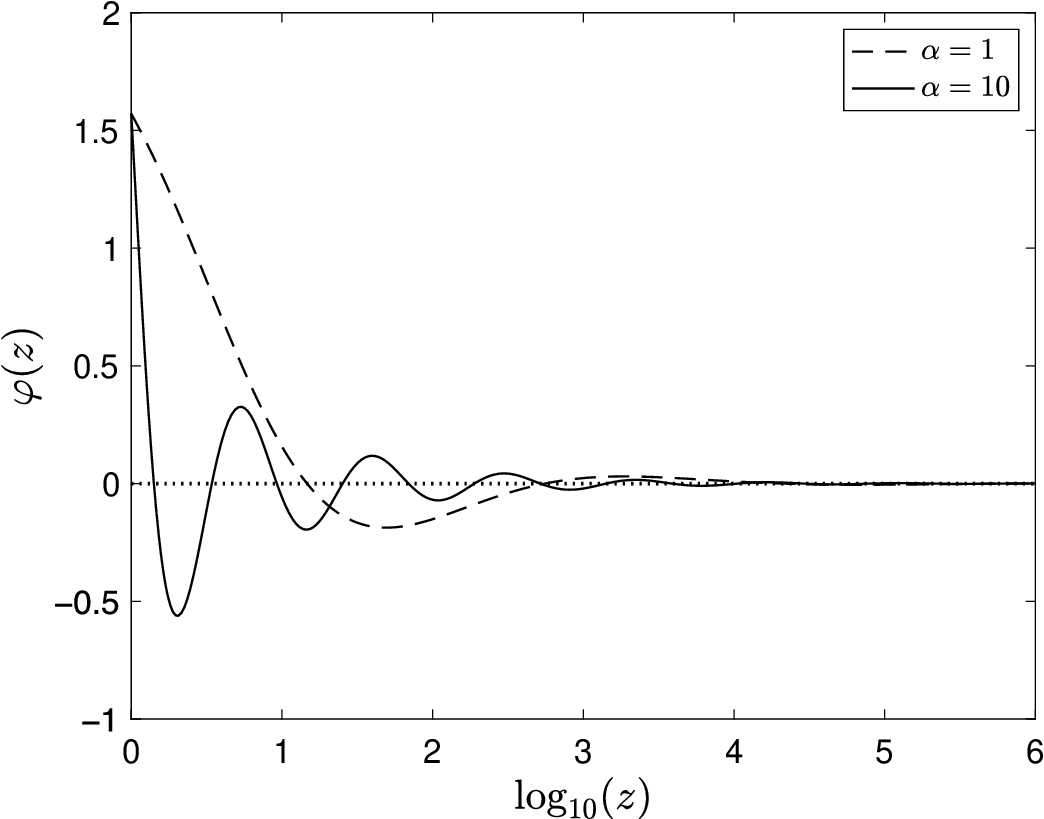}
\caption{ Sample  %MDPI: (1) Figure should be numbered in order of appearance. We revised all figure to be in numerical order; please check all and revise the citations accordingly if needed. (2) Please change the hyphen (-) into a minus sign (−, “U+2212”) in the figure, e.g., “-1” should be “−1”. (3) Please confirm whether an explanation of the dotted line at 0 needs to be added to the figure caption.
 solutions of \eqr{sfe-in-z}  {  in the range
 $1^+<z<\pot 6$, where $1^+ = 1 + \potm{14}$}.  The~boundary conditions
are {$\vp(1^+) = \pi/2, \ \vp'(1^+) =   -(\a/ 2) \, \sin(\vp(1^+))$ .  The 
dotted line represents $\vp(z) \equiv 0$.}}\label{theasymptopia} 
\end{figure}

\subsubsection{$\a<2$}
{The behavior of solutions for $\a < 2$ is {illustrated} in  Figure  \ref{thealphalw}. In~the graphs,
we have started at {$z=1\mmi:=1-\potm{15},$} and integrated the differential equation back to $z=0$.
The noteworthy feature in  Figure  \ref{thealphalw}  is that, regardless of what initial 
value {$0< \vp(1\mmi)<\pi$ is chosen  at $z=1\mmi$}, $\vp(z)$ remains strictly between 0 and $\pi$
in the range {$0\leq z\leq 1\mmi$}.     For~$\a < 2$, a~bounded 
solution with $0<\vp(1)<\pi$
cannot  reach $\vp(0) = \pi$.  Therefore, a~solution
satisfying the conditions in \ct{boundcon} does not exist for $\a<2$.} 

{Figure  \ref{theblowUp} {illustrates} the behavior of a solution that is bounded
in the vicinity of $z=1$ but~does not reach $\vp(0) = \pi$.  If~ we integrate 
{Eq.} (\ref{sfe-in-z}) {from
$z=1^-=1-\potm{15}$ past $z=0$ to $z=-1^+=-1+\potm{14}$, 
 the solution} becomes unbounded near $z=-1$
since it does not satisfy the  condition analogous to \ct{boundcon}, namely
$\vp'(-1) = \a/2 \sin(\vp(-1))${,} at~$z=-1$.}

{The above discussion represents the content of Theorem \ref{thm:t1} (Section
\ref{sec:lss}), in~which we show  that, for $\a<2$, every bounded solution
of \eqr{sfe-in-z}  is topologically trivial.  There are no solitons {for $\a<2$}.  
}

\begin{figure}[H]
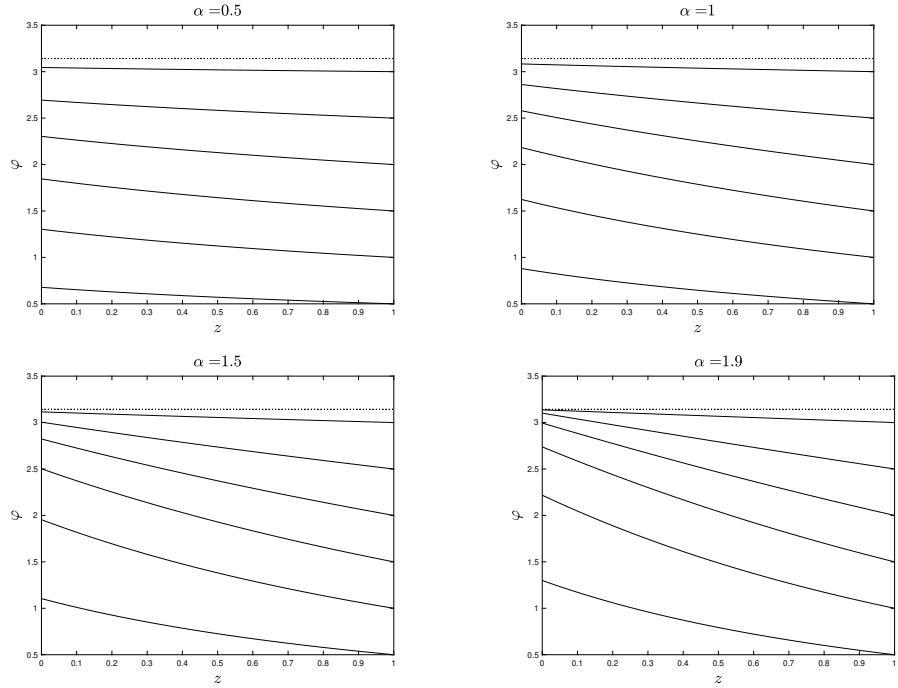

\igwe 2 {0pt5}\hs{0.5}\igwe 2 {1} \\[0.5\bls]
\igwe 2 {1pt5} \hs{0.5}
\igwe 2 {1pt9}  \\
\caption{Solutions  (solid lines) %MDPI: (1) Please confirm whether an explanation of the dotted line and arrow needs to be added to the figure caption. (2) We moved figures 2 and 3 here to reduce empty space, please check and confirm.
 for different values of $\a < 2$ {in the range $0\leq z \leq 1\mmi:=1-\potm{15}$. The shown solutions correspond to $\vp(1\mmi) = 0.5, 1, 1.5, 2, 2.5, 3$.  
 $\vp(0) < \pi$ for all of the shown solutions. The dotted line represents $\vp(z) \equiv \pi$.}}\label{thealphalw}
\end{figure}   

\vspace{-6pt}

\begin{figure}[H]
\igwe 2 {blowUp}
\caption{ Solution (solid line)  %MDPI: (1) Please confirm whether an explanation of the dotted line and arrow needs to be added to the figure caption. (2) Please change the hyphen (-) into a minus sign (−, “U+2212”) in the figure, e.g., “-1” should be “−1”.
 for {$\a = 1.5$  and $\vp(1-\potm{15}) = 1.75, $ in~the range 
 $-1+\potm{14} \leq z \leq 
 1+\potm{15}$. 
 The dotted line represents $\vp(z) \equiv \pi$.}}\label{theblowUp}
\end{figure}   

\subsubsection{$\a > 2$}

%\textls[-15]
{Figure  \ref{theshoot} illustrates the behavior of solutions at a sample value  of  $\a>2$,
namely $\a = 5$.
For $\vp(1\mmi) = 1.14218, \, 1\mmi = 1-\potm{15}$, we have $\vp(0)>\pi$; for $\vp(1\mmi) = 0.14218$, 
$\vp(0)<\pi$.   
%In~Lemma \ref{lem:l9} (Appendix \ref{app:e}), we show that $\vp(0)$
%varies continuously with $\vp(1)$ for a bounded solution.  
%As we vary $\vp(1)$ from 0.14218 to 1.14218, at~an intermediate
%value of $\vp(1)$---specifically at $\vp(1) \approx 0.54218$---we 
%reach $\vp(0)=\pi$.  Taken together with the asymptotic 
%behavior, depicted in {Figure} \ref{theasymptopia},  {Figure} \ref{theshoot}
%shows that topologically nontrivial soliton exists for $\a = 5$.  

\begin{figure}[H]
%\ighe {2} {Shoot}
\includegraphics[width=5.5 cm]{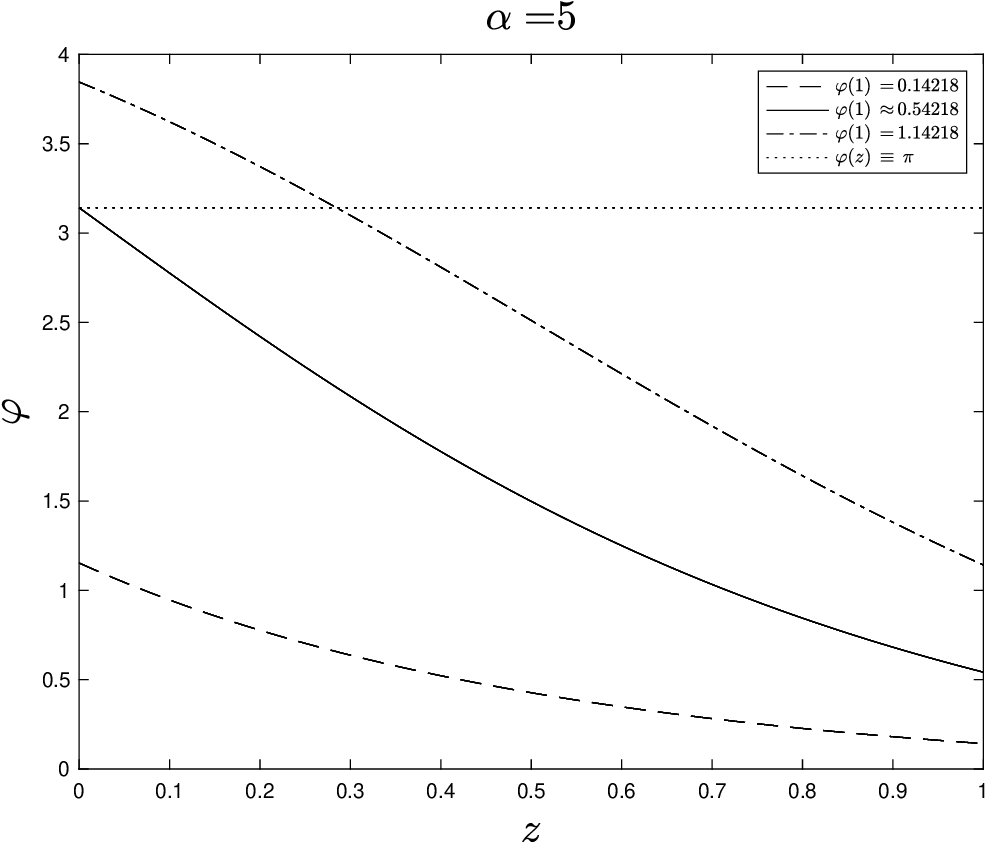}
\caption{{The three plots, in the range $0\leq z\leq 1\mmi:= 1-\potm{15}$, show that} $\vp(0)>\pi$  %MDPI: Please confirm whether an explanation of the arrow needs to be added to the figure caption.
for {$\vp(1\mmi) = 1.14218$; $\vp(0) < \pi$ for 
$\vp(1\mmi) = 0.14218$; $\vp(0) = \pi$ for $\vp(1\mmi) \approx 0.54218$ .
}}\label{theshoot}
\end{figure}   
}

In~Lemma \ref{lem:l9} (Appendix \ref{app:e}), we show that $\vp(0)$
varies continuously with $\vp(1)$ for a bounded solution.  
As we vary $\vp(1\mmi)$ from 0.14218 to 1.14218, at~an intermediate
value of $\vp(1\mmi)$---specifically at $\vp(1\mmi) \approx 0.54218$---we 
reach $\vp(0)=\pi$.  Taken together with the asymptotic 
behavior  depicted in {Figure} \ref{theasymptopia},  {Figure} \ref{theshoot}
shows that topologically nontrivial soliton exists for $\a = 5$.  

%{
%Soliton solutions for some  values of $\a>2$ are displayed in
% {Figure} \ref{thealphagw}.   The~range is limited to $0\leq z \leq 10$ to 
%be able to clearly display the behavior inside the horizon ($0\leq z<1$). 
%Since the initial condition has to be imposed at $z=1$, to~keep
%the solution bounded, 
%the integration was done separately
% in two intervals, $0\leq z \leq 1-\potm{15}$ 
%and  $1+\potm{15}\leq z\leq  10$,
%imposing the initial conditions at $1\pm \potm{15}$.  The~
%choice of $\potm{15}$ is discussed in Section~\ref{sec:numerics}.
%}

%{For comparison, in~ {Figure} \ref{thelorSG}, we have also plotted the Lorentzian sine-Gordon
%antikink $\vp_{_{SG}}(x)$ that interpolates between the vacuum  
%$\vp=2\pi$ at $x\ra -\infty$ and the vacuum 
%$\vp = 0$ at $x\ra \infty$.} 

\begin{figure}[H]
%\begin{adjustwidth}{-\extralength}{-4.5cm}
\centering
\igwe {1.8} {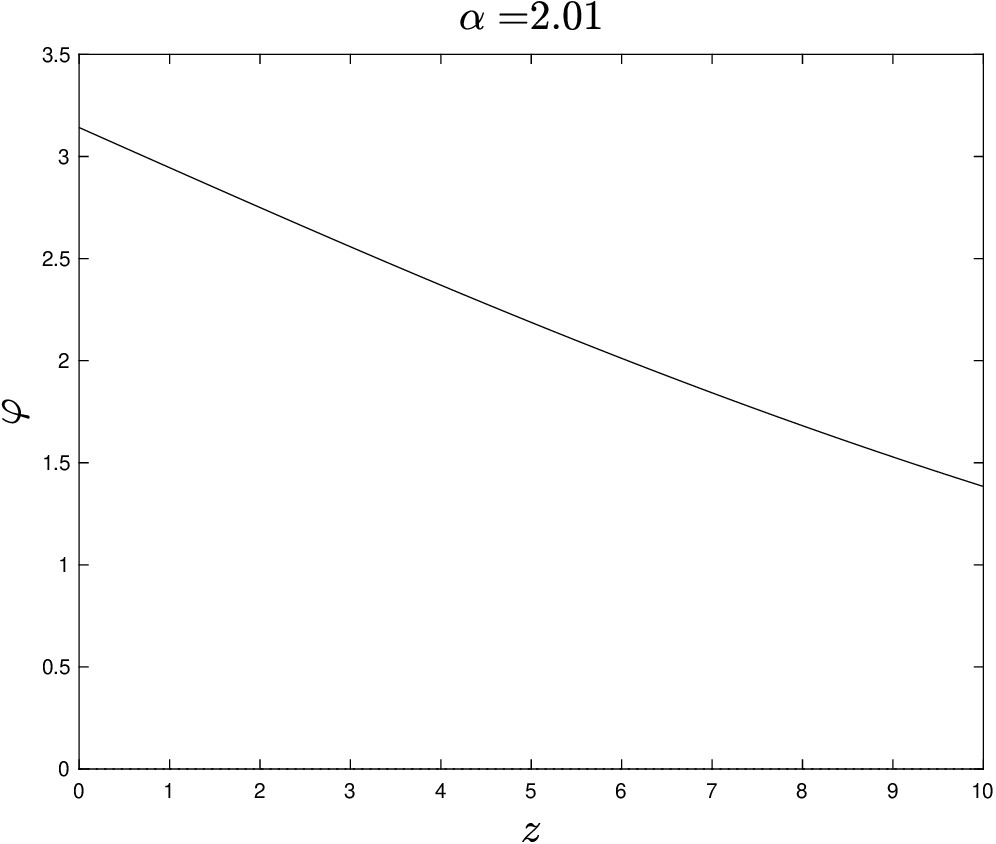} \igwe {1.8} {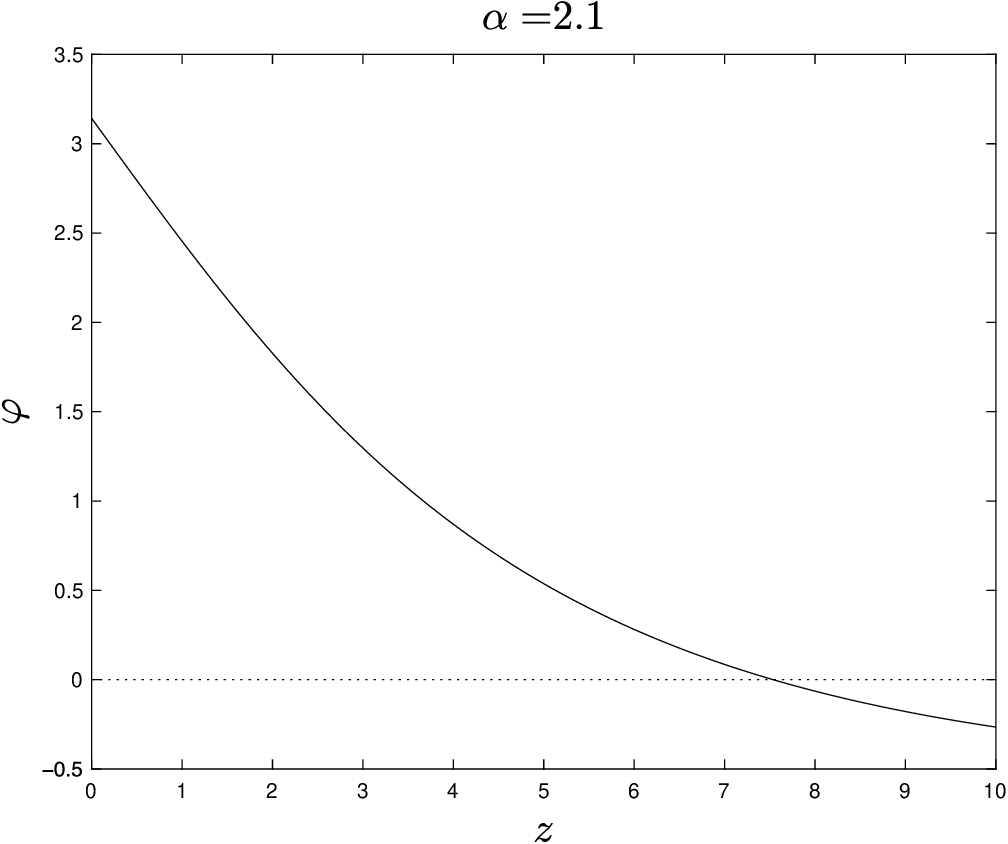} \igwe {1.8} {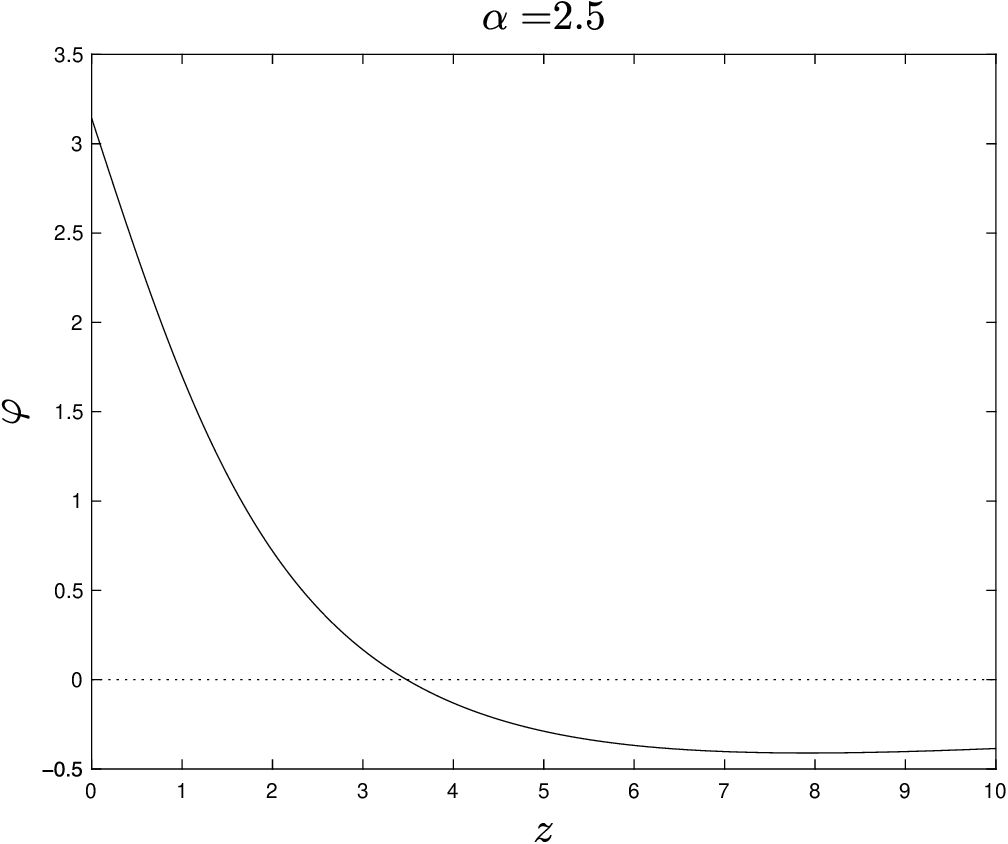} \\
\igwe {1.8} {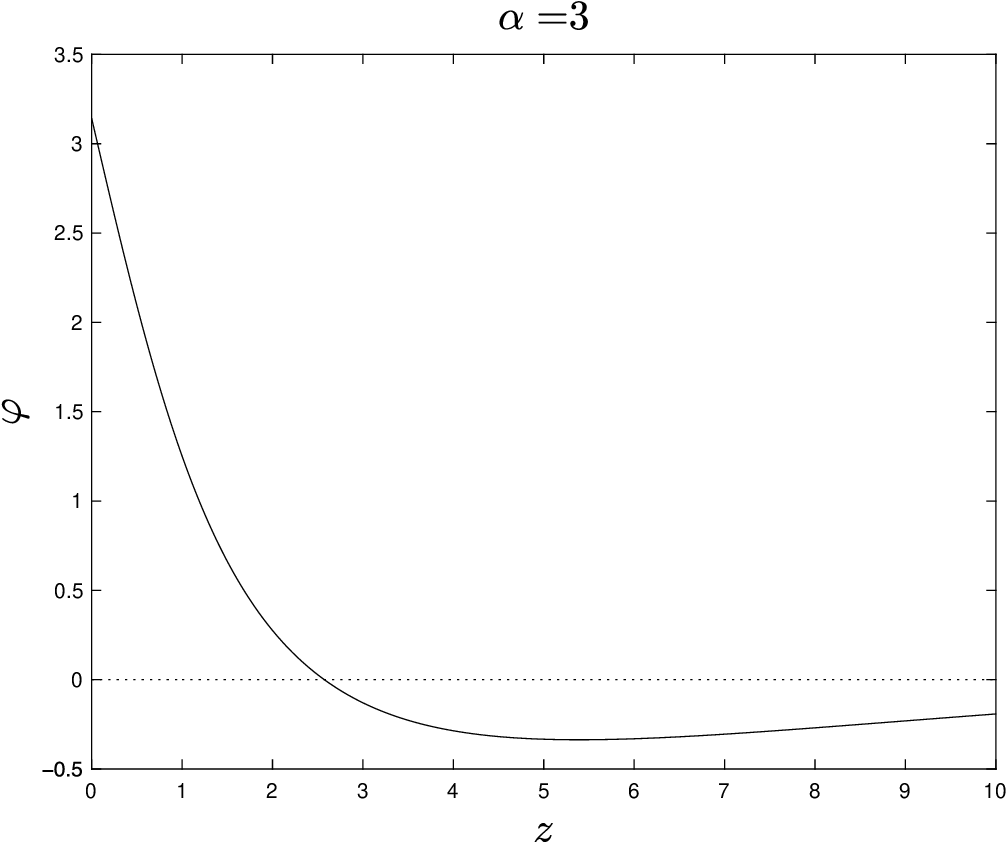} \igwe {1.8} {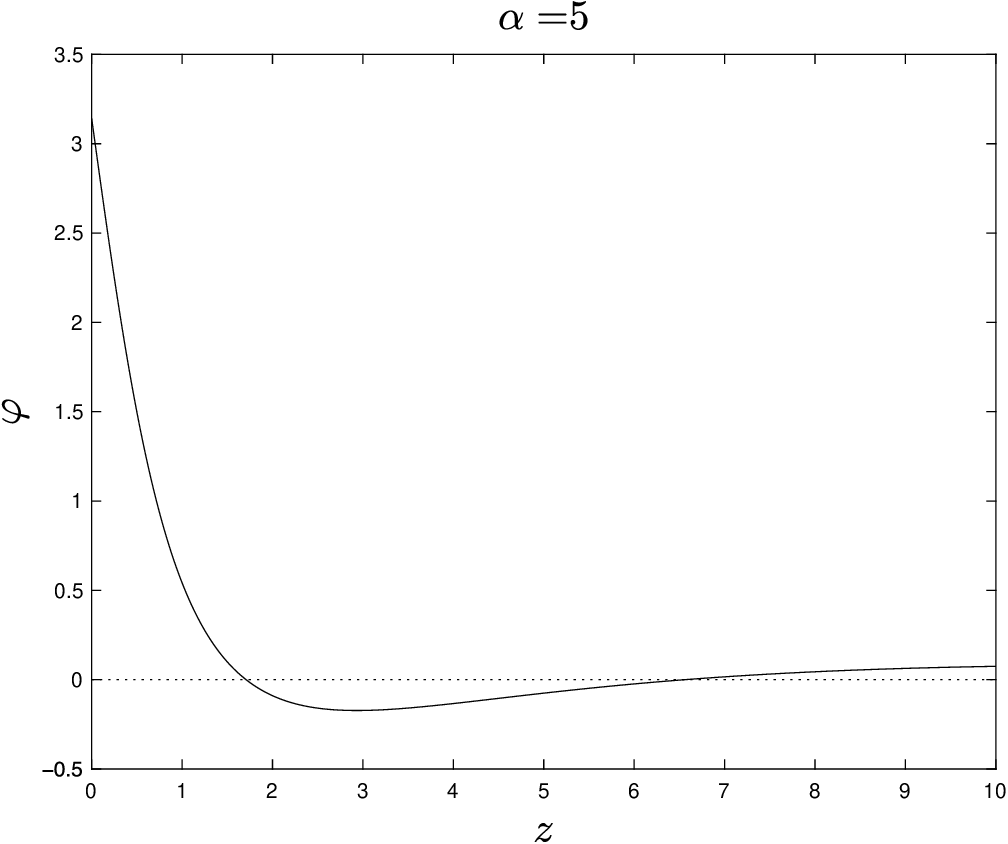} \igwe {1.8} {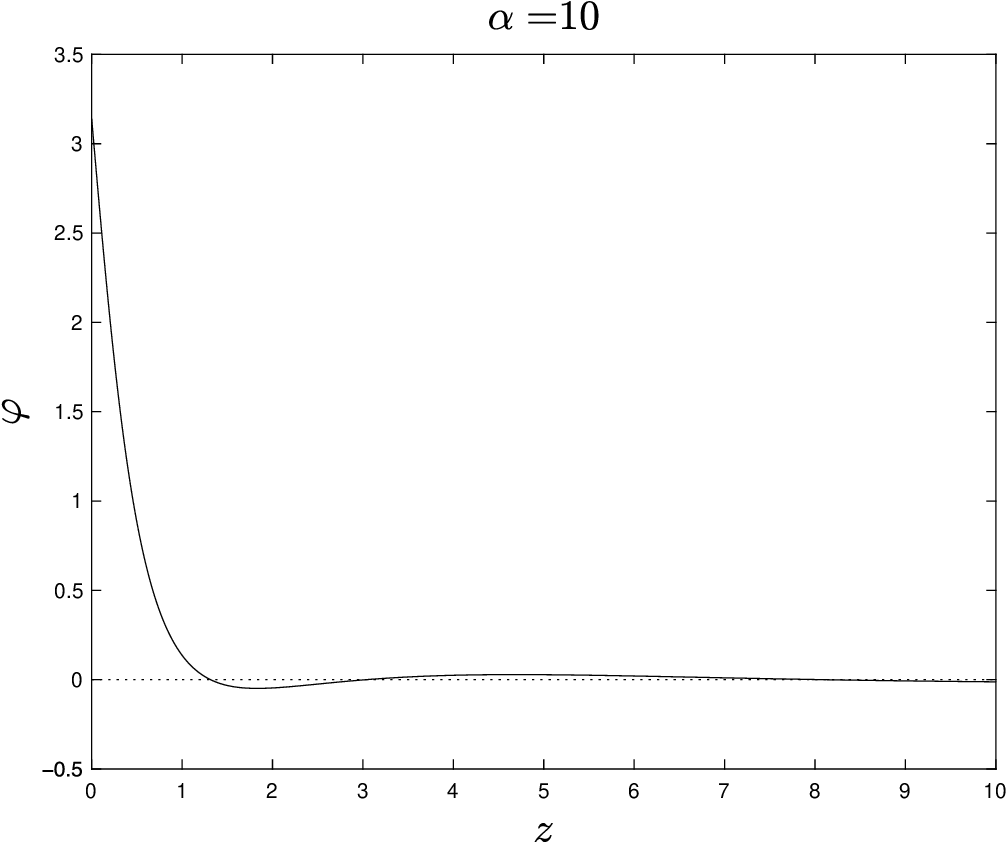}\\
\igwe {1.8} {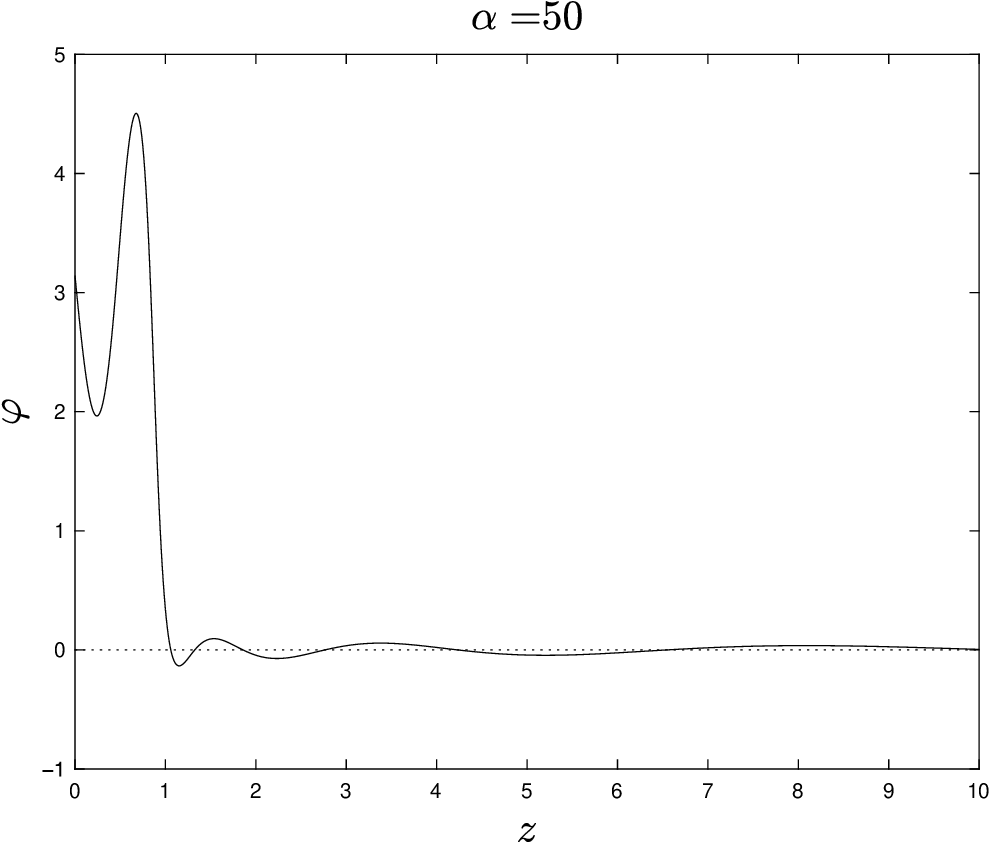} \igwe {1.8} {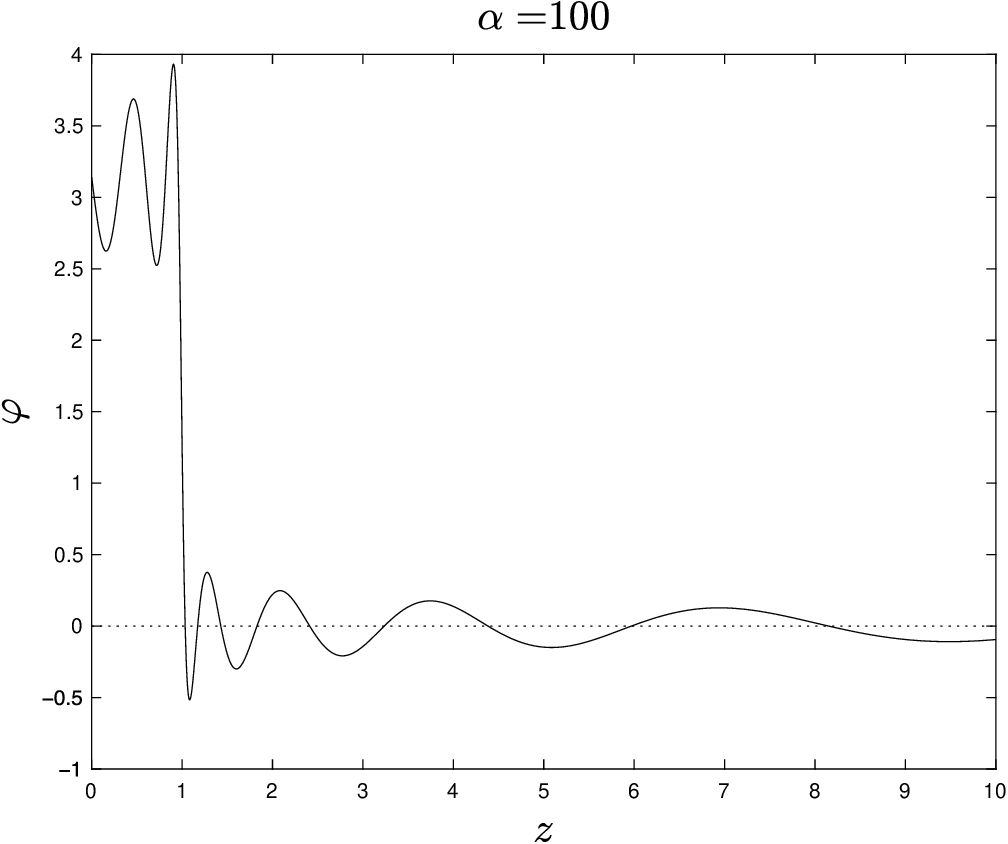} \igwe {1.8} {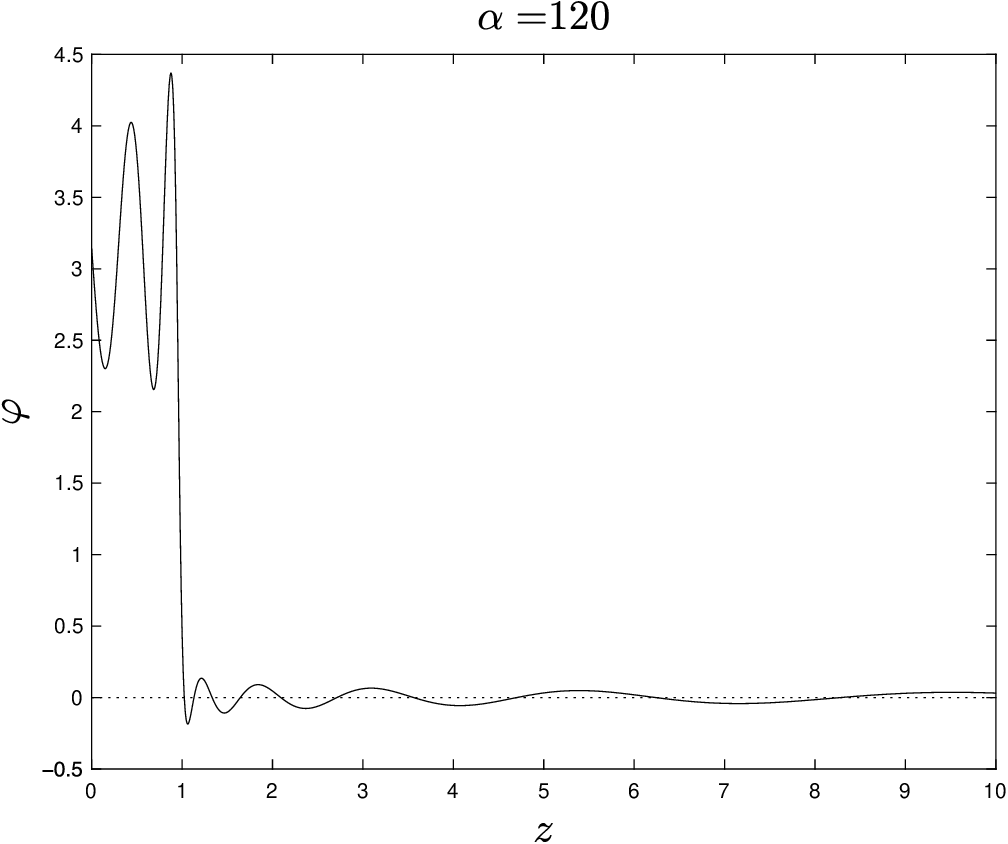}
%\end{adjustwidth}
\caption{{{Soliton} %MDPI: (1) Please confirm whether an explanation of the dotted/solid lines and arrow needs to be added to the figure caption. (2) Please change the hyphen (-) into a minus sign (−, “U+2212”) in the figure, e.g., “-1” should be “−1”.
 solutions (solid lines) for different values of $\a>2$ in the interval $0\leq z\leq 10$. All of the solutions satisfy $\vp(0) = \pi$ and $\lt z\infty \vp(z) = 0$ and~are 
bounded at the singularity $z=1$. {The plots were obtained by integrating} 
Eq. (\ref{sfe-in-z}) {separately in the intervals $[0,1-\e]$ and $[1+\e,10]$, where 
$\e =\potm{15}$. The value $\vp(1-\e)$ was chosen 
to ensure $\vp(0)=\pi$. The dotted lines represent 
$\vp(z) \equiv 0$.}}}\label{thealphagw}
\end{figure}   

\vspace{-6pt}

\begin{figure}[H]
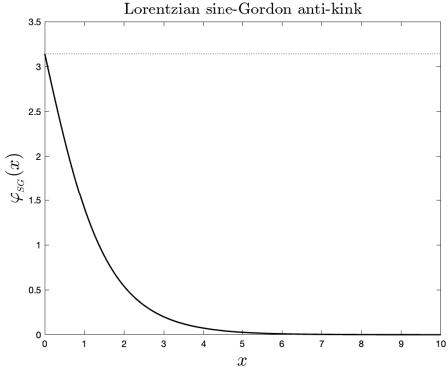

\ighe 2 {lorSG}
\caption{{The} %MDPI: Please confirm whether an explanation of the dotted and solid lines needs to be added to the figure caption.
 Lorentzian sine-Gordon antikink $\vp_{_{SG}}(x) = 4 \tan\inv(e\tu x)$ corresponding to $m=1$ (solid line). {The dotted line represents
 $\vp(z) \equiv \pi$.}}\label{thelorSG}
\end{figure}
   
{
Soliton solutions for some  values of $\a>2$ are displayed in
 {Figure} \ref{thealphagw}.   The~range is limited to $0\leq z \leq 10$ to 
be able to clearly display the behavior inside the horizon ({$0\leq z\leq1$}). 
Since the initial condition has to be imposed at $z=1$, to~keep
the solution bounded, 
the integration was done separately
 in two intervals, $0\leq z \leq 1-\potm{15}$ 
and  $1+\potm{15}\leq z\leq  10$,
imposing the initial conditions at $1\pm \potm{15}$.  The~
choice of $\potm{15}$ is discussed in Section~\ref{sec:numerics}.
}

{For comparison, in~ {Figure} \ref{thelorSG}, we have also plotted the Lorentzian sine-Gordon
antikink $\vp_{_{SG}}(x)$ that interpolates between the vacuum  
$\vp=2\pi$ at $x\ra -\infty$ and the vacuum 
$\vp = 0$ at $x\ra \infty$.} 

{
In the following subsections of this section, we state the 
main results, deferring the proofs to {Appendices} \ref{app:c}-\ref{app:e}. 
\eqr{sfe-in-z} has singularities at $z = \pm 1$,
which form the horizon of the region $-1<z<1$. 
 In
Section~\ref{sec:aah} we present results that prove the 
existence of  bounded solutions 
in  a neighborhood of the singularity, $z=1$. 
In
Section~\ref{sec:asympt}, we present results, {which} show
 that {a} solution{, $\vp(z)$,} approaches the vacuum that is 
 closest to $\vp(1)$ as $z\ra \infty$.  
 Finally, in~Section~\ref{sec:is}, we focus on  the region 
 inside the horizon, $-1<z<1$, and {present results, which} show that
 solitons exist for $\a>2$ and do not exist
 for $\a<2$. }
%%%%%%%%%%%%%%%%%%%%%%%%%%%%%%%
\subsection{{Bounded~Solutions} in a Neighborhood of a {Singularity}}\label{sec:aah}
{\eqr{sfe-in-z} is closely related to the hypergeometric differential equation. We 
use the known solutions of the hypergeometric differential equations to 
establish the existence of bounded solutions of \eqr{sfe-in-z}.   In~Section~\ref{sec:ahg}
below, 
we describe the relation between \eqr{sfe-in-z} and the hypergeometric equation.
In Section~\ref{sec:nc}, we establish a necessary condition that must be satisfied by 
a solution at~the singularity $z=1$ if~it is to be bounded in a neighborhood 
of the singularity.  Finally, in~Section~\ref{sec:poe}, we prove the existence of a 
1-parameter family of bounded solutions of \eqr{sfe-in-z} in a neighborhood of 
the singularity $z=1$.}

\subsubsection{Associated Hypergeometric~Equation}\label{sec:ahg}
{The }hypergeometric differential equation
\bea
x(x-1) \vp\uxx + ((a+b+1)x-c)\vp\uxo + ab \vp \aea 0
\label{hge}
\eea
was first formulated by Euler and was studied extensively by
Gauss and Riemann.  The~relation between Eq.  \eqref{sfe-in-z} and \eqref{hge} can be seen by transforming from the {\it  {static coordinate}} $z$ to the {\it  {hypergeometric coordinate}} $x${,} as  $x = (1-z)/2$.  In~hypergeometric coordinate \mbox{\erf{sfe-in-z} becomes}
\bea
x(x-1) \vp\uxx + (2x-1)\vp\uxo + \a \vp \aea f(\vp)
\qquad f(\vp)\dfn \a(\vp - \sin(\vp)
\label{sfe-in-x}
\eea
With a slight abuse of notation we use $\vp$ to denote both
the solution $\vp(z)$ of \erf{sfe-in-z} as well as the solution
 $\Psi(x) = 
 \vp(1-2x)$ of \erf{sfe-in-x}.

For $a+b=c=1, \ ab = \a$ in \erf{hge},   the~left-hand side
of \erf{hge} coincides with the left-hand side of \erf{sfe-in-x}.
Thus \erf{sfe-in-x} is a nonlinear relative of the hypergeometric
differential equation {Eq.} (\ref{hge}). %\erf{sfe-in-x}.  
The~singularity $z=1$ in 
\erf{sfe-in-z} corresponds to the singularity $x=0$ in 
\erf{sfe-in-x}.

It is well known that the hypergeometric equation \erf{hge} has
two linearly independent solutions in a (punctured) neighborhood of $x=0$.  For~the special values of $a,b,c$ of interest to us, namely $a+b=c=1, \, ab=\a$, the~two linearly independent solutions of \erf{hge} are 
denoted $\vp\so(x)$ and 
$\vp\sw(x)$.    The~first
solution, $\vp\so(x)$, is analytic in the neighborhood ($-$1, 1) %MDPI: We revised the hyphen (-) to a minus sign (“−” U+2212). Please confirm.
and 
is described in \erf{phi-1} in  Appendix \ref{app:b}.  The~second solution is  
\bea
\vp\sw(x) \aea \vp\so(x) \log(|x|) + h(x), \qquad x\in(-\e,\e)\setminus \lrc 0
\label{phi-2}
\eea
where $h(x)$ is an analytic function in ($-\ruh, \ruh$)
for some $0<\ruh <1$. The~ 
form of $\vp\so(x)$ and the construction of  
$\vp\sw(x)$ are described in Appendix \ref{app:b}.
%%%%%%%%%%%%%%%%%%%%%%%%%%%%%%%%%%%%%%%%%%%%%%%%

\subsubsection{A Necessary Condition for Bounded~Solutions}\label{sec:nc}

 First, we show that, if a solution of \erf{sfe-in-x} is to be  bounded at the {singularity}  
$x=0$, then  it must satisfy the boundary condition 
described in Lemma \ref{lem:l4}, {stated below}.

\begin{lemma}\label{lem:l4}
Let $\vp\sb(x)$ be a bounded solution of 
\bea
x(x-1) \vp\sb'' + (2x-1) \vp\sb' + \a \sin(\vp\sb) \aea 0
\label{dex}
\eea
in a neighborhood $(-\e, \e)$, $\e>0$. Then,
\bea
\lt x 0 \vp\sb'(x) \aea  \a  \sin(\vp\sb(0)) 
\label{bc}
\eea
\end{lemma}

{The proof of Lemma \ref{lem:l4} is presented in Appendix \ref{app:c}.}
{Translating} Lemma \ref{lem:l4}
from hypergeometric coordinate  to static coordinate,
{we obtain the following corollary.}

\begin{corollary}\label{cor:c2}
Let $\vp\sb(z)$ be a bounded continuous solution of 
\bea
(z\sq -1) \vp\sb'' + 2z \vp\sb' + \a \sin(\vp\sb) \aea 0
\eea
in a neighborhood $(1-\e, 1+\e)$, $\e>0$. Then,
\bea
\lt z 1 \vp\sb'(z) \aea  -\f\a  2 \sin(\vp\sb(1))
\label{bc2}
\eea
\end{corollary}

\subsubsection{Proof of Existence of Bounded Solutions in a Neighborhood of~Singularities}\label{sec:poe}

{We} show that  one can obtain a solution of 
the differential equation \erf{sfe-in-x}  
by solving a related 
integral equation, {as shown} 
 in Lemma \ref{lem:l2}{ below}.
  
\begin{lemma}\label{lem:l2}
Let $\vp\sb(x)$ be a bounded continuous solution of the integral
equation
\bea
\begin{array}{l}
{
\vp(x) \aea \vp\sw(x) \int_0^x \bsf{\vp\so(y) \, f(\vp(y))}{(y-1) \, y\, W(y)} dy 
%}
%\\
%{\ \ \ \ \ \ \ \ \ \ \ \ \ \  
- \vp\so(x) 
\int_0^x \bsf{\vp\sw(y) \, f(\vp(y))}{(y-1) \, y\, W(y)} dy + \g \vp\so(x)
%\qquad \qquad  
}
\end{array}
\label{ie}
\eea
in some sufficiently small 
interval $(-\e,\e)$,  where $0 < \e < \ruh$.
Then, in
$( -\e,\e )$, 
$\vp\sb(x)$ satisfies the differential equation 
\bea
x(x-1) \vp\uxx + (2x-1)\vp\uxo + a\sin(\vp) \aea 0
\label{deqx}
\eea 
and the boundary condition described in Lemma \ref{lem:l4}.

\hs{0.25}$\vp\so(x)$ and $\vp\sw(x)$ are described in 
\erf{phi-1} and \erf{phi-2}. {$\rho_h$ is defined after}
\erf{phi-2}.  $W(x)$ is the Wronskian of $\vp\so$ 
and $\vp\sw$.    $f(\vp) = \a(\vp-\sin(\vp))$ and 
$\g$ is an arbitrary constant. 
\end{lemma}

Thus, to~obtain a solution of \erf{sfe-in-x} around $x=0$, it is sufficient to construct a continuous bounded solution of the 
the integral {equation} \erf{ie} in a neighborhood of $x=0$.  The~following lemma establishes the existence of a {1-parameter} family of such
solutions for integral {equation~} \erf{ie}.  

\begin{lemma}\label{lem:l3}
There exists an interval $(-\e,\e)$,   $\,\e>0$ in~which the integral equation 
\beasm
\begin{array}{l}
{
\vp(x) \aea \vp\sw(x) \int_0^x \bsf{\vp\so(y) \, f(\vp(y))}{(y-1) \, y\, W(y)} dy 
%}
%\\
%{
%\ \ \ \ \ \ \ \ \ \ \ \ \ \   
- \vp\so(x) 
\int_0^x \bsf{\vp\sw(y) \, f(\vp(y))}{(y-1) \, y\, W(y)} dy + \g \vp\so(x) 
%\qquad 
}
\end{array}
\label{ie2}
\eeasm
has a 1-parameter family of continuous bounded solutions, $\vp\sb(x;\g)$, labeled by the parameter
$\g$.
\end{lemma}

{Lemmas \ref{lem:l2} and \ref{lem:l3} are proved in Appendix \ref{app:c}.  The~proofs of 
Lemmas  \ref{lem:l2} and \ref{lem:l3} rely on technical Lemmas \ref{lem:l4p5}--\ref{lem:l7}, which are stated and proved in Appendix \ref{app:c}. } 

Lemmas \ref{lem:l2} and \ref{lem:l3}
together show that there is a 
1-parameter family of bounded solutions, $\vp(x;\g)$, of~the differential  
\erf{sfe-in-x} in a neighborhood of the singularity $x=0$.  Equivalently, the~lemmas also show that there is a 1-parameter family of bounded solutions $\vp(z;\g)$ of the differential
equation \erf{sfe-in-z} in a neighborhood of the singularity 
$z =1$.

%%%%%%%%%%%%%%%%%%%%%%%%%%%%%%
\subsection{{Bounded Solutions Outside the~Horizon}}\label{sec:asympt}
As shown in  {Figure} \ref{theasymptopia}, 
bounded solutions approach the  vacuum that is closest to $\vp(1)$ as $z \ra \infty$ regardless
of the value of $\a$.   In~Lemma \ref{lem:l1}{ below,} we prove that  bounded 
solutions, for~all $\a>0$, 
behave as shown in  {Figure} \ref{theasymptopia}.

\begin{lemma}\label{lem:l1}
Let $\varphi(z)$ be a solution of
\begin{equation}
(z^2-1)\varphi_{zz} + 2 z \varphi_z + \alpha \sin(\varphi) = 0, \qquad \alpha > 0
\label{diffeqnz1}
\end{equation}
satisfying the boundary conditions
\bea
0 \slt \vp(1) \slt \pi, \qquad \vp_z(1) \deq - \dsf \a 2 \sin(\vp(1)).
\label{bcon}
\eea 
Then, for all $z \geq 1$,
we have $-\pi <\varphi(z) <  \pi$ and
\begin{eqnarray*}
\lim_{z\rightarrow \infty} \varphi(z) = 0
\end{eqnarray*}
\end{lemma}

Lemma \ref{lem:l1} is proved in Appendix \ref{app:a}.
Note that, if $\vp$ is a solution of 
Eq.~(\ref{diffeqnz1}), then so are $-\vp$
and $\vp + 2n\pi$ for all $n\in \mathbb Z$. {Therefore,} 
we obtain the following~corollary.

\begin{corollary}\label{cor:c1}
Let $\varphi(z)$ be a solution of
\begin{equation}
(z^2-1)\varphi_{zz} + 2 z \varphi_z + \alpha \sin(\varphi) = 0, \qquad \alpha > 0
\label{diffeqnz3}
\end{equation}
satisfying the following boundary conditions.  For~some $n\in \mathbb Z$,
\beas
 (2n-1) \pi \slt \vp(1) \slt (2n+1)\pi,   \qquad 
 \vp_z(1) \deq - \dsf\a 2 \sin(\vp(1)).
\eeas
Then, for all $z \geq 1$, we have
$
(2n-1)\pi <\varphi(z) < (2n+1)\pi$ and
\begin{eqnarray*}
\lim_{z\rightarrow \infty} \varphi(z) = 2n\pi
\end{eqnarray*}
\end{corollary}

%%%%%%%%%%%%%%%%%%%%%%%%%%%%%%
\subsection{Bounded Solutions Inside the~Horizon}\label{sec:is}

 \erf{sfe-in-x} has two length scales---the  
length scale {$m\inv$}   specified by
 $m$, the~`mass' parameter of 
 the field, and~the  length
scale {$H\inv$} of the background spacetime specified by the Hubble
parameter $H$.
We say a soliton is `large' 
{if
$m\inv > H\inv/\sqrt 2$ ({that is,} if $ \a:=(m/H)\sq < 2$) and 
`small' if $\a > 2$. }

Inside the horizon, the~behavior of bounded solutions depends on
whether  $\a<2$
or $\a > 2$.   We consider the two cases 
separately{, below.}
% in the next two~subsections.

\subsubsection{Large Static Solitons Are~Forbidden}\label{sec:lss}

For $\a<2$,  %bounded 
solutions {that are bounded near $z=1$} with $0<\vp(1)<\pi$
cannot attain $\vp(0) = \pi$, as~shown by the examples in  {Figure} \ref{thealphalw},
and blow up at $z=-1$, as~shown by the example in {Figure} \ref{theblowUp}.   
The following theorem shows that topologically nontrivial (soliton) bounded solutions
do not exist for $\a<2$.

\begin{theorem}\label{thm:t1}
For $\a < 2$, every bounded solution of 
\bea
(z\sq-1) \vp\uu{zz} + 2z \vp\uu{z} + \a \sin(\vp) \aea 0
\label{ynk}
\eea
is topologically trivial $($i.e.,~$\lt z \infty \vp(z) \deq 
\lt z {-\infty} \vp(z) )$.
\end{theorem}

The proof of the theorem is lengthy and is presented in
Appendix \ref{app:d}. The~proof relies on technical Lemmas
\ref{lem:le9}--\ref{lem:le11}, which are stated and proved in 
Appendix \ref{app:d}.  We describe the outline of the proof 
of Theorem \ref{thm:t1}, 
below. 

As {discussed} %mentioned 
in Section~\ref{sec:lag}, {without loss of 
generality,} we can assume that a solitonic field 
configuration  with charge $Q$ satisfies the following
conditions:
\beas
\lt z {-\infty} \vp(z) \deq -2\pi Q, \qquad 
\vp(0) \deq -\pi Q,\qquad 
\lt z {\infty} \vp(z) \deq 0
\eeas
Further, given
the symmetries of the Lagrangian, we can restrict 
attention to $z\geq 0$. As~Corollary \ref{cor:c1} shows, if~
we insist that $\lt z \infty \vp(z) = 0$, then we must have
$-\pi < \vp(1) < \pi$. 
 
{We} prove Theorem
\ref{thm:t1} 
{by showing}
that, for~ $\a<2$ {and}  $-\pi < \vp(1)<\pi$, 
every bounded solution satisfies 
$-\pi < \vp(0) < \pi$.  But~a topologically nontrivial soliton,
satisfying the above condition,
{can have} integer nonzero charge $Q$ 
only if $\vp(0) \deq {\pm}\pi Q$.
Thus, when $\a < 2$, topologically nontrivial static solitons
cannot~exist.

%%%%%%%%%%%%%%%%%%%%%%%%%%%%%%
\subsubsection{Existence of Small Static~Solitons}\label{sec:sss}
%\textls[-15]
{The behavior of bounded solutions inside the 
horizon  is  illustrated for~$\a =5$ in  {Figure} \ref{theshoot}.   The~
behavior of bounded solutions for all other $\a>2$ is qualitatively similar. }

{In Lemma \ref{lem:l9} (Appendix \ref{app:e}) we   
confirm that $\vp(0)$ is a continuous function of $\vp(1)$.  Using
the continuity property, in~
Lemmas \ref{lem:l10} and \ref{lem:l11}, we establish that, for~
$\a>2$, one can 
find a constant $0<\g\str<\pi$ such that, if $\vp(1) = \g\str$, then 
$\vp(0) = \pi$.    \mbox{Lemmas  \ref{lem:l10} and \ref{lem:l11},} together 
with Corollary \ref{cor:c1}, prove the following theorem.
}

\begin{theorem}\label{thm:t2}
For every $\a > 2$, the~equation 
\beas
(z\sq -1)\vp\uu{zz} + 2z \vp\uu z + \a \sin(\vp) \aea 0
\eeas
has a {bounded} topologically nontrivial soliton solution.
\end{theorem}

%%%%%%%%%%%%%%%%%%%%%%%%%%%%%
\subsection{Back~Reaction}\label{sec:br}

In the above analysis we assumed that the de Sitter background was non-dynamical and we ignored the back reaction of the soliton on the  background.  In~this 
section, we describe two approaches to incorporating the back reaction of 
a soliton 
%and~
{in order to} make the analysis
self-consistent.  The~latter approach, which we 
adopt,  makes the 
analysis self-consistent and also justifies our
assumption that the {background} de Sitter metric is unaffected
by the static soliton for all values of $\a$.

It is well known that, due to its symmetries, the~Riemann tensor has only one independent
component in 1 + 1 dimensional spacetime and~is given by 
\beas
R\dn{\a\beta\m\n} \aea \dsf{R} 2 \lrr{\gd \a\m \gd \beta\n - \gd \a\n \gd \beta\m}
\eeas
%\textls[-15]
{where $R$ is the Ricci scalar and $g\dmn$ the~metric. 
As a result, the~Einstein tensor $G\dmn$ vanishes  {\it  {for every metric}} in 1 + 1 dimensional spacetime} ({$R\dn{\beta\n} \deq 
\f12 R \gd\beta\n \implies G\dn{\beta\n} = 
R\dn{\beta\n} - \f12 R \gd\beta\n \deq 0$}). %MDPI: \hl{Footnote}s are not supported in our journal. We have therefore included this paragraph in the main text. Please confirm. The following highlights are the same. %Author confirmed.  The italics are needed, and hence retained.
Einstein's equation

\beas
G\dmn + \Lambda g\dmn \aea \kappa  T\dmn, \qquad \qquad \kappa \deq \dsf{8\pi G\uu N}{c\tu 4}
\eeas 
where $\Lambda$ is the cosmological constant and $T\dmn$ the~energy momentum tensor of matter, reduces to 
\bea
\Lambda g\dmn \aea \kappa T\dmn. \leqn{ee}
\eea
Since the Einstein tensor vanishes for every 
metric,  one consistent approach  in~1 + 1 spacetime
would be to take the metric of spacetime to be
the solution of Eq. \ct{ee}.   {If}~$T\dmn$ in 
{ Eq.} \ct{ee}  {comes from the sine-Gordon
field,} {$\phi$},  {then }
$T\dmn \deq \pmd \phi \pnd \phi - g\dmn \cll$. 
{Therefore,} 
\beas
g\dmn \deq \brf\kappa\Lambda \lrs{\pmd \phi\pnd \phi - 
g\dmn \cll} \quad \text{ or } g\dmn \deq \dsf{\brf\kappa\Lambda \pmd \phi\pnd \phi}{1 + \brf\kappa\Lambda\cll}
\eeas 
{Using the sine-Gordon lagrangian}  \ct{sga}, {we conclude that
the metric} $g\dmn$ {is the solution to }
%then the metric $g\dmn$ is the solution to 
%the following nonlinear {equation} ({$T\dmn \deq \pmd \phi \pnd \phi - g\dmn \cll$.  
%Therefore, $g\dmn \deq \brf\kappa\Lambda \lrs{\pmd \phi\pnd \phi - 
%g\dmn \cll}$ or $g\dmn \deq \dsf{\brf\kappa\Lambda \pmd \phi\pnd \phi}{1 + \brf\kappa\Lambda\cll}$. Equation \ct{efg} follows
%by using the sine-Gordon Lagrangian in \ct{sga}.})
\bea
%g\dmn \deq \dsf{\brf\kappa\Lambda \pmd \phi\pnd \phi}{
%1 + \brf\kappa\Lambda\lrs{\f12 \gu\a\beta \pad \phi \pbd \phi - \brf{m\sq}{\beta\sq} (1-\cos(\beta\phi))}.} 
g_{\mu\nu} = \frac{(\frac{\kappa}{\Lambda}) \partial_{\mu} \phi \partial_{\nu} \phi}{
1 + (\frac{\kappa}{\Lambda})[\frac{1}{2} g^{\alpha\beta} \partial_{\alpha} \phi \partial_{\beta} \phi - (\frac{m^{2}}{\beta^{2}}) (1-\cos(\beta\phi))]}. \label{efg}
\leqn{efg}
\eea
In the action \ct{sga}, one would use the metric derived as a solution of \ct{efg}{,} instead of the de Sitter
metric{,} and attempt to construct 
a static soliton solution.  {Even if the nonlinear equation}
\ct{efg} {could be solved,}
 {i}t is not clear, {\it a priori}, 
that, in~such an approach  the~solution of \ct{efg} would yield a de Sitter
metric or even an approximation to a de Sitter metric {(if a soliton interpolates between two
vacua,  then $T\dmn \ra 0$ at spatial 
infinity; therefore, $g\dmn \ra 0$ asymptotically; such a  $g\dmn$ would not be a de Sitter metric).}
%,
%even if the formidable problem could be~solved.

{
The alternative approach, which we adopt, is to
embed   sine-Gordon theory {that is} minimally coupled
to gravity, in~the  
so-called Jackiw--Teitelboim  (JT) theory of gravity in 
1~+ 1 spacetime~\cite{jack,teit,turi,mtya}{. The} JT theory 
%contains a dilaton  {field} {$\psi$} ({JT theory has
{has a term}
\beas
S_{_{JT}} = \int d\sq x \, \smg \, \psi\, (R-\Lambda)
\eeas
{that contributes to the bulk  action, where }   
%and also a surface
%term.  
\mbox{$\psi$ is} the dilaton field, $R$, the~Ricci scalar
and $\Lambda$, the~cosmological constant. There is no 
kinetic term for $\psi$ in the JT theory.}  
%in addition to the 
%sine-Gordon field. %~\cite{jack,teit,turi,mtya}.  
The equation of motion of the dilaton {field %,
%which behaves as a Lagrange multiplier, 
yields}
the constraint 
\beas
R \aea \Lambda
\eeas 
\hs{0.2}{Variation of the total action  with respect to the
metric yields a dynamical equation for the 
dilaton field{, $\psi$,}
%{---obtained by varying 
%the metric, in~ both $S_{JT}$ and also in the sine-Gordon
%action, $S_{SG}$, shown in} Eq.~\ct{sga}, {is  
%followed by integration by parts to
%transfer the covariant derivative operators to act on
%the dilaton field $\psi$---}
in~which the energy-momentum
tensor of the sine-Gordon soliton acts as a 
source term; {the dynamical equation for $\psi$ is
obtained by varying 
the metric, in~ both $S_{_{JT}}$ and in the sine-Gordon
action, $S_{_{SG}}$, shown in} Eq.~\ct{sga}, { 
followed by integration by parts to
transfer the covariant derivative operators to act on
the dilaton field $\psi$.}  Thus, a~matter field, such as the 
sine-Gordon field, does not impact the metric
in JT gravity but, instead, impacts the dynamics
of the dilaton field.  The~metric can be stipulated
to be de Sitter and~is not impacted by the 
presence of the sine-Gordon soliton.  There is
no back reaction of the sine-Gordon soliton
on the metric for~any value of $\a$,
and the conclusions of the 
previous sections, including the threshold $\a=2$, remain unchanged. 
}

{Taking the metric to be de Sitter in the JT gravity,  we can compare the energy
in the static soliton {inside the horizon}  with the 
energy resident in the cosmological constant 
inside the horizon. 
\mbox{A straightforward} derivation of the Riemann tensor, starting with the 
de Sitter metric, shows that the Riemann tensor in de Sitter spacetime
is 
\beas
R\dn{\a\beta\m\n} \aea  H  \sq (\gd\a\m \gd\beta\n - \gd\a\n \gd\beta\m).
\eeas
Hence, we conclude that  $\Lambda = 2H\sq$.
}

{Starting with the action \ct{sga},
and using the previous convention
$\vp(z) := \beta\, \psi(z/H)$, where 
$\phi(t,\chi) := \psi\lrr{u}, \ 
u = e\tu{Ht}\, \chi$, a~straightforward 
calculation yields the energy density  
\beas
T\tu{00} 
\aea {H\sq \over \beta\sq} \lrc{
\f12 (z\sq+1) \vp_z\sq + \a (1-\cos(\vp))}  \deq 
{\Lambda \over 2\beta\sq} \lrc{
\f12(z\sq+1) \vp_z\sq + \a (1-\cos(\vp))}
\eeas  
}

{As is well known~\cite{cole}, the~weak coupling 
regime corresponds to $\beta\sq < 8 \pi$. So, we
set $\beta\sq = \beta\sz\sq \dfn 8\pi-\pot{-6}$. 
Focusing on the regime $\a\sim 2$, we take 
$\a = \a\sz := 2.0001$.   Then, the~ratio of
the soliton's energy within the horizon, denoted
$E\uu S$, to the
energy due to the cosmological constant within
the horizon, denoted
$E\uu \Lambda$, in~units where $8\pi G = c = 1$, is obtained {numerically} 
%{{(the integral was evaluated as a 
%Riemann sum with $\pot 7$ mesh points in
%the range $(-1,1)$ via {MATLAB  R2024b).} 
%%MDPI: Please state which version of the software was used.
%%Author: stated
%  The~computation was done using double precision
%with a machine epsilon of $\scn{2.22}{-16}$. 
%\eqr{sfe-in-z} was integrated using
% {\tt {ode45}}, %MDPI: Please confirm if the monospaced text are necessary; if not, please remove them. Please check entire manuscript.
% %Author: need the monospaced text because it is the name of a builtin MATLAB function.
%a~built-in function in MATLAB.
% The computation was started at $z = 1- \potm{12}$ and~integrated backwards to 
% $z = -1 + \potm{12}$.  The~integrity of the 
% computation was verified by confirming 
% that $\vp(1-\potm{12}) = \vp(-1+\potm{12})$
% at double precision}  
% } 
 as 
\bea
\dsf{E\uu S}{E\uu \Lambda} \aea \dsf{\ds\int_{-1}^1
\dsf{\Lambda}{2\beta\sz\sq} \lrc{
\f12 (z\sq+1) \vp_z\sq + \a\sz (1-\cos(\vp))}\, dz}{2\Lambda} \approx 0.0796
\leqn{esbyl}
\eea
{The integral in} \ct{esbyl} {was evaluated as a 
Riemann sum with $\sim \scn 2   9$ mesh points in
the range $(-1,1)$ using  MATLAB  R2024b.  
%MDPI: Please state which version of the software was used.
%Author: stated
The~computation was done using double precision.
%with a machine epsilon of $\scn{2.22}{-16}$.
} 
\eqr{sfe-in-z} {was integrated using}
 {\tt {ode45}}, %MDPI: Please confirm if the monospaced text are necessary; if not, please remove them. Please check entire manuscript.
 %Author: need the monospaced text because it is the name of a builtin MATLAB function.
{a~built-in function in MATLAB R2024b.
 The computation was started at $z = 1- \potm{9}$ and~integrated backwards to 
 $z = -1 + \potm{9}$; the displacement of $\potm 9$
 was determined by memory constraints on the 
 size of the grid.  The~integrity of the 
 computation was verified by confirming 
 that $\vp_z(1-\potm{9}) = \vp_z(-1+\potm{9})$.}
 Within the horizon, the~energy in the soliton 
field   is $\sim$$8$\% of the 
energy  due to the cosmological constant for~the
above choice of parameter values in the weak
\mbox{coupling regime.}

\section{Heuristic Estimates Based on Tensile and Tidal~Forces}\label{sec:ttf}
{In this section, we present heuristic arguments
based on the interplay of tensile force and tidal force.
On the one hand, the~tensile force of a soliton resists
inflationary stretching of the soliton and~acts to preserve
a soliton as a compact field configuration. On~the other hand,  
the tidal force of the inflationary background  tends to stretch 
the soliton.  If~a soliton is to be static, that is, if~ its physical
dimensions are to buck inflationary stretching,
then its internal
tensile force must be strong enough to counteract the 
tidal stretching of the soliton 
by the background.
The arguments we present are crude. They
have only a qualitative heuristic value. }

{ 
In
Section~\ref{sec:s2}, we use a heuristic argument to
 obtain an estimate of the tensile force in a 
Lorentzian sine-Gordon soliton and{, using the estimate,} 
 a threshold for existence of a
 static sine-Gordon soliton in a de Sitter background; the estimate
we obtain with the crude heuristic argument
agrees with  the exact estimate $\a=2$ to~
within an $O(1)$ factor.}

{ In
Section~\ref{sec:s3}, we use a similar heuristic argument to 
obtain an estimate of the tensile force in a Lorentzian
$SO(3)$ `t Hooft--Polyakov monopole and{, using the estimate,} 
a threshold for existence of static $SO(3)$ `t Hooft--Polyakov
monopole in a de Sitter background.  The~existence of
the threshold,
suggested by our heuristic argument,
remains to be confirmed by an analytical argument.
Guided by the heuristic argument,
 we propose a conjecture that there is an $O(1)$ 
threshold
for the existence of a static 
$SO(3)$ `t Hooft--Polyakov monopole in 
non-dynamical de Sitter
background.  
}
%%%%%%%%%%%%%%%%%%%%%%%%%%%%%%%%%%%%%%%%%%%
\subsection{Sine-Gordon~Solitons}\label{sec:s2} 
% Section~\ref{sec:s1} and~Appendices \ref{app:c} to 
%\ref{app:e} present  a rigorous argument to show
% that static solitons
%exist in 1 + 1 non-dynamical de Sitter spacetime if and only
%{if} %MDPI: Please confirm if the bold formatting in variables/equations is necessary; if not, please remove it. Please check entire manuscript and confirm.
% $\scll\uu{dS} > \sqrt 2 \, \scll\uu{SG}$, where $\scll\uu{dS}$ is
%the characteristic length scale of the de Sitter spacetime
%and $\scll\uu{SG}$ is the characteristic length scale of the
%sine-Gordon field.    In~this section, we re-derive the above
%threshold, up~to an O(1) factor, using a crude argument that
%provides  a {heuristic} explanation for the existence of the above~threshold.   
 
In spacetime with  nonzero Riemann tensor 
the separation between two neighboring{,} initially parallel{,} 
 geodesics accelerates
as one moves along the fiducial geodesic.   Since a soliton is an 
extended object two neighboring points in a {soliton %tend to 
move} along  different geodesics{. In de Sitter spacetime
the separation between two neighboring points grows with 
time; in other words, the de Sitter background pulls the 
points apart.} 
%the~separation between {the points}  
%accelerating{; in de Sitter spacetime the neighboring points
%are pulled apart}.  
If the soliton is to  remain a compact static object, the~internal 
tensile forces of the soliton must be strong enough to resist
 the tidal force of the background
spacetime that seeks to tear the soliton~apart.

A rigorous study of the interplay between the internal tensile
force of a soliton and the tidal force of the background requires
analysis of the soliton configuration in the curved background.  
We consider a simpler problem.  We estimate the tensile force
in a sine-Gordon soliton in 1 + 1 Lorentzian spacetime and
compare it with the tidal force operative in de Sitter spacetime.
We show that  even the rather simple analysis is able to reproduce the 
rigorously derived threshold to within an $O(1)$ factor. 

In a Lorentzian background the sine-Gordon Lagrangian is
\bea
\cll \aea \f12\, \pmd\vp \, \pmu\vp - U(\vp), 
\qquad U(\vp) \dfn m\sq  \lrr{1- \cos(\vp)}
\label{sgl}
\eea
{(The Lorentzian sine-Gordon field 
$\vp$ has been
redefined to absorb the coupling constant $\beta$; the resulting overall
 multiplicative factor, $\beta^{-2}$, is irrelevant to our argument; see  } \erf{al}.{)} 
Time-independent solutions of the Euler--Lagrange equation satisfy
\bea
\vp\uxx(x) \aea U'(\vp)   \label{ssg}
\eea
As Coleman observed, if~we interpret  $x$ as time and~$\vp$ as the position of a particle, then
\erf{ssg} describes the behavior of a
particle of unit mass moving in potential $-U(\vp)$. A~soliton solution 
that interpolates between two adjacent minima of $U(\vp)$
corresponds to the trajectory of the 
particle moving, from~infinite past to infinite future,
 between two adjacent 
maxima of $-U(\vp)$.  Since the energy of the particle is conserved,
and, assuming the potential energy as~well as the kinetic energy
of the particle at~the local maximum of $-U(\vp)$
is zero, we have
{
\bea
\f12\, {(\vp'(x))\sq} \aea U(\vp) \label{ec}  
\eea
}

\noindent One can verify that the time-independent 
topologically nontrivial soliton solution
of the sine-Gordon theory, {$\ti \vp(x)$,}  given by  
\bea
\ti\vp(x) \aea 4 \tan\inv \lrr{e\tu{mx}} \label{ssol}
\eea
satisfies \erf{ec}.
From \erf{ec} we see that the energy density of the 
 soliton $\ti\vp$  is given by
\beas
\cle(x) \aea \dsf12{(\ti\vp{'}(x))\sq}  + U(\ti\vp(x)) \deq 2\, U(\ti\vp(x))   
\deq  16\,m\sq \bsf{ e^{2mx}}{(1+e^{2mx})\sq}
\eeas 
Thus{,} most of the energy of the sine-Gordon soliton is contained within a distance of   $l \sim \xi m\inv$ from the center, where $\xi \sim O(1)$.
We define a system $S$ to be the soliton configuration within a 
fiducial box $[-l,l]$ and the energy contained in the system $S$ as  
\bea
E(l) \aea  
\deq 2 \int_0^l \cle(x)\, dx
\deq 8 m \bsf{e\tu{2ml}-1}{e\tu{2ml} + 1}; \qquad
E(\infty) - E(l) \deq \dsf{16m}{e\tu{2ml} + 1}\qquad
\label{el}
\eea

Consider stretching the soliton as $\ti \vp\uu\lambda(x) \deq \ti\vp(x/\lambda),\, 
\lambda \geq 1$, where $\ti\vp\uu\lambda(x)$ represents the stretched configuration.
As a result of the stretching, the~fiducial segment of the soliton in the 
interval $[-l,l]$ stretches to span $[-\lambda l, \lambda l]$, and~the energy in the
stretched segment, denoted $E\uu\lambda(\lambda l)$,  is given by
\beas
E\uu\lambda(\lambda l) \aea 2 \int_0^{\lambda l} \lrs{\f12 (\ti\vp\uu\lambda'(x))\sq + U(\ti\vp\uu \lambda(x))} dx
\deq \brf 1 2   \lrs{\lambda + (1/\lambda)} \, E(l)
\eeas
The increase in the internal energy in the 
fiducial system as a result of the stretching is 
\bea
\Delta E\uu\lambda \aea E\uu\lambda (\lambda  l) - E(l) \deq \f12 \lrr{\lambda  + \f1 \lambda  -2} E(l)
\eea
%%%
Clearly $E\uu\lambda(\lambda l)$ is minimized at $\lambda = 1$, and the 
increase in the internal energy in the system $S$ can be 
interpreted as the work done on the system  against its
internal tensile force, which resists stretching.  We define
$F(\lambda)$  as the centripetal tensile force acting at the 
boundary of the stretched fiducial segment that has been
stretched by a factor $\lambda$.  Denoting the change in the
size of the system as $q \dfn 2(\lambda l - l)$, we have
\beas
\Delta E\uu\lambda \aea 2 \int_1^\lambda F(v) \, l\, dv
\eeas
from which we obtain
\bea
F(\lambda) \aea \brf 1 {2l} \dbyd{ }{\lambda}\, \Delta E\uu \lambda  \deq \brf 1 {4l} E(l) \lrr{1 - \dsf 1 {\lambda\sq}}
\label{cf}
\eea
$F(\lambda)$ is the centripetal force that the 
system $S$ exerts at $x=\lambda l$ on its complement, which is the 
segment of the stretched soliton beyond $\lambda l$.  The~mass of the complement, denoted $m\uu c$ {in units in which the velocity of light in vacuum, $c=1$}, is
\beas
m\uu C \dfn \f12 \lrr{E\uu\lambda(\infty) - E\uu l(\lambda l)}\aea
\f12 \lrr{\lambda + \f 1\lambda} (E(\infty)-E(l))
\eeas
%\textls[-20]
{Assuming that the complement of the stretched system
behaves like a rigid body, the~centripetal acceleration
at $x=\lambda l$  that arises when the external 
stretching force is turned off is}  
\bea
a_\lambda \aea 
\dsf{F(\lambda)}{m\uu C} \deq 
\f{e\tu{2ml} -1}{4l}     \bsf{\lambda\sq-1}{\lambda(\lambda\sq + 1)}
\label{al}
\eea
{(The overall factor $\beta^{-2}$ that  arises from the redefinition of 
the field (see the remark after} \erf{sgl}{) appears in both $F(\lambda)$
and $m_{_C}$, and is hence irrelevant to the calculation of $a_{\lambda}$
in} \erf{al}, {as we remarked earlier.
)}

The centrifugal tidal acceleration 
%{(Strictly, the~expression for tidal acceleration holds
%only for small displacement.  As~we show below,
%applying the expression for the displacement 
%$x=\lambda l$ is justifiable.)} 
of the displacement 
$\chi\tu\a$ is given {by} %MDPI: References should be numbered in order of appearance. We noticed that “Ref hart (Bibliography 22)” appears after “Ref cole (Bibliography 20)”; please rearrange all the references so that they appear in numerical order.
 \cite{hart}
\beas
a_t\tu\a \aea\nabla\dn u \nabla\dn u \chi\tu\a\deq  - R\tu\a\dn{\beta\m\n} u\tu\beta u\tu\n \chi\tu\m
\eeas
where $u\tu\beta$ and $u\tu\n$ are the tangents to the two 
geodesics and~$\chi\upm$ the displacement between the
two geodesics.  
{It should be noted that  the~expression for tidal acceleration holds
only for small displacement.  We show below that the 
small-displacement assumption is satisfied in our argument; see the 
discussion preceding and following} Eq. (\ref{ila}). 
%As~we show below,
%applying the expression for the displacement 
%$x=\lambda l$ is justifiable. 
}
We take $u\tu\beta = u\tu\n = (1,0)$.  Therefore,
the components of the Riemann tensor of interest to us are 
$R\tu\a\dn{0\m 0}$.    Calculation of the Riemann tensor 
shows the only nonzero component  of $R\tu\a\dn{0\m 0}$
is $R\uo\dn{010} = - H\sq$.   Taking $\chi\tu\a = (0, \lambda l)$, 
the tidal acceleration is
\bea 
a_t \aea H\sq \lambda l
\label{at}
\eea

Consider stretching the sine-Gordon soliton using external force 
by a factor $\lambda $ in 
Lorentzian spacetime, as~ described above. As~a result
of the stretching, centripetal tensile forces arise in the 
soliton, as~described in \erf{cf}.  At~$t=0$ we turn off
the external stretching force and simultaneously turn on the 
expansion of the background de Sitter spacetime with~Hubble parameter $H$.  The~de Sitter background induces
a centrifugal tidal acceleration in the soliton.  The~question we are 
interested in is: {\it  {For what value of $\lambda$ are the 
centripetal tensile acceleration of the soliton 
and the centrifugal tidal acceleration balanced at
the boundary of the stretched fiducial system?}} %MDPI: Please confirm if the italics are necessary; if not, please remove them. The following highlights are the same.

We take the size of the fiducial system $S$ to be $l= m\inv$,
the characteristic size of the sine-Gordon soliton.  
Equating the tensile and tidal accelerations at $l=m\inv$
using Eq.  \eqref{al} and \eqref{at} and~setting $\a = {(m/H)\sq}$,
we obtain
\beas
\dsf{\lambda \sq -1}{\lambda \sq (\lambda \sq + 1)} \aea  q, \qquad q := \bsf{4}{e\sq -1 }\dsf 1 \a \approx \dsf{0.6261}\a 
\eeas
Setting $w = \lambda \sq$, we obtain the equation
\bea
q w\sq + (q-1) w + 1  \aea 0 \label{eqw}
\eea
We need at least one of the solutions of the above quadratic
equation to be positive, which can happen only if 
$q < 1$, that is, $\a \gtrsim 0.6261$. In~addition, since $w=\lambda \sq$
must be real, we must have $(q-1)\sq - 4q \geq 0$,
or, in~terms of $\a \approx 0.6261/q$, we must have
\bea
(\a-0.1074)(\a-3.6490) \ \geq\   0 \label{reqe}
\eea
The inequality (\ref{reqe}) is satisfied if
$
\a \lesssim 0.1074$  or $\a \gtrsim 3.6490
$. 
Since positivity of $w$ requires $\a \gtrsim 0.6261$, we
conclude that  a {real  $\lambda $ exists} only if 
\bea
\a \gtrsim 3.6490 \label{ce}
\eea 
For $\a \lesssim 3.6490$, 
$\lambda $ would be complex, which means that tensile
and tidal accelerations cannot be at equilibrium at the
boundary of the fiducial system $S$ at any
real $\lambda $; in other words, the~system 
cannot be a  
static configuration.  It is worth comparing
the inequality  (\ref{ce}) with the result of the exact analysis
in Section~\ref{sec:s1}, which showed that static soliton exists
only if $\a > 2$.
 
In using the expression for tidal acceleration, we made
the assumption that the displacement $\lambda  l$ was `small'. That is,
in terms of the characteristic length scale of the background,
we assumed $\lambda  l/ H\inv \lesssim 1$.  For~$l \sim m\inv$,
the assumption we made was  
\bea
\lambda  \lesssim \sqrt \a. \label{ila}
\eea
{As verified below, the  assumption is always  satisfied 
by 
the smaller root of} \erf{eqw}, {denoted} $\lambda\so$.
{   At $\lambda\so$, 
 the inequality}
(\ref{ila}) {can be written as}
\bea 
1-q - 2\,q\, \a \slt \sqrt{(1-q)\sq - 4 q} \label{omi}
\eea 
 Since  {$2\,q\, \a \approx 1.2521 > 1$ and $q > 0$}, 
 the~
left-hand side of  (\ref{omi})  is negative and 
inequality  (\ref{omi})  is always satisfied for $\a >$ {3.6490}.  
%for $\a \gtrsim
%3.6490$  
%by the smaller root 
%derived from \erf{eqw}. 

The above  analysis is crude  
since it estimates the tensile force of the soliton in Lorentzian
and not de Sitter spacetime.  

\subsection{$`$t Hooft--Polyakov~Monopole}\label{sec:s3}
In this section we  
compare {a heuristic} estimate of
 the tensile force within the {Lorentzian} `t Hooft--Polyakov
monopole{, in the Prasad--Sommerfeld limit,} and the tidal force 
{the monopole} %it 
experiences in  
de Sitter spacetime to~obtain a {heuristic} threshold for the 
existence of static `t Hooft--Polyakov monopole
in de Sitter spacetime.  
`t Hooft--Polyakov monopoles are solitons of a  
Yang--Mills gauge theory in which the gauge fields
couple to a triplet of Higgs scalars.  
The Lagrangian is 
\bea
\cll \aea - \f14 G\dmn\tu a (G\umn)\tu a+ \f12 D\dnm \vp\tu a
D\upm \vp\tu a- \f\xi 4 (\vp\tu a \vp\tu a - F\sq)\sq, \qquad a=1,2,3\qquad
\label{yml}
\eea 
where {$F$ is the minimum of the Higgs potential, and~}
\beas
(D\dnm \vp)\ua
\aea 
\pmd \vp\ua +  g\, \lsc abc  A\dnm\tu b \vp\tu c, \quad 
G\dmn\ua \deq  
 \pmd A\dnn\ua  - \pnd A\dnm\ua + g \, \lsc a b c A\dnm\tu b A\dnn \tu c, 
 \quad {a,b,c = 1,2,3}
\eeas
The equations of motion are
\beas
(D\dnm D\upm \vp)\ua \aea - \xi (\vp\sq - F\sq)\, \vp\ua, \qquad
 D\dnm (G\tu{\m\n})\tu a \deq  g \lsc abc   (D\upn \vp)\tu b \vp\tu c
\eeas
If we choose the temporal gauge $A\ua\sz \deq 0$ and consider a time-independent solution, then
the energy is given by 
\beas
E \aea \int \, d\cu x \lrc{\f14 (G\dn{ij})\ua (G\dn{ij})\ua + \f12 (D\dn i\vp)\ua (D\dn i\vp)\ua
+ \dsf\xi 4 (\vp\sq - F\sq)\sq}
\eeas
%{\textls[-15]
 where the sum over repeated indices is  implied; that is,   
\beas
(G\dn{ij})\ua (G\dn{ij})\ua \aea  \sum_{i,j,a = 1}^3 (G\dn{ij}^a)\sq, \quad 
(D\dn i\vp)\ua(D\dn i\vp)\ua \deq \sum_{i,a = 1}^3 (D\dn i \vp\ua)\sq.
\eeas   
{In} the Prasad--Sommerfeld limit~\cite{prso}, in~which the Higgs self-coupling vanishes,  the~equations of motion are 
\bea
(D\si D\ui \vp)\ua \aea 0, \qquad D\si (G\tu{ij})\ua \deq g \lsc abc (D\uj \vp)\tu b \vp\tu c
\label{eomt}
\eea
and the energy is
\bea
E\aea \int \, d\cu x \, \clh, \qquad 
\clh \deq \lrc{\f14 (G\dn{ij})\ua (G\dn{ij})\ua + \f12 (D\dn i\vp)\ua (D\dn i\vp)\ua}
\label{ent}
\eea 
Prasad and Sommerfeld showed that the equations of motion \erf{eomt} are satisfied by 
\bea
\ba{ll}
\vp\ua \deq x\ua K(r),  \qquad K(r) \deq  \brf 1 {gr\sq} \lrs{m\uu V \, r \coth(m\uu V r) - 1} 
\\ 
\ \\
A\ua\si \deq \e\dn{aij} x\uj W(r),  \qquad W(r) \deq \brf 1{gr\sq}\lrs{1 - m\uu V r/ \sinh(m\uu V r)}
\ea, \qquad m\uu V \deq gF; \qquad 
\label{soln}
\eea
Further, the~solution shown in \erf{soln} satisfies the Bogomolny equation
\beas
\f12 \lsc ijk G\dn{ij}\ua \aea - (D\tu k\vp)\ua
\eeas
from which we obtain the equipartition of energy density between
the Higgs and \mbox{gauge fields}
\bea
\f12 (D\dn i \vp)\ua (D\dn i\vp)\ua 
\aea \f14 G\dn{ij}\ua (G\dn{ij})\ua \label{epe}
\eea
Using the notation of Eq.  \eqref{ent} and \eqref{soln}   we see that the energy density  of the soliton is spherically symmetric
\beass
\clh(r) \aea \clh\uu A(r) + \clh\uu\vp(r),  \quad 
\clh\uu \vp(r)  \deq \f12 (D\si\vp)\ua (D\si\vp)\ua, \quad 
\clh\uu A(r) \deq \f14 G\dn{ij}\ua (G\dn{ij})\ua 
\eeass
The characteristic length scale of the soliton is $\rho \sim \mv\inv$.
We define the fiducial system, $S$, as~the Higgs and gauge field
configurations in the static soliton within a sphere of radius
$\rho\sim \mv\inv$ centered at the origin. 
The energy in the system and its complement are  
 \beasm
E(\rho) \aea 4\pi \int_0^\rho \, r\sq\,   \clh (r)\, dr, \quad
E\sps c(\rho) \deq 4\pi \int_\rho^\infty \, r\sq\,   \clh (r)\, dr,
\qquad \label{esc}
\eeasm
Following the argument for the sine-Gordon soliton, we stretch the
soliton with 
a dilatation $\vec y \deq \lambda \vec x, \ {\lambda > 1}$.  The~Higgs and gauge fields in the 
new coordinate system are
\beas
\ti\vp(\vec y ) \aea \vp(\vec y/\lambda), \qquad \ti A\si\ua (\vec y) 
\deq \dsf 1 \lambda A\si\ua (\vec y/\lambda)
\eeas
The energy densities of the Higgs and gauge fields in the 
stretched system are denoted $\ti\clh\uu{\ti\vp}(r)$ and 
$\ti \clh\uu{\ti A}(r)$.  
The energy   in the stretched system is 
\beass
\ti E(\lambda\rho) \aea 4\pi \int_0^{\lambda\rho} r\sq \, 
\lrs{\ti\clh\uu{\ti A}(r)  + \ti \clh\uu{\ti\vp}(r)} \, dr
 \deq 4\pi \lambda\cu \int_0^{\rho}  s \sq \, 
\lrs{\brf 1 \lambda\tu 4 \clh\uu{ A}(s)  +  \brf 1 \lambda  \sq\clh\uu{ \vp}(s)} \, ds
\eeass
Using the equipartition of energies shown in \erf{epe} we obtain
\beas
\ti E(\lambda\rho) \aea \f12 \lrr{\lambda + \f 1 \lambda} E(\rho)
\eeas
The increase in the internal energy of the system, as~a result
of stretching, is the 
work done against the negative pressure $P(\lambda)$ in the 
system and~is
\beas
\De E \aea \ti E(\lambda\rho) - E(\rho) \deq 
\f12 \lrr{\lambda + \f 1 \lambda - 2} E(\rho)\deq 
\int_1^\lambda P(\lambda') \, dV 
\deq \int_1^\lambda f(\lambda') \, \rho\, d(\lambda')
\eeas
where $f(\lambda')$,   the~total centripetal restoring force arising
at the boundary of the system stretched by a factor $\lambda'$,
is given by
\bea
f(\lambda) \aea \dsf 1 \rho \dbyd{\De E}\lambda \deq \f1{2\rho} \lrr{1-\f1{\lambda\sq}} E(\rho)
\label{force}
\eea
We assume that the centripetal force at the boundary of the stretched system
acts on the complement of the stretched system, which we 
assume behaves as a rigid body of mass $\bar m$ given by
\beas
\bar m \aea 4\pi \int_{\lambda\rho}^\infty r\sq \lrs{\ti\clh\uu{\ti \vp}(r)
+\ti \clh\uu{\ti A}(r) } \, dr
\deq \f12\lrr{\lambda + \dsf 1\lambda} E\sps c(\rho)
\eeas
where $E\sps c(\rho)$ is defined in \erf{esc}.  The~{magnitude
of the } centripetal { tensile }
acceleration, as~a function of scale factor $\lambda$, is then
\bea
a\dn \lambda \aea \dsf{\lrr{\lambda\sq- 1}}{\lambda \rho \lrr{\lambda\sq + 1}} \, \dsf{E(\rho)}{E\sps c(\rho)} \label{ate}
\eea

To determine the tidal force in 3 + 1 de Sitter, again we take
the tangents of the fiducial and a nearby spatially separated
geodesic to be $u=(1,0,0,0)$.  The~spatial separation between
the geodesics is taken to be $\chi = (0, \vec\chi)$.
Then the relevant components of the Riemann tensor are $R\tu i\dn{\ \ 0 j 0}$.  Straightforward calculation yields
\beas
R\tu i \dn{\ \ 0 j 0} \aea - \de\tu i\dn j \, H\sq
\eeas
where $H$ is the Hubble constant.  Then the {magnitude 
of the centrifugal} tidal acceleration
at the boundary of the stretched system, that is, at~a distance 
$\lambda \rho = \lambda /\mv$ from the origin, is
\bea
a\dn t \aea  {H\sq\, \lambda \, \rho} \deq  H\sq \brf\lambda \mv \label{ati}
\eea
The tensile and tidal accelerations are balanced at the 
boundary of the stretched system if, at $r = \lambda /\mv$,
we have, using Eq.  \eqref{ate} and \eqref{ati},
{
\bea
\dsf{(\lambda \sq -1)}{\lambda \sq\, (\lambda \sq+1)}   \aea \dsf 1 {\kappa \a\sps{tP}} 
\leqn{eql}
\eea
}

\noindent where ${\a\sps{tP}} \dfn  (\mv /H)\sq$,    
and $\kappa \dfn E(\mv\inv)/E\sps c(\mv\inv)$.
Setting $w \dfn \lambda \sq$,
the above equation becomes
\bea
w\sq + (1-\kappa{\a\sps{tP}}) w + \kappa{\a\sps{tP}} \aea 0 \leqn{eqw}
\eea
Using
Eq.  \eqref{soln} and \eqref{epe} it can be verified that 
\bea
\clh(r)
\aea (r\, K'(r) + K(r))\sq + 2\, K(r)\sq\, (1 - g\, r\sq\, W(r))\sq \label{ede}
\eea
{Since $\clh(r) \geq 0$, from~(\ref{esc}) and the definition of $\kappa$, 
it follows that $\kappa \geq 0$.  Using (\ref{ede}) and (\ref{esc}), $\kappa$ can be obtained by numerical
integration %and is 
{to be} $\kappa \approx 1.4820$. }
{Recalling that ${\a\sps{tP}} \geq 0$, consistency of
\ct{eql} requires that $\lambda \sq \geq 1$. From~\ct{eql}, it then follows that 
$\kappa{\a\sps{tP}} \geq 1$ or 
}
\bea
{\a\sps{tP}} \geq \dsf 1 \kappa \ \approx \ 0.6748 \label{c1}
\eea
Further, for~$w=\lambda \sq$ to be real, we must have
\beas
(1-\kappa{\a\sps{tP}})\sq - 4\kappa{\a\sps{tP}} \deq \lrr{\kappa{\a\sps{tP}} - (3+2\sqrt 2)}\lrr{\kappa{\a\sps{tP}} - (3-2\sqrt 2)} \geq 0
\eeas
or 
\bea
{\a\sps{tP}} \ \leq \ 0.1158, \qquad \mbox{or} \qquad {\a\sps{tP}}\geq 3.9328
\label{c2}
\eea
From (\ref{c1}) and (\ref{c2}) we conclude that, if 
${\a\sps{tP}} = (\mv/H)\sq < 3.9328$, then there is no (real) scale factor
at which the tidal and tensile forces can reach equilibrium
at $r\deq \lambda /\mv$.   

{Figure} \ref{thetidal} shows a plot of the ratio $a\uu\lambda /a\uu t$ as a function of {$\lambda$,  for } %${\a\sps{tP}}$ for~
${\a\sps{tP}} = 1, \ldots, 6$.  As~the plots show,
when {${\a\sps{tP}} \lesssim 3.9328$}, the~tidal force is stronger the tensile force at 
$r = \lambda /\mv$ for  every value of scale factor $\lambda $, and~the strong
tidal forces do not permit the existence of static {`t Hooft--Polyakov monopoles}. 

{
The existence of the 
threshold, suggested by the above heuristic argument,
remains to be confirmed analytically.  However,
it {serves} as a plausibility argument
for the \mbox{following conjecture.}

\begin{conjecture}
{Static 't Hooft-Polyakov monopole solution exists 
in non-dynamical de Sitter
spacetime, in the Prasad-Sommerfeld
limit, if and only if  $\a := (\mv/H)\sq \geq \a\sz$,
where $\a\sz \sim O(1)$ is a constant, $\mv$ is the mass
of the vector boson and $H$ is the Hubble 
constant.} 
%There exists an ${\a\sps{tP}}\sz \sim O(1)$ {such that}  for
%all ${\a\sps{tP}} \dfn (\mv/H)\sq > {\a\sps{tP}}\sz$, static `t Hooft--Polyakov monopole solution exists {and} for ${\a\sps{tP}}  < {\a\sps{tP}}\sz$, a 
%static monopole solution does not exist
%in non-dynamical de Sitter spacetime in the Prasad--Sommerfeld
%limit; $\mv$ is the mass
%of the vector boson and $H$ is the Hubble 
%constant. 
\end{conjecture}
}

\begin{figure}[H]
\hspace{-16pt}
\includegraphics[width=5.5in]{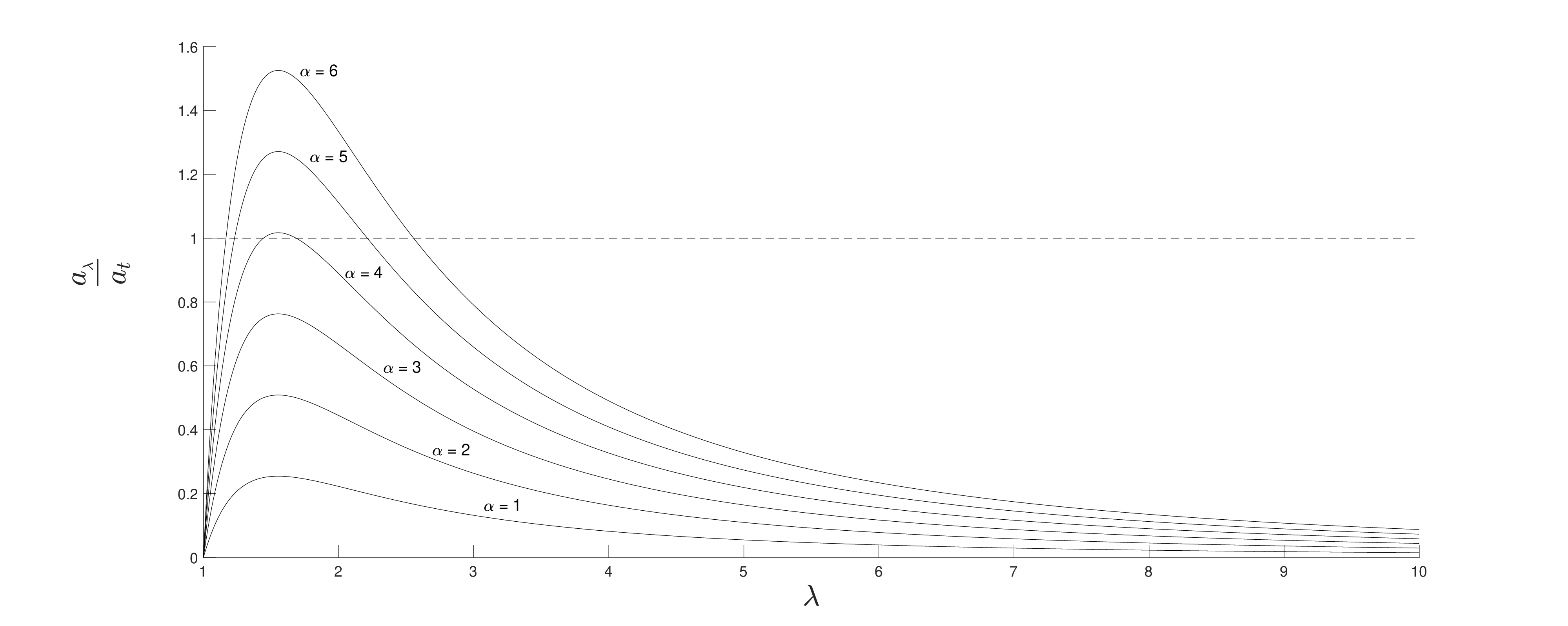}
\caption{{${a\uu\lambda }/{a\uu t}$} %MDPI: We moved figure here to reduce empty space above, please check and confirm.
 as a function of $\lambda $ for~
${\a\sps{tP}} = 1,\ldots, 6$.  {$\a\sps{tP}$ is abbreviated to $\a$
in the labels shown in the plots.} The plots (solid lines) 
show that, for $\a\sps{tP} < 3.9328$,
the tidal forces are too strong{---$a\dn\lambda /a\dn t < 1$---}and do not allow static `t Hooft--Polyakov monopole
to exist in de Sitter spacetime.  The dashed line represents 
${a\uu\lambda }/{a\uu t} \equiv 1.$}\label{thetidal}
\end{figure}   

\section{Secondary~Inflation}\label{sec:s4}
Linde~\cite{li94,li942} has suggested that, although the Hubble constant $H$ 
in new inflation  
is $\sim~10^{10}$ {GeV} %MDPI: Units should not be italics. Please confirm this revision. The following highlight is the same.
, which is about five orders of magnitude smaller than {the mass of the 
$X$ boson,} $\mx \sim 10^{15}$ {GeV}, making $\a\sps{GUT}\dfn (\mx/H)\sq \sim 10^{10}$, 
the inflationary background could trigger secondary inflation
at a {GUT} monopole's core.
Linde's argument is that, near the core, the $X$ boson is nearly massless and  the  field gradients near the core of 
the monopole are rapidly decreased once inflation starts in the background,  creating
conditions suitable for the onset of secondary  inflation at the core of the~monopole.  
  
{In Section~\ref{iotpm}, we 
present a heuristic argument, which suggests that 
a weak inflationary background cannot initiate secondary
inflation at~the core of an $SO(3)$ `t Hooft--Polyakov monopole. 
Based on the heuristic analysis of secondary inflation in the
`t Hooft--Polyakov monopole, we conjecture,  
in Section~\ref{iosufm}, that  the secondary
inflation in the {GUT} monopole,  suggested by Linde, is infeasible.
}

\subsection{Secondary Inflation Inside `t Hooft--Polyakov~Monopole}\label{iotpm}
This section presents a heuristic answer to the following   question: 
{\it   
 {Consider placing an initial field configuration corresponding to the {Lorentzian}
 `t Hooft--Polyakov monopole   in~an inflating background.   Under~what conditions does the inflating background {drive} %exponentially decrease 
 the initial field gradients near~the core of the monopole to~zero?}} %MDPI: Please confirm if the italics are necessary; if not, please remove them.   
  
If inflation is to drive the field gradients   near the core in the initial monopole configuration to zero, then, {as the monopole is stretched by 
inflation,}  the tidal force of the de Sitter background must exceed the
tensile force of the monopole near the core of the monopole.  Specifically,  at~some 
{$\lambda\rho = \lambda\xi \mv\inv$}, where $0 < \xi < O(1)$ {and $\lambda > 1$},
the centrifugal tidal acceleration must exceed the centripetal tensile~acceleration.  

Using the arguments that preceded Eq. ~(\ref{ate}) and 
(\ref{ati}),  and~defining the fiducial system to be the
 Higgs and gauge fields
inside a sphere of radius $\rho= \xi\mv\inv$, 
 the magnitudes of the centripetal 
tensile acceleration  and centrifugal
tidal acceleration, at~$\lambda\rho = \lambda\xi \, m\dn V\inv$, 
denoted $a\dn\lambda(\xi/\mv )$ and $a\dn t(\xi/\mv)$,
after the monopole has been stretched by a 
factor $\lambda$, are
\beas
\begin{array}{l}
{
a\dn\lambda  \lrr{\xi/\mv}  \aea \dsf{\lambda\sq -1}{\lambda (\lambda\sq + 1)} \brf{\mv}\xi 
\dsf{E(\xi\mv \inv)}{E\sps c(\xi\mv\inv)}
}
\\
{
a\dn t\lrr{\xi/\mv} \aea H\sq \brf{\lambda\xi} {\mv} 
}
\end{array}
\eeas
If the tidal force dominates at a distance $\lambda\rho=\lambda\xi/\mv$     
from the center of the monopole  for~some $0<\xi<O(1), \ \lambda>1$, then we must have
\bea
a\dn\lambda  \lrr{\xi/\mv} < a\dn t  \lrr{\xi/\mv},  \ \implies 
\brf{\mv}{H}\sq  \bsf{E(\xi\mv\inv)}{E\sps c(\xi\mv\inv)} \dsf 1 {\xi\sq}\ <  \dsf{\lambda\sq (\lambda\sq+1)}{\lambda\sq-1}
 \leqn{linde0}
\eea
Defining
\bea
q(\lambda) \dfn \dsf{\lambda\sq (\lambda\sq+1)}{\lambda\sq-1}, 
\qquad k(\xi) \dfn \dsf 1 {\xi\sq} \bsf{E(\xi\mv\inv)}{E\sps c(\xi\mv\inv)}
\qquad \label{linde}
\eea
Eq. \ct{linde0} becomes
\bea
\a\sps{tP}\, k(\xi)  <  q(\lambda) \leqn{akq}.  
\eea
Inequality \ct{akq} must be satisfied for the tidal acceleration to dominate
tensile acceleration at~a distance $\lambda\rho=\lambda\xi\mv$
when the monopole has been stretched by a factor $\lambda > 1$.

 {Figure} \ref{thelinde} shows the plot of $q(\lambda)$ and  
{$k(\xi)$}.   $q(\lambda)$ attains its minimum   $q_{min}\sim 5.8285$ at $\lambda\str \sim 1.55$ and 
  {$k(\xi)$} attains its minimum value  $k_{min}\sim 1.4687$ at~
  {$\xi\str \sim 1.1300$}.
 
\begin{figure}[H]
\includegraphics[height=2.5in]{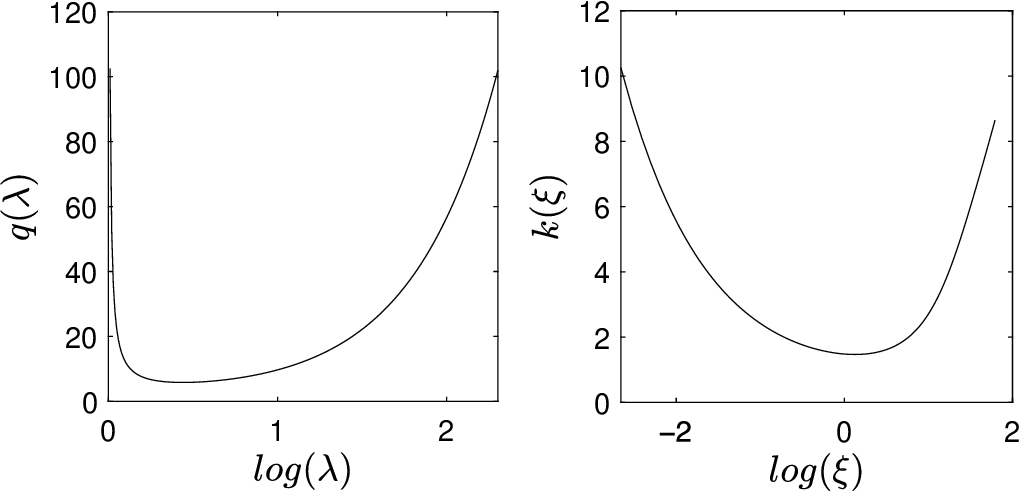} 
\caption{ {Plots} %MDPI: Please change the hyphen (-) into a minus sign (−, “U+2212”) in the figure, e.g., “-1” should be “−1”.
 of $k(\xi)$ and $q(\lambda)$.}\label{thelinde}
\end{figure}

Let  $\a\sps{tP}\sz = q_{min}/k_{min} \approx 3.9685$.   If~$\a\sps{tP} < \a\sps{tP}\sz$, then
\beas
\a\sps{tP} \ala \a\sps{tP}\sz \deq \dsf{q_{min}}{k_{min}} \ \leq  \dsf{q(\lambda)}{k_{min}}, 
\text{ for all $\lambda \geq 1$ at $\xi\str \sim 1.13$}
\eeas
That is, at~a distance $\rho = \xi\str \mv$ from the center, tidal 
acceleration is at least as large as the tensile acceleration for all $\lambda$. 
In other words, there is at least one distance from the center 
($\xi\str \mv\inv$) in 
the unstretched monopole at~
which the tidal force is at least as large as the 
 tensile force for all $\lambda>1$.

{The more interesting case  is   when $\a\sps{tP} > \a\sps{tP}\sz$.   
We note that $q(\lambda)$ diverges as $\lambda \ra 1$.  
As $\lambda$ increases from 1,
$q(\lambda)$ decreases until $\lambda$ reaches $\lambda\str$.  Since $\a\sps{tP} > \a\sps{tP}\sz$, 
there is a  
$\lambda\sz \in (1, \lambda\str)$ at which
\beas
q(\lambda\sz) \aea \a\sps{tP}\, k_{min}
\eeas
For any $\lambda \in (\lambda\sz, \lambda\str)$, 
\beas
q(\lambda) < q(\lambda\sz) \deq \a\sps{tP} \, k_{min} \leq \a\sps{tP} k(\xi), \text{ for all $\xi$.}
\eeas
That is, beyond~the stretching factor $\lambda\sz$, inequality \ct{akq} is not
satisfied at any $\xi$. In~other words, beyond~stretching factor $\lambda\sz$,
the tensile acceleration exceeds the tidal acceleration 
at the  monopole's core
%{We reiterate that the 
%estimate of the tidal force is valid only for `small' displacements.  Thus
%the argument cannot  be trusted far from the center of the monopole. However, we are interested in inflation near the core.}, 
and~stretching of
the monopole by the inflationary background  {ceases}. 
  %{\label{fn:lz}{
{Specifically,}
\bea
\lambda\sz(\a\sps{tP}) 
%\aea \dsf 1 {\sqrt 2} \lrs{(\a\sps{tP} k_{min} - 1) - \sqrt{(\a\sps{tP})\sq k_{min}\sq - 6 \a\sps{tP} k_{min} + 1}}\tuh \\
\aea  
\dsf 1 {\sqrt 2} \lrs{(1.4687\a\sps{tP}   - 1) - \sqrt{2.1571(\a\sps{tP})\sq   - 8.8122 \a\sps{tP}   + 1}}\tuh. \label{lzatp}
\eea
 For $\a\sps{tP} \geq 628$, $\lambda\sz(\a\sps{tP}) < 1$, 
which does not represent stretching. 
$\lambda\sz(\a\sps{tP})$ decreases monotonically as $\a\sps{tP}$ increases from
4 to 627, with~$\lambda\sz(627) \deq 1 + \scn{2.12}{-6}$.
%}).%}

{The  discussion in this section  is  a 
suggestive heuristic argument, which is 
useful  
 only as a plausibility argument for the
conjecture stated in Section~\ref{iosufm}.  
}  

\subsection{Secondary Inflation in {GUT} Monopole}\label{iosufm}

{Unlike the simple 
$SO(3)$ gauge theory of Higgs and gauge fields, used in the 
construction of the `t Hooft--Polyakov monopole,
the $SU(5)$
gauge theory, underlying new inflation, has  Higgs, gauge
and fermion fields; monopole
solution  arises when $SU(5)$ is broken at the
GUT scale.  $SU(5)$ monopole has been 
discussed in the literature~\cite{vach}.  
In the following discussion,
we suggest that a heuristic argument,
analogous to that presented in 
Section~\ref{iotpm},  can be 
developed for $SU(5)$
monopoles.  Again, we restrict attention
to  the 
Prasad--Sommerfeld limit in which the 
coupling constant in the Higgs
potential---which is independent of
the unified gauge coupling constant
of $SU(5)$---goes to zero.}
 
{
Specifically, we can consider a 
fiducial system comprising the fermion, gauge and Higgs fields of an $SU(5)$ monopole,
within a sphere $S$ of radius $\rho = \mx\inv$ from the center of the monopole, where  
$\mx$ is the mass of the $X$ boson that
emerges when the $SU(5)$ symmetry is 
broken.   An~estimate
of the tensile force  in the monopole can
be obtained by regarding the increase in 
energy when the   fiducial system is stretched
by a dilatation factor $\lambda$
as the work done against a centripetal
tensile force acting on the fields outside the
sphere of radius $\lambda\, \mx\inv$.   
The tidal acceleration of the 
de Sitter spacetime  is given by Eq.~(\ref{ati}).}

{The heuristic argument in
Section~\ref{iotpm} suggests that, for~an $SO(3)$
monopole, there is a $\a\sps{tP}\sz \approx 3.9685
\sim O(1)$, and, for
\beas
\a\sps{tP}\dfn 
\brf{\mv}{H}\sq >\a\sps{tP}\sz 
\eeas
the tidal force of the inflationary 
background is too weak to continue stretching
the core of the 
monopole beyond some small
dilatation factor $\lambda\sz$. In~other words,
after the soliton is 
stretched by a factor $\lambda\sz$, the~centripetal
force dominates the tidal force, preventing
further stretching; see the discussion leading up to 
and following \eqr{lzatp}.
%{Footnote} \ref{fn:lz}.}

{It seems plausible that, for the GUT monopole too, there is a $\a\tu{GUT}\sz
\sim O(1)$, and, for 
\beas
\a\tu{GUT} \dfn 
\brf{\mx}{H}\sq > \a\tu{GUT}\sz
\eeas
the tidal force of the inflationary background is 
too weak to continue stretching the core of an
$SU(5)$ monopole beyond some small 
dilatation factor $\lambda\sz\sps{GUT}$. } 

Based on the heuristic argument for the 
$SO(3)$ `t Hooft--Polyakov monopole and~the speculated validity of an analogous
argument for GUT monopole, we state
the \linebreak  following conjecture.

\begin{conjecture}
{For $%\a \tu{GUT}:=
(\mx/H)\sq \gtrsim O(1)$, where $\mx$ is the mass of the 
$X$ boson and $H$ the Hubble constant of inflationary 
background,   
secondary inflation cannot occur 
at the core of  a GUT monopole.}
\end{conjecture}

{
Our conjecture is 
incompatible with Linde's suggestion that,  even 
at  $\a\tu{GUT} \approx \pot{10}$, the~
inflationary background can  stretch the 
core of a   GUT  monopole sufficiently 
to make the field gradients vanish near the core.
Our conjecture 
can be settled with an 
exact analysis of the evolution of a
GUT monopole in the 
de Sitter background that may have 
prevailed during the GUT phase 
transition. }

\section{Discussion}\label{sec:disc}
{Besides the two conjectures stated above,
we mention   two other open problems in~Sections~\ref{sec:stab}
and \ref{sec:end}. The~first problem
pertains to the stability of static solitons of de Sitter spacetime.
The second problem  pertains to the behavior of static solitons of 
de Sitter spacetime when inflation ends.   
We also discuss below, the~ impact of 
quantum corrections and the precision of numerical calculations  
presented in the paper.}

\subsection{Stability of the~Soliton }\label{sec:stab}
{The stability of the soliton solution described
in Section~\ref{sec:s1} is an open problem.
The equation satisfied by  $\de\phi(\chi,t)$,
a  small perturbation of the  soliton solution, can be  derived starting with Eq.~(\ref{ele}); 
not surprisingly, the~equation for $\de\phi(\chi,t)$ depends on 
the soliton solution, which is not available in
closed form.    Establishing the 
stability of the soliton involves showing
that $\de\phi(\chi,t)$ decays for arbitrary initial
conditions. 
}

\subsection{Behavior of Solitons at the End of~Inflation}\label{sec:end}
{When inflation ends, $H\ra 0, \, \a \ra \infty$. If~spacetime becomes
Lorentzian after inflation ends, it is natural  to ask if a sine-Gordon 
soliton solution in 
de Sitter spacetime  approaches the sine-Gordon soliton solution of flat spacetime
as $H\ra 0$.
The answer may be nuanced. In~the $H\ra 0$ limit,
the equation for static solutions in de
Sitter \mbox{spacetime---\eqr{sfe-in-z}---does} not approach the equation for static solutions in flat spacetime.  At~every $H\neq 0$, however small $H$ is, \eqr{sfe-in-z} has 
two singularities at $z= \pm 1$ or $u =\pm H\inv$. 
At $H=0$, the~number of singularities
of the equation for static solutions drops {\it discontinuously} to zero.  Given
the discontinuity, we may not expect the static solitons in de Sitter spacetime
to morph continuously into the static solitons of flat spacetime as~$H\ra 0$.   
  The 
behavior of static solitons in de Sitter spacetime when inflation ends
warrants \mbox{closer examination.}}

\subsection{Quantum~Corrections}

{Quantization of soliton solutions has been 
discussed in~\cite{raja}.  In~the regime \mbox{$\a:=(m/H)\sq < 2$,} where quantum fluctuations are important,
a classical static soliton solution  does not exist. In~
the regime $\a\gtrsim 2$, the~static soliton 
solution in 
de Sittter spacetime, denoted $\phi_s(\chi,t)$, can be regarded as a classical
background for quantum fluctuations $\zeta(\chi,t)$; the classical solution discussed
in Section~\ref{sec:s1} can then be retained unchanged.
The equation for $\zeta(\chi,t)$, derived from 
\eqr{ele}, is linear---in the small $\lrb{\zeta}$ limit---but with a damping term, which is  related to the suppression of quantum fluctuations
by the inflation in the background.  
}  

\subsection{Numerical~Calculations}\label{sec:numerics}

{The numerical calculations reported in this paper were done in MATLAB {R2024b} using its default double precision arithmetic.  The~double 
precision's  machine  epsilon  is \mbox{$\de\sz =\scn{2.22}{-16}$;} 
that is, $\de\sz$ is the smallest number for which  $1$ and $1+\de\sz$ are regarded as different numbers in double precision
arithmetic.  The~machine epsilon determines
how closely one can approach the singularity $z=1$ in numerical investigations.
We solved the differential equation using the   built-in routine, {\tt {ode45}},
which expects the second-order differential equation, such as \eqr{sfe-in-z},
to be presented as a system of first-order equations, such as
\beas
\begin{array}{l}
{
 \vp'(z) \aea \psi(z)
 }
 \\
{
\psi'(z) \aea -\bsf{2z \psi(z) +\a \sin(\vp(z))}{(z+1)(z-1)}
}
\end{array}
\eeas
If $|z-1|<\de\sz$, then, numerically, it is treated as $z-1=0$, and the numerical
calculation leads to an overflow error.  
}
{Therefore, we integrated the differential equation only up to $z=1 \pm 
\potm{15}$ to ensure that the numerical calculations
remained within the limits set by the machine epsilon.  The~differential
\eqr{sfe-in-z} was solved separately on either side of $z=1$.}

{None of the results
presented in this paper {relies} on the exactness of our
numerical calculations.  The~results of the numerical
calculations  have only heuristic value and serve  only as 
plausibility arguments for the conjectures we propose. }

%%%%%%%%%%%%%%%%%%%%%%%%%%%%%%%%%%%%%%%%%%
\vspace{6pt}

%\funding{ {This research was supported by the Office of Naval Research grant number N-00014-96-1-0281. }} %Please add: ``This research received no external funding'' or ``This research was funded by NAME OF FUNDER grant number XXX.'' and  and ``The APC was funded by XXX''. Check carefully that the details given are accurate and use the standard spelling of funding agency names at \url{https://search.crossref.org/funding}, any errors may affect your future funding.

%\dataavailability{
\section{Data Availability}
Data sharing is not applicable to this article as no datasets 
were generated or analyzed during the current~study. 

%\acknowledgments{
\section{Acknowledgment}
I thank Alan Guth for the many helpful conversations
while this work was being done \cite{prab}.
{The research reported in the paper was done mostly 
at the Center for Theoretical Physics (CTP), MIT. I thank CTP
for its support. 
%{I thank} %MDPI: Please ensure that all individuals included in this section have consented to the acknowledgement.
% Alan Guth for the many helpful conversations %while this work was
%%being done 
%{\cite{prab}.  
I also thank Sergiu Moroianu for 
 helpful conversations.} I  thank the three anonymous
referees for their many helpful comments, questions 
and suggestions, which have improved both the content and the presentation of
the paper.     
This research was supported by the Office of Naval Research grant number N-00014-96-1-0281.  No LLM was  used to generate 
any portion of the text or calculations reported in this paper. The 
author takes full intellectual responsibility for the contents of the 
paper. 
%and ONR grant N-00014-96-1-0281.} %MDPI: Attention Assigned Editor: please confirm whether the funding information in the Acknowledgments section should be moved to the Funding section. Please ensure the contents in funding and acknowledgments sections should not be duplicated
%}

%\conflictsofinterest{ {The author declares no conflicts of 
%interest.  The funders had no role in the design of the study; in the collection, analyses, or interpretation of data; in the writing of the manuscript; or in the decision to publish the results.}} %Please declare conflicts of interest or state ``The authors declare no conflicts of interest.'' Authors must identify and declare any personal circumstances or interest that may be perceived as inappropriately influencing the representation or interpretation of reported research results. Any role of the funders in the design of the study; in the collection, analyses or interpretation of data; in the writing of the manuscript; or in the decision to publish the results must be declared in this section. If there is no role, please state ``The funders had no role in the design of the study; in the collection, analyses, or interpretation of data; in the writing of the manuscript; or in the decision to publish the results''.

%\appendixtitles{yes}
%\appendixstart
\appendix
\section{{Solutions} %MDPI: We revised A Appendix to Appendix A, please check and confirm. The same for all Appendix heading.
of Hypergeometric Differential~Equation}\label{app:b}
This appendix presents the construction of the two linearly independent
solutions of the hypergeometric {equation} %MDPI: Figures, tables, and equations in the Appendix should be marked as “Figure A1, Table A1, Equation (A1)”, etc. We have corrected them. Please confirm these revisions.
\bea
x(x-1) \vp\uxx + ((a+b+1) x- c) \vp\uxo + ab \vp \aea 0 \label{hde}
\eea
in a neighborhood of $x=0$ for the special case in which
$a+b=1, \, ab = \a, \, c =1$ or
\bea
a \deq \dsf{1+\sqrt{1-4\a}}2, \quad b \deq \dsf{1-\sqrt{1-4\a}}2, \quad c = 1 \label{spc}
\eea
For the special values of $a,b,c$ shown in (\ref{spc}),
\erf{hde}   
becomes
\bea
x(x-1) \vp\uxx + (2x-1) \vp\uxo + \a \vp \aea 0
\label{sfe-in-xx}
\eea
The first solution of \erf{hde}{, which }
that is analytic at $x=0${,} is given
by the hypergeometric series \cite{whwa} {(\S14.2, p. 283)} %MDPI: We revised reference citation, please check and confirm. The following highlight is the same.
\beas
F(x; a,b,c) \ada \sum_{n=0}^\infty \dsf{(a)\dn n (b)\dn n}{n! \, (c)\dn n} \ x\tu n, \qquad (p)\dn n \dfn 
a(a+1) \ldots (a+n-1), \quad (p)\sz \deq 1
\eeas
%\textls[-20]
{which converges for $|x|<1$.  We denote the  solution
of \erf{sfe-in-x} that is analytic at $x=0$}
as
\bea
\vp\so(x) \dfn F\lrr{x; \dsf{1+\sqrt{1-4\a}}2,\dsf{1-\sqrt{1-4\a}}2,1} \label{phi-1}
\eea
We construct the 
second solution of \erf{sfe-in-x},
 denoted $\vp\sw(x)$,  as~follows (see \cite{whwa} {(\S 10.32)}).  
The  $\vp\sw(x)$  we construct will be linearly  independent with respect to $\vp\so(x)$
and analytic in a punctured interval $P\dn\de \dfn$ [$-\e,\e$]$\setminus\lrc{0}$ for some $\e>0$.

In a punctured interval  of $x=0$ that excludes $x=0$
we write $\vp\sw(x)$  as
\bea
\vp\sw(x) \aea v(x) \, \vp\so(x) \label{vp2ans}
\eea
where $v(x)$ is a non-constant function of $x$.
Substituting (\ref{vp2ans}) into \erf{sfe-in-xx},  
noting that $\vp\so$ satisfies  \erf{sfe-in-xx} and 
setting $w(x) \dfn v'(x)$, we get, for~the 
special values of $a,b,c$ shown in (\ref{spc}),
\bea
\lrs{x(x-1)\, \vp\so} w' + \lrs{2x(x-1) \vp\so' + 
((2x-1) \vp\so} w \aea 0 \label{weqn}
\eea
The differential equation~(\ref{weqn}) is singular
at $x =0$.    Let $z\str$ denote the zero of $\vp\so(x)$ 
that is closest to the origin.  We choose a $\de $
such that $0 < \de < \min(z\str, 1)$.  In~the punctured
interval $(-\de,\de)\setminus\lrc 0$ the solution of
\erf{weqn} is
\bea
w(x) \aea \dsf{c}{x (1-x) \vp\so\sq(x)}, \qquad c \neq 0
\label{w}
\eea
where $c$ is a constant. If~$c=0$, then $v(x)$ would be a 
constant and $\vp\sw(x)$ would not be linearly independent
of $\vp\so(x)$.  Since we are constructing the function $v(x)$,
we can set $c=1$.

In the interval $I\dn\de \dfn (-\de,\de)$, 
$ \lrs{(1-x) \, \vp\sq\so(x)}\inv$
is an analytic function.  Further, $\vp\so(0) \deq 1$.  Therefore
we can expand the function $\lrs{(1-x)\, \vp\so\sq(x)}\inv$ in Taylor series around $x=0$ within
the interval $I\dn \de$ as 
\bea
\dsf 1{(1-x) \, \vp\so\sq (x) } \aea 1 + \sum_{n=1}^\infty q\dn n x\tu n, \dfn 1 + q(x)\qquad 
-\de < x < \de
\label{q}
\eea
where $q(x)$, the~function represented by the series, is analytic
in $I\dn\de$. 

From Eq.  \eqref{w} and \eqref{q} we have (recalling that we set $c=1$)
\bea
w(x) \aea v'(x) \deq \dsf 1 x + \sum_{n=1}^\infty q\sn x\tu{n-1}
\label{vpr}
\eea
Since the series shown in (\ref{vpr}) converges in $I\dn\de$, it 
converges uniformly and absolutely in a compact sub-interval $ 
K\dfn [-\e,\e] \subset (-\de,\de)$,
and therefore the series can be integrated term by term in $K$.
Since $v'$ is singular at $x=0$, we consider the differential
equation separately in the two sub-intervals
$[-\e,0)$ and $ (0,\e]$.  Integrating
(\ref{vpr}) we obtain
\bea
v(x) \aea \left\{ \ba{ll} \log(|x|) + \ti q(x) +  c\so, \quad x \in [-\e,0)\\
\ \\
\log(|x|) + \ti q(x) +  c\sw, \quad  x \in (0,\e] \ea\rtd, \qquad 
\ti q(x) \deq \sum_{n=1}^\infty \brf{q\sn} n x\tu n
\label{v}
\eea
Since we are interested in $\vp\sw(x)$, which is linearly
independent of $\vp\so(x)$, we can ignore the constants
$c\so, c\sw$ in \erf{v}  and  
we have
\beasm
\vp\sw(x) \aea  \log(|x|) \, \vp\so(x) + h(x),     
\qquad x\in [-\e,\e]\setminus
\lrc{0}
\label{phi-2a}
\eeasm
where $h(x) = \ti q(x) \, \vp\so(x)$, and $\ti q(x)$ is derived from the function
$1/(1-x)\, \vp\so\sq(x)$ as described above.  Finally, the~condition
for $\vp\sw(x)$ to be linearly independent of $\vp\so(x)$ is that
the Wronskian $\vp\so \vp\sw' - \vp\so' \vp\sw \deq v' \vp\so\sq \neq 0$, which holds if $c\neq 0$ in \erf{w}, as \mbox{mentioned above.}

\section{{Neighborhood} of the~Horizon}\label{app:c} 

%\begin{proof}
{\bf Proof of Lemma \ref{lem:l4}:}
Rewriting \erf{dex} and using Mean Value Theorem we have for $x\in\intvl\e\setminus\lrc 0$
\beasm
\vp\sb'(x) \aea -\bsf 1{x(x-1)} \izx{\a\, \sin(\vp\sb(y))} \deq 
\dsf {\a\, \sin(\vp\sb(x'))}{1-x}    
\qquad 
\label{ubd}
\eeasm
where $0<|x'|<|x|$. From~(\ref{ubd}) we have
\beas 
\lt x 0\ \vp\sb'(x) \aea \a\, \sin(\vp\sb(0))
\eeas 
\qed %\end{proof}

\noindent 
We need the following Lemmas \ref{lem:l4p5}--\ref{lem:l7}
 to prove Lemmas \ref{lem:l2} and \ref{lem:l3}.  
 Lemma \ref{lem:l4p5} is standard lore, and proofs for it can be found in many
 books on advanced calculus. We   state and prove it below to make it easier
 to reference it repeatedly in later arguments.
 \begin{lemma}\label{lem:l4p5}
If a function $v(x)$ is continuous in \intvl\e$\setminus\lrc 0$, not defined at
$x=0$ and $\lt x 0 v(x) = L$ exists, then  
defining
 $v(0) \dfn L$ makes $v(x)$ well defined and continuous in \intvl\e.
 
 Secondly, 
if $v(y)$ is a continuous bounded function in the interval \intvl\e,\  
 then  
 \beas
 I(x) \dfn \int_0^{ x}\, v(y)\, dy
 \eeas
  is a continuous function at every $x\in\intvl\e$. 
  \end{lemma}

%  \begin{proof}[
\noindent{\bf Proof:} % of Lemma \ref{lem:l4p5}:}
%  ]
Since $L$ is the limit of $v(x)$ as $x\ra 0$, given a $\de>0$, there exists a $\beta>0$ such that, for $\norm{x} < \beta$, $ \norm{L-v(x)} =\norm{v(0)-v(x)}<\de$.   
 
%\textls[-15]
{ For any given $\de>0$ we can find a $\beta>0$ such that, for all $z$ satisfying
 $\norm{x-z} < \beta$, $\norm{I(x)-I(z)} < \de$ because~$\norm{\int_{z}^x \, v(y)\,  dy}< M \norm{x-z} < M \beta$, where $\norm{v(y)} < M$ for \mbox{$y\in\intvl\e$.}}
 The claim follows by choosing   $0<\beta < \de/M$.  \qed
 %\end{proof}
 
 \begin{lemma}\label{lem:l5}
 Let $\vp\so(x)$ and $\vp\sw(x)$ be the two linearly independent solutions of 
 the hypergeometric equation
 \bea
 x(x-1) \vp\uxx + (2x-1)\vp\uxo + \a \vp \aea 0 \label{hom-eq-x}
 \eea
 shown in Eq. \eqref{phi-1} and \eqref{phi-2a}.  Let $W(x) \deq \vp\so \vp\sw' - \vp\so' \vp\sw$ denote their Wronskian. Define $g(x) \dfn (x-1)\, x \, W(x)$.   Let the function $h(x)$ in \erf{phi-2a}
 be analytic in $(-\rho\uu h, \rho\uu h)$, $0<\rho\uu h <1$.  
 
 Then, 
$
g(0) \deq -1
$,
and there exists a $0<\de<\ruh$, and, for  
 $x\in (-\de, \de)$,
 $g(x)$ is an analytic function and $|g(x)|\geq \f12$. 
 \end{lemma}
 
% \begin{proof}
  \noindent {\bf Proof:} 
Differentiating $\vp\so$ and $\vp\sw$,   $g(x)$ 
 can be  {written as} 
% {( $x\vp\sw'(x)$ contains a term of the form  $q(x) \dfn x \ \dbyd{\log(|x|)}x \ \vp\so\sq(x)$,
% which has a {\it removable singularity} at $x=0$.  By~defining 
% $q(0) \dfn 1$, the~function $q$ can be made analytic throughout
% $(-\de,\de)$, including at $x=0$.  \linebreak  See \cite{chur} {(\S 68)} for details.} %MDPI: The \cite{chur} is invalid, please check and revise. Please ensure that all reference appear in numerical order.
 \beas
 g(x) \aea  (x-1)\lrs{\vp\so\sq(x) + x \, \vp\so(x)\, h'(x) - x\, \vp\so'(x) \, h(x)}
 \eeas
  {We note that}  $x\vp\sw'(x)$  contains a term of the form   $q(x) \dfn x \ \dbyd{\log(|x|)}x \ \vp\so\sq(x)$,
 which has a {\it removable singularity} at $x=0$.  By~defining 
 $q(0) \dfn 1$, the~function $q$ can be made analytic throughout
 $(-\de,\de)$, including at $x=0$.  %\linebreak}  
 {See} \cite{chur}  {(\S 68)}  for details.   
 Since $\vp\so, \vp\so', h, h'$ are analytic in $(-\ruh,\ruh)$, it follows that 
 $g(x)$ is analytic in the same neighborhood.  Noting that $\vp\so(0) \deq 1$
 we conclude that $g(0) = -1$.  Since $g(0)=-1$ and $g$ is a continuous function
 in $(-\ruh,\ruh)$ there exists an interval $(-\e,\e)$,   $0<\e<\ruh<1$ in which $|g(x)| \geq \f12$. \qed
 %\end{proof}

  \begin{lemma}\label{lem:l6}
  Define 
  \bea
  V(x;\vp) \dfn \vp\sw(x) \int_0^x \bsf{\vp\so(y) \, f(\vp(y))}{y\,(y-1)\, W(y)} dy - \vp\so(x)\int_0^x \bsf{\vp\sw(y) \, f(\vp(y))}{y\,(y-1)\, W(y)} dy
  \label{defv}
  \eea
  where $f(\vp(y)) \dfn \a(\vp(y)-\sin(\vp(y)))$.  Consider an interval 
  $(-\e,\e), \ 0<\e<\ruh<1$ in which $\norm{y \, (y-1)\, W(y)}\geq \f12\   ($see Lemma \ref{lem:l5}$)$.  If~$\vp(y)$ is a continuous
  bounded function in $(-\e,\e)$, then,
   in the interval $(-\e, \e)$,
 $V(x;\vp)$ can be written as
{\small
\beas
V(x;\vp) \aea h(x) \int_0^x \vp\so(y) \, q(y) \,dy - \vp\so(x) \int_0^x h(y) \, q(y)\, dy + \vp\so(x) \int_0^x \dsf 1 y \lrs{\int_0^y \vp\so(u) \, q(u) \, du} dy
\eeas
}

\noindent
where
\beas
q(y) \ada \dsf{\a(\vp(y)-\sin(\vp(y))}{y\, (y-1) \, W(y)}
\eeas
\end{lemma}

\noindent {\bf Proof:} 
First, we observe that 
$q(y)$
is a continuous bounded function of $y$ in \intvl\e. 
Using the form of $\vp\sw$ shown in \erf{phi-2a} and performing the 
differentiation and partial integration, we obtain
\beasm
\begin{array}{l}
{
V(x;\vp) \aea  \lrc{\vp\so(x) \, \log(|x|) + h(x)}\izx{\vp\so(y)\, q(y)}
- \vp\so(x)\ltd \lrc{\log(|y|) \int_0^y {\vp\so(u) \, q(u)} \, du}\right|^x_0
}
\\
{
%\nonumber\nln
 \ \ \ \ \ \ \ \ \ \  \ \ \ \ \ \ \ \ \ + \vpo(x) \izx{ \dsf 1 y \lrs{\izy{\vpo(u)\, q(u)}}} - \vpo(x) \izx{h(y)\, q(y)}
\qquad 
}
\end{array}
\label{exe}
\eeasm
If $v(x)$ is a continuous function in an interval $(-\e,\e)$, then, using the 
Mean Value Theorem \cite{wrsp} {(Chapter 5)}, we have 
\bea
\lt x 0 \ \log(|x|) \izx{v(y)} \aea 0 \label{mvt}
\eea
The claim in the lemma follows from applying \erf{mvt} in \erf{exe}. 
\qed %\end{proof}

 \begin{lemma}\label{lem:l7}
  Let $\ti \vp(x)$ be an analytic function in the interval $(-\rh,\rh)$ for~some 
  $\rh > 0$.  Then 
  \beas
  V(x;\ti \vp) \aea \vpw(x) \izx {\dsf{\vpo(y) \, f(\ti\vp(y))}{(y-1)\, y\, W(y)}}
  - \vpo(x) \izx {\dsf{\vpw(y) \, f(\ti\vp(y))}{(y-1)\, y\, W(y)}}
  \eeas
  is an analytic function in 
  the interval $(-\de,\de)$ for some $0<\de < \min(\rh,\ruh)$.
 \end{lemma}

\noindent{\bf Proof:} 
From Lemma \ref{lem:l5} we know that there exists an interval
 $(-\de,\de)$, $0<\de<\ruh<1$ in which $(y-1)\, y \, W(y)$ is analytic and 
 nonzero.    
 Choosing a $\de${,} satisfying $0 \slt \de \slt \min(\rh, \rh_h)$, we
 conclude that
 \beas
 q(y) \ada \dsf{f(\ti\vp(y))}{(y-1)\, y\, W(y)}
 \eeas
 is analytic in \intvl\de.   We can  expand $\vpo(x)\, q(x)$ in
 Taylor series around $x=0$ within its domain of analyticity as
 \bea
 \vpo(x) \, q(x) \aea \sum_{n=0}^\infty \brf{a\sn}{n!} x\tu n, \qquad x\in \intvl\de
 \qquad \label{srs0}
 \eea
 Since a power series converges uniformly and absolutely in a sub-interval of 
 its interval
 of convergence \cite{wrsp}  {(Theorem 9, Chapter 11)}, %MDPI: Please confirm if Theorem 9, Chapter 11 belongs to \cite{wrsp}, not this manuscript.
 we can integrate the
 power series term by term to get the integral of the power series 
 \cite{wrsp}  {(Theorem 7, Chapter 11)} %MDPI: Please confirm if Theorem 7, Chapter 11 belongs to \cite{wrsp}, not this manuscript.
 within a sub-interval $[-\de',\de']
 \subset (-\de,\de)$.
 Thus,
 \bea
\dsf 1 x  \izx{\vpo(x)\, q(x)} \aea \sum_{n=0}^\infty \bsf{a\sn}{(n+1)!} x\tu n,
\qquad x\in [-\de',\de']\qquad \label{srs1}
 \eea
 Since series (\ref{srs0}) converges in $[-\de',\de']$, it satisfies the Cauchy
 convergence criterion, and,
 for any $\e>0$, we can find an $n\sz > 1$ such that, for $m,n > n\sz$ and 
 $x\in [-\de',\de']$,
 \beas
 \norm{\sum_{k=m}^n \bsf{a\sk}{(k+1)!} x\tu k} < \norm{\sum_{k=m}^n  \bsf{a\sk}{k!} x\tu k} < \e
 \eeas
 showing that the series (\ref{srs1}) converges in $[-\de',\de']$ by the 
 Cauchy~criterion.  
 
  Next, we define a function 
 \bea
 p(x) \ada \sum_{n=0}^\infty \bsf{a\sn}{n+1} \dsf{ x\tu n}{n!}, \qquad x\in [-\de',\de'],  
  \label{defp}
 \eea
 %%%%%%%%
%If a power series $S$  {has a radius} of convergence
% {$R$}, then we know that
% the series $S'$ obtained by differentiating
% $S$ term by {term
%% , as~well as the series $\ti S$ obtained by
%% integrating $S$ term by term, also converge and have
% also has the same 
% radius  of convergence  $R$}  (see \cite{wrsp}  {(Chapter 11)}). 
%%for 
%%a proof of the claim about $S'$; the claim  about 
%% $\ti S$ can be established using the comparison test). 
%%%%%%%%
If a power series $S$ converges with an interval of convergence
 $I$ {(e.g., $(-\delta,\delta)$)}, then we know that
 the series $S'$ obtained by differentiating
 $S$ term by term, as~well as the series $\ti S$ obtained by
 integrating $S$ term by term, also converge and have the same 
 interval of convergence $I$  (see \cite{wrsp}  {(Chapter 11)} for 
a proof of the claim about $S'$; the claim  about 
 $\ti S$ can be established using the comparison test). 
Therefore, we have  
 \beas
 \dbyd{^k \, p(0)}{x\tu k} \aea \dsf{a\sk}{k+1}
 \eeas
 and  (\ref{defp}) is a Taylor expansion of $p(x)$ about $x=0$, 
 making $p(x)$  an
 analytic function in  {$[-\de',\de']$}.  That is,
 \beas
 \dsf 1 x \izx{\vp\so(y)\, q(y)}
 \eeas
 is an analytic function in  {$[-\de',\de']$}.
% \cintvl{\de'}.  
 The~claim in the lemma
 follows from Lemma \ref{lem:l6}.  
\qed %\end{proof}

\noindent{\bf Proof of Lemma \ref{lem:l2}:} 
We define
 \beas
 A\sj(x; \vp)  \ada \izx {\dsf{\vp\sj(y)\, f(\vp(y)}{(y-1)\, y\, W(y)} }, \qquad 
 j = 1,2 
 \eeas
 Then \erf{ie} can be written as
 \beas
 \vp(x) \aea \vpw(x) A\so(x; \vp) - \vpo(x) A\sw(x; \vp) + \g \vpo(x)
 \eeas
 and 
 \beasm
\begin{array}{l}
{
 \vp\sb' \aea \vpw'  A\so(x;\vp\sb) - \vpo'  A\sw(x;\vp\sb) + \g \vpo', \label{phip}}
 \\
 {
 \vp\sb'' \aea \vpw''  A\so(x;\vp\sb) - \vpo''  A\sw(x;\vp\sb) + \g \vpo'' + \dsf{f(\vp\sb)}{x(x-1)}\nonumber
}
\end{array}
 \eeasm
 Defining the differential operator 
 \beas
 D \ada x(x-1) \dbyd{\sq}{x\sq} + (2x-1) \dbyd{} x + \a
 \eeas
 and differentiating at $x\in \intvl\e\setminus\lrc 0$ we obtain
 \bea
 D[\vp\sb(x)] \aea D[\vp\sw(x)] A\so(x;\vp\sb) - D[\vp\so(x)]\lrc{A\sw(x;\vp\sb) -\g}
 + f(\vp\sb(x))
 \label{ome}
 \eea
 $\vp\so$ and $\vp\sw$ being solutions of the hypergeometric equation
 $D[\vp\sw]=D[\vp\so]=0$, and the first two terms vanish and we get the 
 equation
 \bea
 D[\vp\sb(x)] \aea f(\vp\sb(x)), \qquad x\in \intvl\e\setminus\lrc 0
 \label{dqe}
 \eea
 which is \erf{deqx}. Defining
 \beas
 F(x) \ada x(x-1)\vp\sb''(x) + (2x-1)\vp\sb' (x)+ \a\sin(\vp\sb(x))
 \eeas
 \erf{dqe} also shows that $F(x) = 0$ for $x\in \intvl\e\setminus\lrc 0$. 
 Therefore,
 \beas
 \lt x 0 \ F(x)\aea 0
 \eeas
 Therefore we can define $F(0) \dfn 0$, and \erf{deqx} is satisfied  for all $x\in\intvl\e$ \mbox{by $\vp\sb(x)$.}
 
 Next, we show that $\vp\sb'$ is well defined and continuous in \intvl\e.   
 As shown above, $\vp\sb$ satisfies \erf{deqx}, which can be rewritten as
 \bea
 \vp\sb'(x) \aea -\bsf 1 {x(x-1)} \izx {\a\, \sin(\vp\sb(y))}
 \label{eqvpp}
 \eea 
From Lemma \ref{lem:l4p5} we conclude that 
$\vp\sb'(x)$ is well defined and continuous at every $x\in\intvl\e\setminus\lrc 0$.
Further recall that, if a function $v(x)$ is
 continuous in \intvl\e$\setminus\lrc 0$ and~$\lt x 0 v(x) \deq L$, then defining
 $v(0)\deq L$ makes $v(x)$ well defined and  continuous.  
 (Given a $\de>0$, there exists a $0<\beta<\e$, and, for $\norm{x} < \beta$, $\de>\norm{L-v(x)} =\norm{v(0)-v(x)}$.)  
 Thus, to establish the continuity of $\vp\sb'(x)$ over \intvl\e, it is
 sufficient to show that $\lt x 0 \vp'(x)$ exists.  Using Mean Value Theorem we 
 have
 \beas
 \lt x 0 \vp\sb'(x) \aea \lt x 0 \lrc{-\bsf 1{x(x-1)} \izx {\a \, \sin(\vp\sb(y))}}\\
 \aea \lt x 0 \lrc{-\bsf 1{x(x-1)} \, x \, \a\, \sin(\vp\sb(x')),} 
 \deq \a\, \sin(\vp\sb(0))\qquad x'\in(0,x) 
 \eeas
 Thus we define 
 \beas
 \vp\sb'(0) \ada \lt x 0 \vp\sb'(x) \deq \a\, \sin(\vp\sb(0))  
  \eeas
  and with that definition
 $\vp\sb'(x)$ is well defined and continuous for all $x\in\intvl\e$. 
\qed %\end{proof}

 \noindent%{\bf Proof:} 
{\bf Proof of Lemma \ref{lem:l3}:} 
 Using the notation in Lemma \ref{lem:l6} we can write \erf{ie2} as 
 \bea
 \vp(x) \aea V(x;\vp) + \g \vpo(x) \label{cie}
 \eea
We show that \erf{cie} has the required solution by showing that the 
following recurrence equation converges uniformly to a fixpoint, 
starting with an analytic function $\vp\sps 0(x)$.
\bea
\vp\sps n(x) \aea V(x; \vp\sps{n-1}) + \g \vpo(x) \label{re}
\eea
From Lemma \ref{lem:l7} we know that, if $\vp\sps{n-1}(x)$ is analytic in
some interval \intvl\rh, where $0< \rh < \ruh<1$, then $\vp\sps n(x)$ is also 
analytic in the same interval. 
It then follows that, if we start with a function $\vp\sps 0(x)$ that is analytic
in some \intvl\rh, with~$0<\rh< \ruh< 1$, then all of the functions
$\vp\sps n(x), \ n=1,2,\ldots$ generated by
the recurrence equation \erf{re} will also be analytic in \intvl\rh.  Since 
$\vp\so(x)$ is analytic in $\intvl 1$  
we choose our $\vp\sps 0 (x) \dfn \vp\so(x)$. 

Since the sequence of functions $\vp\sps n(x), \ n=0, 1, 2, \ldots$ are analytic in
\intvl\rh, they are bounded in every closed sub-interval \cintvl\e $\subset$
\intvl\rh, $0<\e<\rh$.  We fix such an $\e$ \mbox{and define} 
\beasm
M\sn(\e) \ada \max_{-\e\leq x\leq \e} \norm{\vp\sps n(x) - \vp\sps{n-1}(x)}
\deq \max_{-\e\leq x\leq \e} \norm{V(x;\vp\sps{n-1}) - V(x;\vp\sps {n-2})}
\qquad \label{mn}
\eeasm
Further,   
 $h(x)$ and $\vpo(x)$ and $1/x (x-1) W(x)$ are 
analytic functions in \intvl \rh, and, if $\rh$ is chosen to be sufficiently
small as we assume it is, then, from Lemma \ref{lem:l5}, we know that
$1/x (x-1) W(x)$ is analytic in $\intvl\rh$ as well.  Therefore,
$h(x), \vpo(x)$ and $1/x (x-1) W(x)$ are also bounded in \cintvl\e,  and~we
define
\beasm
C\uu h(\e) \aea \max_{-\e\leq x\leq \e} \norm{h(x)}, \ 
C\uu p(\e) \deq \max_{-\e\leq x\leq \e} \norm{\vpo(x)}, \ 
C\uu w(\e) \deq \max_{-\e\leq x\leq \e} \norm{\dsf 1{(x-1) \, x\, W(x)}} 
\qquad \label{con}
\eeasm
We note that
\bea
\begin{array}{l}
{
\norm{f(\vn)-f(\vnm)} \aea \a \norm{\vn - \sin(\vn) - \vnm + \sin(\vnm)}
}
\\
{ \ \ \ \ \ \ \ \ \ \ \ \ \ \ \ \ \ \ \ \ \ \ \ \ \ \ \ \ \ \ \ \ \ \ \ \ \ \ \ 
\alea \a \norm{\vn-\vnm} + \a \norm{\sin(\vn) - \sin(\vnm)}
}\\
{ \ \ \ \ \ \ \ \ \ \ \ \ \ \ \ \ \ \ \ \ \ \ \ \ \ \ \ \ \ \ \ \ \ \ \ \ \ \ \ 
\aea \a \norm{\vn-\vnm} + \a \norm{\cos(\vp\sps{n,n-1}) \lrs{\vn -\vnm}}
}\\
{ \ \ \ \ \ \ \ \ \ \ \ \ \ \ \ \ \ \ \ \ \ \ \ \ \ \ \ \ \ \ \ \ \ \ \ \ \ \ \ 
\alea 2\a \norm{\vn-\vnm} 
}
\end{array}
\label{if}
\eea
In the third step, we have used the Mean Value Theorem, and~$\vp\sps{n,n-1}
\deq (1-\m)\, \vn + \m\, \vnm$ for some $\ 0\leq \m \leq 1$. 

Using Lemma \ref{lem:l6} and Eq.  \eqref{mn}--(\ref{if}), we have
\beas
\begin{array}{l}
{
M\sn(\e) \alea 
2 C\uu h(\e) C\uu p(\e) C\uu w(\e) \lrs{2\a M\uu{n-1}(\e)} \, \e + 
C\uu p(\e)\sq  C\uu w(\e) \lrs{2 \a M\uu{n-1}(\e)} \e
}
\\
{\ \ \ \ \ \ \ \ \ \ \ \,\, 
\aea \lrc{ \lrs{4 C\uu h(\e) C\uu p(\e)  C\uu w(\e)\, \a + 2 C\uu p(\e)\sq C\uu w(\e)\, \a} \e } M\uu{n-1}
}
\end{array}
\eeas
As $\e$ decreases, so do $C\uu h(\e), C\uu p(\e), C\uu w(\e)$.  Therefore
\beas
\lt \e 0   \lrs{4 C\uu h(\e) C\uu p(\e)  C\uu w(\e)\,\a + 2 C\uu p(\e)\sq C\uu w(\e)\,\a} \e \aea 0
\eeas
which means that, if we choose a small $\e\str >0$, then 
\beas
\lrs{4 C\uu h(\e\str) C\uu p(\e\str)  C\uu w(\e\str)\,\a + 2 C\uu p(\e\str)\sq C\uu w(\e\str)\,\a} \e\str < 1
\eeas
and, as $n\ra \infty,$ $M\sn(\e\str) \ra 0$, which means that, in the interval 
\cintvl{\e\str}, the sequence of analytic functions $\vp\sps 0, \vp\sps 1, 
\ldots$ converges uniformly to a fixpoint of the recurrence relation \erf{re}.

Since analytic functions are continuous, the~uniform limit of a sequence of analytic functions over \cintvl{\e\str} is a continuous function over \cintvl{\e\str}
\cite{tama} {(\S 20.3, Theorem III)}.  %MDPI: Please confirm if this Theorem III belongs to ref {tama}, not this manuscript.
Further, since all of the functions $\vp\sps 0, \vp\sps 1, 
\ldots$ are bounded over \cintvl{\e\str}, so is their uniform limit $\vp$ since, for~any
$\de>0$, we can find a $n\sz$ such that, for all $n\geq n\sz$
\beas
\max_{-\e\str \leq x\leq \e\str}\norm{\vp\sps{n\sz}(x) -\vp(x)} < \de, 
\eeas
The boundedness of $\vp(x)$ over \cintvl{\e\str} then follows from the boundedness of $\vp\sps{n\sz}$ over the~interval.

Thus, we conclude that the uniform limit of the sequence of analytic
functions bounded over \cintvl{\e\str} yields a continuous bounded
function, which being a fixpoint of \mbox{\erf{re}} is a solution of \erf{ie2}. 
Since we have demonstrated the existence of a solution of \erf{ie2}
for every value of $\g\in \mathbb R$,
we have proved the existence of a 1-parameter family of solutions
for \erf{ie2}.  
\qed %\end{proof}

\section{{Beyond} the~Horizon}\label{app:a}

%\begin{proof}
{\bf Proof of Lemma \ref{lem:l1}:}
  Consider the function
\begin{eqnarray*}
V(z;\varphi,\varphi_z) := {\varphi_z^2 \over 2} + {\alpha (1-\cos(\varphi)) \over z^2 - 1}.
\end{eqnarray*}
Differentiating $V$ and using \erf{diffeqnz1} we get
\begin{eqnarray}
{dV \over dz} & = & {\partial V \over \partial \varphi}\cdot \varphi_z +
{\partial V \over \partial \varphi_z}\cdot \varphi_{zz} + {\partial V \over \partial z}
 = - {2z \over z^2 -1} \left[{\varphi_z^2 \over 2} + V \right].
\label{dVbydz}
\end{eqnarray}
Using the boundary condition (\ref{bcon}), we see that 
\begin{eqnarray*}
B(z;\varphi,\varphi_z) := (z^2-1)V(z;\varphi,\varphi_z) = {z^2-1 \over 2} \cdot \varphi_z^2 + \alpha \cdot (1-\cos(\varphi))
\end{eqnarray*}
is well defined at $z=1$.
Using Eq. ~(\ref{dVbydz}) and~(\ref{diffeqnz1})
we have
\begin{eqnarray}
{dB \over dz} = - z \varphi_z^2.
\label{dQbydz}
\end{eqnarray}
which shows that
$B(z;\varphi,\varphi_z)$ is a monotonically decreasing function of $z$ for~$z> 1$.  
Since $0 < \varphi(1) < \pi$, we have, for $z>1$,
\begin{eqnarray}
B(z; \vp, \vp_z) \leq B(1;\varphi,\varphi_z) = \alpha \cdot (1-\cos(\varphi(1))) < 2 \alpha  
\leqn{bat1}
\end{eqnarray}

Next, we show that $-\pi < \varphi(z) < \pi$ at every finite $z>1$.  Arguing by contradiction, let us
assume that $\varphi(z_*) = \pi$ (or $\varphi(z_*) = -\pi$) at~some finite $z_*>1$.  Then
\begin{eqnarray*}
B(z^*;\varphi,\varphi_z) = {(z_*^2-1) \varphi_z^2(z_*) \over 2} + \alpha \cdot (1-\cos(\varphi(z_*)))
\geq 2\alpha  
\end{eqnarray*}
which contradicts   
\ct{bat1}.  Therefore,  $-\pi < \varphi(z) < \pi$ for~all $z > 1$.

Finally we show that $\varphi(z)$ vanishes as $z \rightarrow \infty$. For~$z \gg 1$,
\erf{diffeqnz1} behaves effectively as
\begin{equation}
z^2 \varphi_{zz} + 2z \varphi_z + \alpha \sin(\varphi) = 0
\label{diffeqn-at-large-z}
\end{equation}
Defining $t:= \ln z$ and   $X(t) := \varphi(e^t)$,
Eq.~ (\ref{diffeqn-at-large-z}) can be rewritten as
\begin{equation}
X_{tt}(t) + X_t(t) + \alpha \sin(X(t)) = 0
\label{damped-motion}
\end{equation}
If we interpret $X(t)$ as the position of a particle of unit mass at time $t$,
then Eq.~ (\ref{damped-motion}) describes motion of the particle in a potential
$\alpha (1-\cos(X))$ in the presence of frictional force. If~the particle starts
at $0 \leq X(0)=\gamma < \pi$, then the preceding argument shows that $-\pi < X(t) < \pi$
for all $t>0$. Owing to the frictional force the particle eventually comes to
rest at the local minimum of the potential, $\varphi=0$, as~$t,z \rightarrow \infty$. 
\qed %\end{proof}

\section{{Large}~Solitons}\label{app:d}
\begin{lemma}\label{lem:le9}
Let $\psi(t)$ be a solution of
\begin{equation}
\ddot{\psi} = \alpha e^{-t} \psi; \qquad \alpha \in {\bf R}
\label{linearizedeqn}
\end{equation}
that satisfies the boundary condition $\psi(t_0)\neq 0$ for~some $t_0 > 0$.  Then
\eas{
\displaystyle\lim_{t\rightarrow\infty} \psi(t) \neq 0
}
\end{lemma}

\noindent{\bf Proof:} 
%\textls[-20]
{The lemma is trivial when $\alpha=0$.  We consider the two
cases $\alpha < 0 $ and $\alpha > 0$ separately.}

Let $\alpha < 0$. Define
\eas{V(\psi,\dot\psi,t) := {\dot\psi^2\over 2} + {|\alpha| \, \psi^2\,  e^{-t} \over 2}}
Using \eqnref{linearizedeqn} we have
\eas{
{dV\over dt} = {\partial V \over \partial \dot\psi} \cdot \ddot\psi + {\partial V \over \partial \psi} \cdot \dot\psi
+ {\partial V \over \partial t} = - V + {\dot \psi^2 \over 2}
}
Multiplying both sides with $e^t$ and rearranging we get
\eas{ {d \left[\, e^t \, V \right] \over dt} = {e^t \, \dot\psi^2 \over 2}}
which shows that $e^t\, V$ increases monotonically with $t$. Therefore, if~$\psi$
oscillates about zero, with~$\dot\psi$ vanishing at $t_1 \slt  t_2 \slt \ldots$, then
\bea
|\psi(t_1)| < |\psi(t_2)| < \ldots \label{divo}
\eea
That is, if~$\psi$ oscillates about zero, then the oscillation is divergent.  We
show that, if $\psi(t_0) \neq 0$ and $\lim_{t\rightarrow \infty} \psi(t) = 0$,
then $\psi$ must oscillate about zero.   Then, from~(\ref{divo}), we obtain a contradiction,
which will  prove the~claim.

First we consider the case $\psi(t_0) > 0$.  If~ $\lim_{t\rightarrow \infty} \psi(t) = 0$,
then  there exists a \mbox{$\ti t \in[t\sz,\infty)$} at which 
$\psi(\ti t) > 0$ and $\dot\psi(\ti t) < 0$. 
Since $\alpha <0$,
\eqnref{linearizedeqn} implies that $\ddot\psi < 0$ when $\psi>0$. 
Therefore, $\psi(t') = 0$
for some $\ti t < t' < \ti t + {\psi(\ti t)\over \rule{0pt}{10pt}|\dot\psi(\ti t)|}$. Further,
$|\dot\psi(t')| > |\dot\psi(t_0)|$.  If~$\lim_{t\rightarrow \infty} \psi(t) = 0$, then 
we must have $\psi(t'') < 0, \ \dot\psi(t'') > 0$ at some $t''>t'$ and~subsequently
$\psi(t''') = 0$ for~some $t'''>t''$. Again, $|\dot\psi(t''')| > |\dot\psi(t'')|$. Repeating the above argument we
conclude that $\psi$ must oscillate about zero, and, from (\ref{divo}), we conclude
that the oscillations are divergent, 
which contradicts the assumption that
$\psi$ approaches zero asymptotically.  A~similar argument shows that $\psi$ oscillates if $\psi(t_0) < 0$
as well, completing the proof of the lemma for $\alpha < 0$.

Let $\alpha >0$. We will restrict the following argument to the
case $\psi(t_0) > 0$.  The~argument for $\psi(t_0) < 0$ is similar to the one presented
below.

We start by showing that, if $\psi$ is to vanish asymptotically, then $\dot\psi(t_0)<0$, and,
for $t>t_0$, $\psi(t)$ must be monotonic. If~$\dot\psi(t_0)\geq 0$, then \eqnref{linearizedeqn} implies that $\dot\psi(t) > 0$ for
all $t> t_0$, contradicting the assumption that $\psi$ vanishes asymptotically. To~prove the
second claim, we argue by contradiction. If~possible,  let $\psi(t') = 0$ for some $t'>t_0$.
Further assume that we have
chosen the smallest   $t'>t_0$ at which $\psi$ vanishes. Clearly $\dot\psi(t') \leq 0$.
If $\dot\psi(t') = 0$, the~solution to \eqnref{linearizedeqn} with initial conditions
$\psi(t') = \dot\psi(t') = 0$ would not be unique in the neighborhood of $t'$ (since
$\psi\equiv 0$ is a solution). On~the other hand, if~$\dot\psi(t') < 0$, then $\ddot\psi(t)<0$
for all $t>t'$, making it impossible for $\psi \rightarrow 0$ as $t\rightarrow \infty$.
Therefore, we conclude that $t'$ cannot exist and $\psi$ must decay to zero~monotonically.

Next we show that, for all $t>t_0$,
\eae{|\dot\psi(t)| \leq \alpha\, e^{-t}\, |\psi(t)| \label{inequality1}}
Using \eqnref{linearizedeqn} we have, for any $t'>t$,
\eas{\psi(t') = \int_t^{t'} \int_t^\tau \alpha\, e^{-\lambda} \, \psi(\lambda)\,  d\lambda \, d\tau +
(t'-t)\,  \dot\psi(t) + \psi(t)}
If $\psi(t_0) > 0$, then, using the monotonic decay of $\psi$ established above, for~any $t'>t$,
\eas{0 & < & \psi(t') < \a\,\psi(t)\,   \int_t^{t'} \int_t^\tau   e^{-\lambda } \, d\lambda \, d\tau +
(t'-t)\,  \dot\psi(t) + \psi(t) \nonumber \\
& \leq & \a\, \psi(t) \, e^{-t} \int_t^{t'} d\tau + (t'-t)\dot\psi(t) + \psi(t)\nonumber\\
& < & \left\{\alpha\, \psi(t)\, e^{-t} + \dot\psi(t)\right\} (t'-t) + \psi(t)}
Recalling that $\dot\psi(t) <0$, the~claim~follows.

Integrating inequality (\ref{inequality1}) we get
\eas{
-\int_{\psi(t_0)}^{\psi(t)} \frac{d\psi}{\psi} \leq \alpha \int_{t_0}^t e^{-\tau}\, d\tau
}
or
\eae{
\psi(t) \geq \left\{\psi(t_0) e^{-\alpha\, e^{-t_0}}\right\} e^{\alpha\, e^{-t}} \label{inequality2}
}
from which we obtain a lower bound
\eas{|\psi(t)| \geq \left\{|\psi(t_0)| e^{-\alpha\, e^{-t_0}}\right\}}
for all $t>t_0$, contradicting the assumption that $\psi$ vanishes asymptotically.
The contradiction completes the~proof.

%\textls[-20]
{The preceding argument can be used to establish inequalities
(\ref{inequality1}) and (\ref{inequality2})  even when $\psi(t_0)<0$. Therefore
a solution satisfying  $\psi(t_0)<0$ cannot vanish asymptotically either.} 
\qed %\end{proof}

\begin{lemma}\label{lem:le10}
Let $\varphi$ be a solution of
\begin{equation}
  (z^2-1)\varphi_{zz} + 2 z \varphi_z + \alpha \sin \varphi = 0
  \label{diffeq2}
\end{equation}
that satisfies the boundary condition $\varphi(1)=n\pi$. Then
$\varphi(z) \equiv n\pi$.
\end{lemma}

\noindent{\bf Proof:} 
Clearly $\varphi(z) \equiv n\pi$ is a solution
of Eq.~ (\ref{diffeq2}) that satisfies the boundary condition $\varphi(1)=n\pi$.
Therefore all we need to show is that $\varphi(z) \equiv n\pi$ is the only solution satisfying the
given boundary condition. If~$\varphi(z)$ is a solution of Eq.~ (\ref{diffeq3}),
then so is $\varphi(z) + 2m\pi$, where $m$ is an integer. Therefore it is sufficient
to  establish   uniqueness of the solution $\varphi(z) \equiv n\pi$ for~$n=1$ and $n=0$.

It is convenient to work with the hypergeometric coordinate $x= \frac{1-z}{2}$,
in terms of which Eq.~ (\ref{diffeq2}) can be rewritten as
\begin{equation}
  x(x-1)\varphi_{xx} + (2x-1) \varphi_x + \alpha \sin \varphi = 0.
  \label{diffeq3}
\end{equation}
The singularity at $z=1$ in Eq.~ (\ref{diffeq2}) corresponds to the singularity at
$x=0$ in~\mbox{Eq.~ (\ref{diffeq3}).}  We prove the uniqueness of the solution $\varphi\equiv n\pi$ in a
neighborhood ($-\epsilon,\epsilon$) for~some sufficiently small $\epsilon>0$ by
showing that, if $\varphi(\epsilon) \neq 0$ (resp. $\varphi(-\epsilon)\neq 0$),
then $\varphi(0) \neq 0$.

First we consider the limit $x\rightarrow 0^+$. It is convenient to change the
coordinate to $w = -\log(x)$. As~$x\rightarrow 0^+$, $w\rightarrow \infty$.
In terms of variable $w$, \eqnref{diffeq3} can be
\mbox{rewritten as}
\eas{
{d\over dw}\left[(1-e^{-w}){d\varphi\over dw} \right]=\alpha \, e^{-w} \sin(\varphi)
}
If $\epsilon$ is chosen to be sufficiently small, then $w_\epsilon := - \log(\epsilon) \gg 1$
and the above differential equation effectively becomes
\eae{
{d^2\varphi\over dw^2} = \alpha e^{-w}\sin(\varphi)\label{diffeqn5}}

Arguing by contradiction,
if $\varphi(w) \rightarrow 0$, as~$w\rightarrow \infty$, then,
for any $\delta>0$, it is possible to find a $w_\delta > w_\epsilon$ such that
$|\varphi(w)| < \delta$ for all $w\geq w_\delta$.  If~$\delta$ is chosen to be sufficiently small,
then $\sin(\varphi) \approx \varphi$ for $w\geq w_\delta$, and~the above differential
equation effectively becomes
\eae{{d^2\varphi\over dw^2}=\alpha \, e^{-w}  \varphi\label{diffeqn4}}
If $\varphi(w_\epsilon) \neq 0$, then there can be no finite $w'\geq w_\epsilon$ such that
$\varphi(w)=0$ for~all $w>w'$; if such a $w'$ exists, then the solution to
\eqnref{diffeqn4} would not be unique in an open neighborhood of the smallest such
$w'$. Therefore, there must be a $w^*>w_\delta$ at~which $\varphi(w^*) \neq 0$.
Lemma \ref{lem:le9} then implies that $\varphi\not\rightarrow 0$ as $w\rightarrow \infty$,
contradicting our assumption.  The~preceding argument shows that, if $\varphi(x=0) = 0$, then
$\varphi(x) =0$ for $0\leq x\leq \epsilon$, establishing uniqueness in [$0,\epsilon$)
for $n=0$.

Uniqueness in [$0,\epsilon$) for $n=1$ follows immediately by observing that, if
$\varphi$ is a solution of \eqnref{diffeqn5}, then   $\psi:=\pi-\varphi$
satisfies
\eae{{d^2\psi\over dw^2} = -\alpha e^{-w}\sin(\psi)\label{diffeqn6}}
If $\varphi(\epsilon) \neq \pi$ and $\varphi(w) \rightarrow \pi$ as $w\rightarrow \infty$, then $\psi(\epsilon) \neq 0$
and $\psi(w) \rightarrow 0$ as $w\rightarrow \infty$. Arguing as we did above,
the \eqnref{diffeqn6} can be replaced by
\eas{{d^2\psi\over dw^2} = -\alpha e^{-w} \psi }
for  $w>w_\delta$.  As~we showed above, we can find a $w^* > w_\delta$
at which $\psi(w^*) \neq 0$. Since $\psi(w^*)\neq 0$ for some
finite $w^*$, Lemma \ref{lem:le9} then implies that $\psi(w) \not\rightarrow 0$ as $w\rightarrow \infty$.

In order to investigate the limit $x\rightarrow 0^-$ it is convenient to set
$w=-\log(-x)$.  In~terms of $w$, \eqnref{diffeq3} can be
rewritten as
\eas{{d\over dw}\left[(1+e^{-w}) {d\varphi\over dw}\right] = - \alpha e^{-w} \sin(\varphi)}
Choosing $\epsilon$ to be sufficiently small as we did before, $w_{\epsilon}\dfn 
-\log(\e)$ can be made
sufficiently large, and, in the interval $(w_\e,\infty)$, the above equation effectively becomes
\eas{{d^2\varphi\over dw^2} = -\alpha e^{-w}\sin(\varphi)}
The arguments presented for $n=0$ and $n=1$ over the interval [$0,\epsilon$)
establish uniqueness for $n=1$ and $n=0$, respectively, over the interval ($-\epsilon,0$]. 
\qed %\end{proof}

\begin{lemma}\label{lem:le12}
Let $\vp(x)$ be a solution of 
\bea
x(x-1) \vp\uxx(x)+ (2x-1)\vp\uxo(x) + \a\, \sin(\vp(x)) \aea 0 
\label{yaex}
\eea
satisfying the boundary condition $ 0 \slt \vp(0) \slt \pi$.
If $\vp(x) < \pi$ at every  $x\in[0,a] \subseteq [0,\f12]$,
then $\vp(x) > 0$  at every $x\in [0,a]$.
\end{lemma}
The lemma asserts that $\vp$ cannot reach 0 before it
reaches $\pi$.   That is, the~behavior shown in 
 {Figure} \ref{theforbidden} is forbidden by \erf{yaex}.

\vspace{-6pt}
\begin{figure}[H]
 \includegraphics[height=2in]{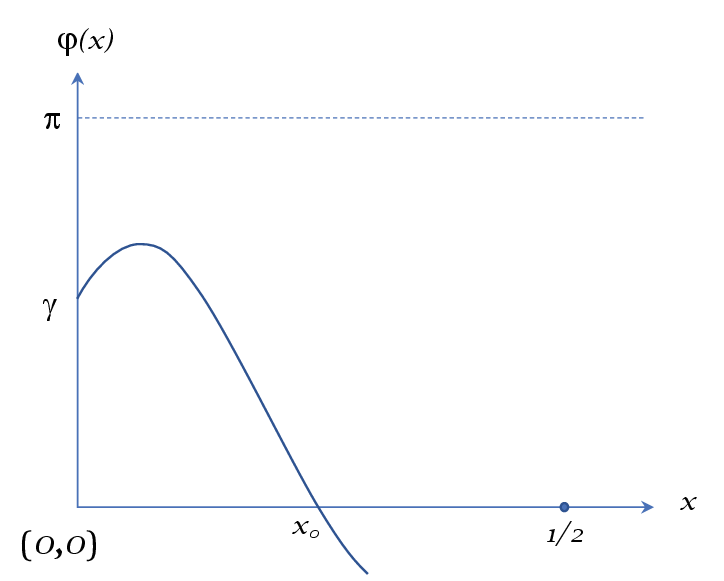}
\caption{{The} %MDPI: Please confirm whether an explanation of the dotted and solid lines needs to be added to the figure caption.
{behavior illustrated by the solid line is forbidden 
for a solution of} Eq. (\ref{yaex}).   {The dotted line represents}
 $\vp(x) \equiv \pi$.}  \label{theforbidden}
\end{figure}   

\noindent{\bf Proof:} 
 Assume to the contrary that there is a $x\sz\in (0,\f12]$
 at which $\vp(x\sz) = 0$ and \linebreak  $0<\vp(x) <\pi$ for all $x \in 
 [0,x\sz)$. 
 Then, rewriting \erf{yaex}, we see that 
 \beas
 \vp'(x) \aea \dsf 1{x(1-x)}  \int_0^{x}{\a \, \sin(\vp(y))} \, dy  \ > \ 0, 
 \qquad x\in(0,x\sz]  
 \eeas
Since $\vp' > 0$ at all $x\in (0,x\sz]$ and  $\vp(0)>0$, 
we conclude that $\vp(x\sz) > \vp(0) > 0$, which 
contradicts the assumption that $\vp(x\sz) \deq 0$.  The~
contradiction proves the claim.   
\qed %\end{proof}

Restating Lemma \ref{lem:le12} in static coordinates
we obtain the following corollary.
\begin{corollary}\label{cor:c4}
Let $\vp(z)$ be a solution of 
\bea
(z\sq-1) \vp\uu{zz}(z) + 2z\,\vp\uu z(z) + \a, \sin(\vp(z)) \aea 0 
\label{yaee}
\eea
satisfying the boundary condition $ 0 \slt \vp(1) \slt \pi$.
If $\vp(z) < \pi$ at every  $z\in[a,1] \subseteq [0,1]$,
then $\vp(z) > 0$  at every $x\in [a,1]$.
\end{corollary}

\begin{lemma}\label{lem:le11}
For every bounded solution of 
\bea
(z\sq-1) \vp'' + 2z \vp' + \a\,\sin(\vp(z)) \aea 0
\label{yaez}
\eea
satisfying the boundary condition $0\,<\,\vp(1)\,<\, \pi$,
if   $0\, < \,\vp(z)\, < \,\pi$,  at~every  $z\in[0,1]$, then \erf{yaez} 
does not have a topologically nontrivial solution.
\end{lemma}

\noindent{\bf Proof:} 
\erf{yaez} has the following property: if 
 $\vp(z)$ is a bounded solution of \erf{yaez}, then $-\vp(z)$,
 $\vp(-z)$ and 
$\vp(z)+2n\pi, \ n\in\mathbb Z$ are also bounded solutions of
\erf{yaez}.   
Using the above properties of \erf{yaez} and the assumptions 
in the lemma, we obtain two conclusions.  For~a bounded solution
of \erf{yaez}, 

\begin{enumerate}
\item
%\textls[-20]
{if $n\pi \slt \vp(1) \slt (n+1)\pi$ for~$n\in \mathbb Z$,
then 
$n\pi \slt \vp(z)\slt (n+1)\pi$
at every $z\in[0,1]$;} 
\item
if $n\pi \slt \vp(-1) \slt (n+1)\pi$ for~$n\in \mathbb Z$,
then 
$n\pi \slt \vp(z)\slt (n+1)\pi$
at every $z\in[-1,0]$.
\end{enumerate}
To prove the lemma, assume to the contrary that there is a 
topologically nontrivial solution $\ti\vp$ with charge $C(\ti\vp) = Q\neq 0$.
From Lemma \ref{lem:le10}, we know that $\ti\vp(-1) \neq m\pi$
for any $m\in \mathbb Z$ and $\ti \vp(1) \neq n\pi$ for any
$n\in \mathbb Z$. Therefore, assume that, for some $m, n \in 
\mathbb Z$, 
\beas
m\pi \slt \ti\vp(-1) \slt (m+1)\pi, \qquad 
n\pi \slt \ti\vp(1) \slt (n+1)\pi
\eeas
But, from the two conclusions  above, we have
\beas
m\pi \slt \ti\vp(0) \slt (m+1)\pi, \qquad n\pi \slt \ti\vp(0) \slt 
(n+1)\pi
\eeas
which leads us to conclude that $m=n$.   From~Lemma \ref{lem:l1}
we conclude that, if $m$ is even, 
\beas
\lt z {-\infty} \ti\vp(z) \aea \lt z {\infty} \ti\vp(z) \deq m\pi
\eeas
and, if $m$ is odd,
\beas
\lt z {-\infty} \ti\vp(z) \aea \lt z {\infty} \ti\vp(z) \deq (m+1)\pi
\eeas
In either case, $\ti\vp$ is a topologically trivial solution, contradicting
our assumption that $\ti\vp$ is topologically nontrivial.   
\qed %\end{proof}
  
\noindent{\bf Proof of Theorem \ref{thm:t1}:}
From Lemmas \ref{lem:le10} and
 \ref{lem:le11} it follows that, to prove the theorem, it is sufficient
 to show that, if a bounded solution $\vp$ satisfies $0 \slt \vp(1) \slt \pi$,
 then $0 \slt \vp(z) \slt \pi$ for all $z\in [0,1]$ for~$\a< 2$.
We prove the above statement in 
hypergeometric coordinates.  Specifically, we will
show (in hypergeometric coordinate $x$)
 that, if 
 $\vp$ is a bounded \mbox{solution of}
 \bea
 x\,(x-1) \vp'' + (2x-1) \,\vp' + \a\, \sin(\vp) \aea 0, \qquad \a < 2
 \label{yx}
 \eea
 and $0 < \vp(0) < \pi$, then 
 $0<\vp(x) < \pi$ at~all $x\in[0,\f12]$.   
We set  $\g = \vp(0)$.  We prove 
 the following {claims.} %MDPI: We formatted as the end of proof, please check and confirm.
%\qed %\end{proof}

\begin{claim}
{%\it 
 If $0 < \vp(0) < \pi$ and 
 $\vp''(x) \leq 0$ at all  $x\in(0,\f12]$, then 
 \beas
 0 \ <\  \vp(x)  <  \pi, \qquad \mbox{ at every } x\in(0,\f12]
 \eeas
 }
\end{claim}
  
\noindent{\bf Proof:} 
Assume
$\vp''(x) \leq 0$ throughout $(0,\f12]$.
 Then, using Lemma \ref{lem:l4}, we have at \mbox{$x \in (0,\f12]$}
 \beas
 \vp'(x) \aea \lt \e  0 \lrs{\int_\e^x \vp''(y) dy + \vp'(\e)}
 \leq \a \sin(\g)
 \eeas
 Therefore, at~every $x\in (0,\f12]$, recalling that $\a < 2$, we have 
 \beasm
 \vp(x) \aea \g + \izx {\vp'(y)} \ \leq\  \g + x \, \a\, 
 \sin(\g) \ \leq\  \g + \dsf \a 2 \sin(\g) < \g + \sin(\g)
 \label{iax}
 \eeasm
 If $\vp(x) \deq \pi$ at some $x\in \ioh$, then, from 
 (\ref{iax}), we have 
 \bea
 \pi < \g + \sin(\g)  \implies \dsf{\sin(\pi-\g)}{\pi-\g} > 1
 \label{mk}
 \eea
 which is impossible.  The~contradiction proves  
 that $\vp(x) \slt \pi$ at all $x\in\left(\ltd0,\f12\right]\rtd.$
From Lemma \ref{lem:le12} we conclude that $ \vp(x) > 0$ 
at all $x\in\left(\ltd0,\f12\right]\rtd.$ 
\qed %\end{proof} 
%\ed

Therefore, in~the remainder of the proof
we will assume that $\vp''(x) >0$ at some 
$x\in \ioh$.
Define
\bea
x\so \aea \inf_{x\in \ioh} \lrc{x\, |\, \vp''(x) > 0}
\label{xtil}
\eea 
It is possible that $x\so \deq 0 \not\in \ioh$.  

We prove the theorem by contradiction. If~possible, let 
 $x\sr$, defined below,  
exist.
\bea
x\sr \aea \inf_{x\in \ioh} \lrc{x \, | \, \vp(x) \deq \pi}
\label{xstr}
\eea
If $x\sr$ exists then we note that $x\sr \in \ioh$. 
$x\so$ and $x\sr$ are illustrated in  {Figure} \ref{thesalient}. 

\begin{figure}[H]
\includegraphics[height=2in]{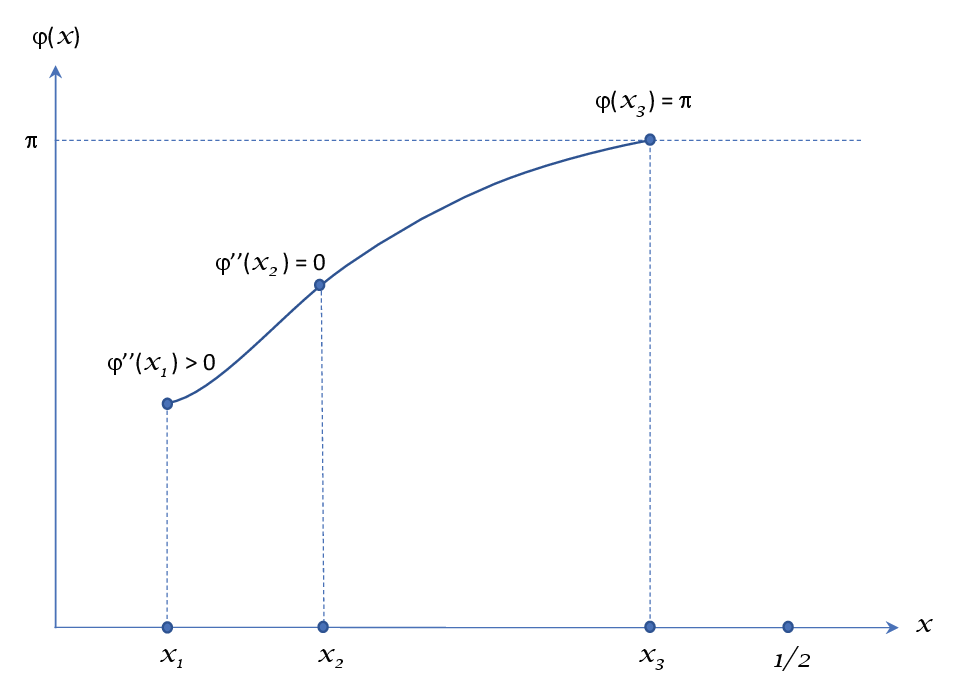}
\caption{{Illustration} %MDPI: Please confirm whether an explanation of the dotted and solid lines needs to be added to the figure caption.
 of the salient values $x\so, x\sw$ and $x\sr$ used in the proof. 
 {The solid line represents a solution of Eq.} (\ref{yx}).  
 {The dotted line represents $\vp(x) \equiv \pi$.}}\label{thesalient}. 
\end{figure}

\begin{claim}\label{claim2}
At every $x\in \lrs{0,x\so}$ we have 
{$0<\vp(x )<\pi$.}
\end{claim}

\noindent{\bf Proof:} 
 If $x\so = 0$,
then, from  $\vp(x\so) \deq \g  < \pi$, 
the claim holds.
On the other hand, if~$0 < x\so \leq \f12$, then, from the
definition of $x\so$,  in~(\ref{xtil}),  $\vp''(x)\leq 0$
at every $x\in (0,x\so]$.  Therefore, recalling that $\a \slt 2$,
 at every $x\in   \lrs{0,x\so}$,
\beass
\vp(x ) \aea \g + 
\int_0^{x }{\lt \e 0 \lrs{\int_\e^y \vp''(z)\, dz + \vp'(\e)}}\, dy \leq \g + x \, \a \, \sin(\g) < \g + \sin(\g) < \pi
\eeass
The last inequality holds since $0<\g<\pi$ and $\sin(\g)/(\pi-\g) < 1$ 
at $\g\neq \pi$.  From \mbox{Lemma \ref{lem:le12}} we conclude that $\vp(x) > 0$
at every $x\in \lrs{0,x\so}$. 
\qed %\end{proof}

\begin{claim}\label{claim3}
{\it At every $x\in (x\so, x\sr]$, $\vp'(x) > 0$.}
\end{claim}

\noindent{\bf Proof:} 
Rewriting \erf{yx} and using Lemma \ref{lem:l4},
we have, for $x\in (x\so,x\sr]$,
\bea
\vp'(x) \aea \dsf 1 {x \, (1-x)} \izx {\a\,\sin(\vp(y))}
\label{og}
\eea
From Claim \ref{claim2} above, Lemma \ref{lem:le12}, and~the definition of $x\sr$, we know that,
at every   $x\in[0, x\sr]$,   $0 < \vp(x) \leq \pi$. From~\erf{og}
we conclude that $\vp'(x) > 0$ at all $x\in [0, x\sr]$.  
\qed %\end{proof}

\begin{claim}\label{claim4}
{\it At every $x\in (x\so, x\sr]$, $\vp''(x) > 0$.}
\end{claim}

\noindent{\bf Proof:} 
Assume to the contrary that $\vp''(x\sw) = 0$ at some $x\sw \in 
(x\so, x\sr]$.   See  {Figure} \ref{thesalient}. 
%\ref{thealphagw}.  
We will assume that $x\sw$ is the 
smallest value of $x\in (x\so, x\sr]$ at which $\vp''$ vanishes.
Therefore $\vp''(x) > 0$ at all $x\in (x\so,x\sw)$, and, since 
$\vp''(x\sw) = 0$, we must have $\vp'''(x\sw) \leq 0$.   
However, we will show that $\vp'''(x\sw) > 0$. The~
contradiction will prove the~claim.

Differentiating \erf{yx} we obtain
\bea
\vp'''(x) \aea \dsf 1{x(1-x)}
\lrs{(4x-2)\vp'' + (2 + \a\, \cos(\vp)) \vp'}
\label{tr}
\eea
We see from \erf{tr} that $\vp'''$ is well defined in 
$(x\so, x\sr]$.   For~$\a < 2$, we have $2 + \a \cos(\vp) > 0$.
Since $\vp''(x\sw) = 0$,   noting that 
$x\sw \, (1-x\sw) > 0$ and~using Claim \ref{claim3}, we have
\beas
\vp'''(x\sw) \aea \dsf {(2+\a\, \cos(\vp(x\sw)))\, \vp'(x\sw)}{x\sw\, (1-x\sw)} \ > \ 0
\eeas
The contradiction shows that $x\sw$ does not exist. Therefore 
$\vp''(x) >0$ at all $x\in (x\so, x\sr]$.  
\qed %\end{proof}

%\ed
From Claim \ref{claim4} and \erf{yx} we have at every $x\in (x\so, x\sr]$,
\beas
\vp''(x) \aea \dsf{1}{x(1-x)}\lrs{(2x-1)\, \vp'(x) + \a\, \sin(\vp(x))} > 0 
\eeas
In particular, at~$x\sr$, we have 
\beas
0 \deq \a\, \sin(\vp(x\sr)) > (1- 2x) \, \vp'(x\sr) \geq 0
\eeas
The contradiction shows that 
$x\sr$ does not exist and $0 < \vp(x) < \pi$
at every \mbox{$x\in [0,\f12]$.}  From~Lemma \ref{lem:le11} it follows
that every bounded solution
of \erf{yx} is {topologically  {trivial.}} %MDPI: We removed extra \quad. please confirm.
\qed %\end{proof}

\section{{Small}~Solitons}\label{app:e}
\begin{lemma}\label{lem:l9} 
For $0\leq \gamma \leq \pi$, let $\varphi(x;\gamma)$ be a
continuous and bounded solution of
\begin{eqnarray}
x(x-1) \varphi_{xx} + (2x-1) \varphi_x + \alpha \sin (\varphi) = 0
\label{lemma241eqn}
\end{eqnarray}
in the interval $\left[-\frac 12,\frac 12\right]$, satisfying the boundary conditions
\begin{eqnarray}
\varphi(0;\gamma) = \gamma, \qquad \varphi'(0;\gamma) = \alpha \sin(\gamma).
\label{lemma241boundarycondns}
\end{eqnarray}
Then, at~any $\tilde x\in \left[0,\frac 12\right]$, $\varphi(\tilde x; \gamma)$ 
and $\varphi'(\tilde x; \g)$ are  
continuous functions of $\gamma$.
\end{lemma}

\noindent{\bf Proof:} 
Observe that, if we choose an $\epsilon$ such that $0<\epsilon\leq \tilde x\leq \frac 12$,
then the differential \mbox{Eq.~ (\ref{lemma241eqn})} has no singularities in the interval
[$\epsilon,\tilde x$].  Therefore, $\varphi(\tilde x;\gamma)$ depends continuously on the initial conditions
$\varphi(\epsilon;\gamma)$ and $\varphi'(\epsilon;\gamma)$.  To~prove the lemma it is therefore sufficient to
show that $\varphi(\epsilon;\gamma)$ and $\varphi'(\epsilon;\gamma)$ are continuous functions of $\gamma$ at
some \mbox{$0<\epsilon\leq \tilde x\leq \frac 12$.}

If $\varphi(x;\gamma)$ is a solution of \eqnref{lemma241eqn}, satisfying the boundary condition
(\ref{lemma241boundarycondns}), then $\varphi(x;\gamma)$ is also a solution of the integral equation
\begin{eqnarray}
\varphi(x;\gamma) = V(x;\varphi(x;\gamma)) + \gamma \cdot \varphi_1(x)
\label{integraleqn2}
\end{eqnarray}
where $V(x;\varphi(x;\gamma))$ is defined in (\ref{defv}) and $\varphi_1(x)$ is a solution of
the homogeneous hypergeometric \eqnref{hom-eq-x}.

To show that $\varphi(x;\gamma)$ is a continuous function of $\gamma$ at $x$, we
show that, for any given $\rho > 0$, there is a $\delta > 0$ such that,
if $|\gamma - \gamma'| < \delta$, then $|\varphi(x;\gamma) - \varphi(x;\gamma')| < \rho$.
Using \eqnref{integraleqn2} we have
\beasm
|\varphi(x;\gamma) - \varphi(x;\gamma')| & \leq & |V(x;\varphi(x;\gamma)) - V(x;\varphi(x;\gamma'))| +
|\varphi_1(x)|\cdot |\gamma-\gamma'|\qquad 
\label{upboundvarphi}
\eeasm
In order to derive an upper bound on the right-hand side of the above inequality, 
we note, using Lemma \ref{lem:l6},  that 
\beasm
V(x;\varphi) = h(x) \int_0^x \varphi_1(y) q(y) dy + \varphi_1(x) \int_0^x \frac 1 y
\int_0^y \varphi_1(u) q(u)\, du\, dy - \varphi_1(x) \int_0^x h(y) q(y) dy 
\label{simplifiedV}
\eeasm
$h(x)$, described in the discussion
following \eqnref{phi-2a}, $xW(x)$ and $\varphi_1(x)$ are analytic in
[$-\epsilon,\epsilon$] for sufficiently small $\e\in\lrr{0,\f12}$; see 
Lemma \ref{lem:l5}.  
Therefore, we
can find constants $k_{\varphi_1}(\epsilon), k_h(\epsilon), k_w(\epsilon)$ such that, for all $0\leq x\leq \epsilon$,
$|\varphi_1(x)|< k_{\varphi_1}(\epsilon), \ |h(x)| < k_h(\epsilon)$ and $|x(x-1)W(x)|^{-1} < k_w(\epsilon)$.

We bound the
variation of $f(\varphi(x;\gamma)) = \alpha(\varphi(x;\gamma) - \sin(\varphi(x;\gamma))$ with $\gamma$.
Observe that $|\sin(\theta) -\sin(\theta')| \leq |\cos(\theta^*)||\theta - \theta'| \leq | \theta - \theta'|$, where $\theta^*$ is some
number between $\theta$ and $\theta'$. Therefore, for~any $x\in[-\epsilon,\epsilon]$,
\begin{eqnarray}
|f(\varphi(x;\gamma)) - f(\varphi(x;\gamma'))| & \leq & 2 \alpha |\varphi(x;\gamma) - \varphi(x;\gamma')|
\leq  2 \alpha M(\epsilon,\g,\g')
\end{eqnarray}
where
\begin{eqnarray}
M(\epsilon,\g,\g') := \max_{x\in [-\epsilon,\epsilon]} \left\{ |\varphi(x;\gamma) - \varphi(x;\gamma') |\right\}\label{megg}
\end{eqnarray}

Using \eqnref{simplifiedV} and the~bounds on $|h(x)|, |\varphi_1(x)|$ and $|x (x-1) W(x)|^{-1}$ in $[-\epsilon,\epsilon]$,
we obtain
\begin{eqnarray}
|V(x,\varphi(x;\gamma)) - V(x,\varphi(x;\gamma'))| \leq
C(\epsilon) \cdot \epsilon \cdot M(\epsilon,\g,\g')
\label{upboundV}
\end{eqnarray}
where $C(\epsilon) = 4 k_h(\epsilon) k_{\varphi_1}(\epsilon) k_w(\epsilon) \alpha + 2 \cdot (k_{\varphi_1} (\epsilon))^2 k_w(\epsilon)
\alpha$.

Using (\ref{upboundvarphi})  {and} %MDPI: We revised comma to and, please check if the meaning is retained.
 (\ref{upboundV})  we have
\begin{eqnarray*}
|\varphi(x;\gamma) - \varphi(x;\gamma')| \leq C(\epsilon) \cdot \epsilon \cdot M(\epsilon,\g,\g') + k_{\varphi_1}(\e) |\gamma - \gamma'|
\end{eqnarray*}
The above inequality holds at every $x\in[-\epsilon,\epsilon]$ for~a 
sufficiently small $\e$.  Therefore,
\begin{eqnarray*}
M(\epsilon,\g,\g') \leq C(\epsilon) \cdot \epsilon \cdot M(\epsilon,\g,\g') + k_{\varphi_1}(\e) |\gamma - \gamma'|
\end{eqnarray*}
Choosing a sufficiently small $\e\so\in\lrr{0,\f12}$ we have
\bea
M(\e\so,\g,\g') \leq \dsf{k\dn{\vp\so}(\e\so) \norm{\g-\g'}}{1- \e\so\, C(\e\so)}
\label{ime}
\eea
Observe that, for every $x\in [0,\e\so]$, 
\begin{eqnarray*}
M(x,\g,\g') \leq M(\epsilon_{_1},\g,\g') \leq {k_{\varphi_1}(\epsilon_{_1}) \cdot |\gamma - \gamma'| \over 1- \epsilon_{_1} \cdot C(\epsilon_{_1})}; 
\end{eqnarray*}
Given any $\rho > 0$, we choose $\delta = \displaystyle {\rho (1 - \epsilon_{_1}\cdot C(\epsilon_{_1}) ) \over k_{\varphi_1}(\e\so)}$.
Then, at any $x\in[0,\epsilon_{_1}]$,
\begin{eqnarray*}
|\varphi(x; \gamma) - \varphi(x; \gamma')| \leq M(x) \leq M(\epsilon_{_1}) < \rho, \qquad \mbox{ if \ } |\gamma-\gamma'| < \delta
\end{eqnarray*}
showing that $\varphi(x;\gamma)$ is a continuous function of $\gamma$ at every $x\in[0,\epsilon_{_1}]$.

Using the integral Eq.~ (\ref{integraleqn2}) and the~simplified form of $V$ shown in \eqnref{simplifiedV} we can
similarly show that, at every $x\in[0,\epsilon_{_2}]$, for~some sufficiently small $\epsilon_{_2}$,
\begin{eqnarray*}
|\varphi'(x;\gamma) - \varphi'(x;\gamma')| \leq K(\epsilon_{_2}) \cdot 
\lrs{2 \a M(\e\sw,\g,\g')}
\end{eqnarray*}
%\textls[-15]
{where 
\beas
K(\epsilon_{_2}) := k_w \left[k'_h k_{\varphi_1} \epsilon_{_2} + k_h k_{\varphi_1}+k'_{\varphi_1}k_{\varphi_1} + (k_{\varphi_1})^2
+ k'_{\varphi_1}k_h\epsilon_{_2} + k_{\varphi_1}k_h\right]
\eeas
and 
%$K(\epsilon_{_2})$ depends on the bounds on $h, h', \varphi_1, \varphi_1', \{x(x-1)W(x)\}^{-1}$ in the}
%interval $[0,\epsilon_{_2}]$, as shown in the  {Footnote}  ({$K(\epsilon_{_2}) := k_w \left[k'_h k_{\varphi_1} \epsilon_{_2} + k_h k_{\varphi_1}+k'_{\varphi_1}k_{\varphi_1} + (k_{\varphi_1})^2
%+ k'_{\varphi_1}k_h\epsilon_{_2} + k_{\varphi_1}k_h\right]$, where
and the bounds on $h, h', \varphi_1, \varphi_1', \{x(x-1)W(x)\}^{-1}$ in the 
interval $[0,\epsilon_{_2}]$ are
$|h(x)| < k_h,\  |h'(x)| < k'_h,\  |\varphi_1(x)| <  k_{\varphi_1},\
|\varphi_1'(x)| < k'_{\varphi_1},\  [x(x-1)W(x)]^{-1} < k_w.\ $ 
%for $x\in[0,\epsilon_{_2}]$.
} 
Since $M(\e\sw,\gamma,\gamma') \rightarrow 0$ as $\gamma\rightarrow \gamma'$,  for~a given
$\rho > 0$, we can make $M(\e\sw,\gamma,\gamma') < \dsf{\rho}{ 2\a\,K(\epsilon_{_2})}$ by choosing $|\gamma-\gamma'| < \delta$ for~some $\delta > 0$.  Then,
\begin{eqnarray*}
|\varphi'(x;\gamma) - \varphi'(x;\gamma')|  
< {K(\epsilon_{_2}) \cdot M(\e\sw,\gamma,\gamma') < \rho}, \qquad \mbox{ if } |\gamma-\gamma'| < \delta
\end{eqnarray*}
which shows that $\varphi'(x;\gamma)$ is also a continuous function of $\gamma$ at any $x\in [0,\epsilon]$. 
\qed %\end{proof}
%%%%%%%%%%%%%%%%%%%%%%%%%%%%%%%%%%%%%%

{Lemmas \ref{lem:l10} and \ref{lem:l11} use the continuity property
established in Lemma \ref{lem:l9} to
show, respectively, that there exists 
a $\g$ for which the solution overshoots $\vp=\pi$ in the interval $\lrs{0,\f12}$
and a $\g$ for which the solution reaches $\vp=\pi$ at $x\deq \f12$.  A~
similar argument has been used to study dynamics of Yang--Mills fields
in asymptotically hyperbolic spacetime~\cite{bima}.}
\begin{lemma}\label{lem:l10}  
Let $\varphi(x;\gamma)$ be a continuous bounded solution of
\begin{eqnarray}
x(x-1)\varphi_{xx} + (2x-1) \varphi_x + \alpha \sin(\varphi) = 0, \qquad \alpha > 2
\label{lemma242eqn}
\end{eqnarray}
in the interval $\lrs{0,\f12}$, satisfying the initial conditions
\begin{eqnarray}
0\slt \varphi(0;\gamma) = \gamma\slt \pi; \qquad \varphi_x(0;\gamma) = \alpha \sin(\gamma).
\label{lemma242boundaryconditions}
\end{eqnarray}
Then there exists a $\gamma \in (0,\pi)$ for which
\begin{eqnarray*}
M(\gamma) := \max_{0\leq x\leq \frac 12} \varphi(x;\gamma) \geq \pi
\end{eqnarray*}
\end{lemma}
%%%%%%%%%%%%%%%%%%%%%%%%%%%%%%%%%%%%%%

\noindent{\bf Proof:} 
For  $\alpha > 2$, there exists a $\bar\gamma\in\left(\frac\pi 2, \pi\right)$ such that, for
every $\gamma^*\in(\bar\gamma,\pi)$,
\begin{eqnarray}
\alpha \cos(\gamma^*) + 2 < 0
\label{defnofbargamma}
\end{eqnarray}
We will show that, for every $\gamma^* \in (\bar\gamma,\pi)$, $M(\gamma^*) \geq \pi$.

Pick a $\gamma^*\in (\bar\gamma,\pi)$ and~assume to the contrary that $M(\gamma^*) < \pi$. First,
we claim that, if $M(\gamma^*) < \pi$, then, at every $x\in\ozch$,
({\it i}) 
$\varphi_x(x) > 0$, and~({\it ii})
$\varphi_{xx}(x) < 0$. 

We begin by noting that, if  $M(\gamma^*) < \pi$, then, from Lemma \ref{lem:le12}, 
$\vp(x;\g\str) > 0$ at every $x\in \lrs{0,\f12}$.
To prove the first claim,
assume to the contrary that $\varphi_x(x') = 0$ at~some $x'\in(\left. 0,\frac 12\right]$. Further,
let $x'$ be the smallest such $x'$.  Since $0\slt \vp(x;\g\str) \slt \pi$ for $0 \slt x \slt x'$,
\begin{eqnarray*}
\varphi_x(x') = \alpha \sin(\gamma) + \frac 1{x'(1-x')} \int_0^{x'} \alpha \sin(\varphi(x)) \, dx
\geq \alpha \sin(\gamma) > 0
\end{eqnarray*} 
The contradiction proves the claim.
To prove the second claim observe that
\begin{eqnarray}
\varphi_{xx}(0) & = & \lim_{x\rightarrow 0} \ \left({1\over 1-x}\right)
\left\{ {(2x-1) \varphi_x(x) + \alpha\sin(\varphi(x)) \over x}\right\}
\label{psd}\\[0.5\baselineskip]
& = & \lim_{x\rightarrow 0} \
 {(2x-1) \varphi_{xx}(x) + (2+\alpha\cos(\varphi(x))) \varphi_x(x)}\nonumber
\end{eqnarray}
The second equality is obtained using {\it 'l Hospital's} rule.  Taking the limit and rearranging, 
we get
\begin{eqnarray}
\varphi_{xx}(0) =  \frac12\cdot (2 + \alpha\cos(\gamma^*)) \alpha\sin(\gamma^*) < 0 \label{fsd}
\end{eqnarray}
The last inequality follows from (\ref{defnofbargamma}).  We see from 
(\ref{fsd}) that the limit in (\ref{psd}) exists, and~therefore there exists a
$\lambda  > 0$ such that, at
every $\ti x\in (-\lambda,\lambda)$,
$\vp_{xx}(\ti x) < 0$.   Therefore $\vp_{xx}(\ti x) \neq 0$ for~$\ti x \in (-\lambda,\lambda)$.

Again, arguing by contradiction, if~possible, let $\varphi_{xx}$ change sign in a small
neighborhood of $\bar x$, with~$\varphi_{xx}(\bar x) = 0$ and $\varphi_{xxx}(\bar x) > 0$, 
at some
$\bar x\in \ozch$.   As~we noted above, $\vp_{xx}(x)$ is nonzero in the interval $(-\lambda,\lambda)$, so $\bar x \neq 0$. Differentiating \eqnref{lemma242eqn}, we obtain
\begin{eqnarray}
\varphi_{xxx}(\bar x) = {1 \over \bar x(1-\bar x)} \cdot 
\left\{(2 + \alpha \cos(\varphi(\bar x))) \varphi_x(\bar x)\right\} < 0
\label{ptd}
\end{eqnarray}
Note that $\vp_{xxx}(x)$ is well defined at $\bar x\in \ozch$. 
The inequality in (\ref{ptd}) follows from (\ref{defnofbargamma}), and, by observing that $\varphi_x(x) > 0$ for $x\in\left[0,\frac 12\right]$, 
$\varphi(\bar x) \geq \varphi(0)=\gamma^*>0$.  The~contradiction proves that, if $M(\gamma^*) < \pi$, then $\varphi_{xx}(x) < 0$ for $x\in\ozch$.   

But,
\begin{eqnarray*}
\varphi_{xx}\left(\frac 12 \right) = 4 \left\{ \alpha \sin\left(\varphi\left(\frac 12\right)\right) 
\right\} > 0,
\end{eqnarray*}
since $0 <  \varphi\left(\frac 1 2 \right) < M(\gamma^*) < \pi$. 
The contradiction proves that $M(\gamma^*) \geq \pi$. 
\qed %\end{proof}

%%%%%%%%%%%%%%%%%%%%%%%%%%%%%%%%%%%%%%
\begin{lemma}\label{lem:l11}
Let $\varphi(x;\gamma)$ be a continuous bounded solution of
\begin{eqnarray}
x(x-1)\varphi_{xx} + (2x-1) \varphi_x + \alpha \sin(\varphi) = 0, \qquad \alpha > 2
\label{lemma243eqn}
\end{eqnarray}
in the interval $\lrs{0,\f12}$, satisfying the initial conditions
\begin{eqnarray}
\varphi(0;\gamma) = \gamma; \qquad \varphi_x(0;\gamma) = \alpha \sin(\gamma).
\label{lemma243boundaryconditions}
\end{eqnarray}
Then there exists a $\gamma^* \in (0,\pi)$, for~which $\varphi\left(\frac 1 2; \gamma^*\right) = \pi$.
\end{lemma}
%%%%%%%%%%%%%%%%%%%%%%%%%%%%%%%%%%%%%%

\noindent{\bf Proof:}  
From Lemma \ref{lem:l10} we know that there exists at least one $\gamma\in(0,\pi)$ for which $M(\gamma) \geq \pi$.
Therefore, one can define $\gamma^*$ as follows.
\begin{eqnarray*}
 \gamma^* := \inf_{0 < \gamma < \pi} \ \{ \gamma \, |\,  M(\gamma) \geq \pi\}.
\end{eqnarray*}
If $M(\gamma^*) > \pi$, then, from Lemma \ref{lem:l9}, we can find a $\de$
such that $M(\g) > \pi$ for all \linebreak  $\g \in (\g\str-\de, \g\str+\de)$,  
contradicting the definition of $\g\str$.  
Therefore, we conclude that $M(\gamma^*)=\pi$.

We note that $\varphi(0;\gamma^*) = \gamma^* < \pi$.
Let $x^*$ be the leftmost point at which $\varphi(x^*;\gamma^*) = \pi$.  If~$x^*\in \lrr{0,\f12}$,
then $\varphi_x(x^*;\gamma^*)=0$. From~\eqnref{lemma243eqn} $\varphi_{xx}(x^*;\gamma^*) = 0$.
Repeated differentiation of \eqnref{lemma243eqn} shows that all the higher derivatives of
$\varphi(x;\gamma^*)$ must also vanish at $x^*$.  But, by uniqueness of the solution at $x^*$,
$\varphi(x;\gamma^*) \equiv \pi$ in a small neighborhood of $x^*$, which contradicts the assumption
that $x^*$ is the leftmost point at which $\varphi(x;\gamma^*) = \pi$. Therefore, we conclude that
$\varphi(x;\gamma^*) < \pi$ for~$0\leq x < \frac 12,$ and $\varphi\left(\frac 12;\gamma^*\right) = \pi$.
\qed %\end{proof}

%%%%%%%%%%%%%%%%%%%%%%%%%%%%%%%%%%%%%%%
\begin{lemma}\label{lem:l12}
For $\alpha> 2$, there exists a
solution of
\begin{eqnarray}
(z^2-1) \varphi_{zz} + 2z \varphi_z + \alpha \sin (\varphi) = 0
\label{theorem244eqn}
\end{eqnarray}
satisfying the boundary conditions
\begin{eqnarray}
\lim_{z\rightarrow -\infty} \varphi(z) = 2\pi; \qquad \lim_{z\rightarrow \infty} \varphi(z) = 0
\label{theorem244boundarycondns}
\end{eqnarray}
\end{lemma}
%%%%%%%%%%%%%%%%%%%%%%%%%%%%%%%%%%%%%%%

\noindent{\bf Proof:} 
 From Lemmas \ref{lem:l11} and \ref{lem:l1}   
we know that \eqnref{theorem244eqn} has a
solution $\varphi(z)$ in~$[0,\infty)$ with the following properties:
\begin{eqnarray*}
\varphi(0) = \pi; \qquad 0 < \varphi(1) < \pi; \qquad \varphi(\infty) = 0
\end{eqnarray*}
The symmetries of \eqnref{theorem244eqn}, $\varphi \ra \vp + 2n\pi$, $n\in \mathbb Z$, $\varphi \rightarrow -\varphi$ and $z\rightarrow -z$,
imply the existence of the solution claimed in the theorem. \qed %\end{proof}

%\begin{adjustwidth}{-\extralength}{0cm}
%\reftitle{References}

%%%%%%%%%%%%%%%%%%%%%%%%%%%%%%%%%%%%%%%%%%
%\PublishersNote{}
%\end{adjustwidth}
\end{document}